\documentclass[11pt,letterpaper]{article}

\usepackage[margin=1in]{geometry}
\usepackage[T1]{fontenc}
\usepackage[utf8]{inputenc}
\usepackage{lmodern}

\usepackage{amsmath,amssymb,amsthm,mathtools}
\usepackage{bm}

\usepackage{xcolor}
\usepackage{booktabs}
\usepackage{tabularx}
\usepackage{graphicx}
\usepackage{enumitem}
\usepackage{tikz}
\usetikzlibrary{arrows.meta,positioning,calc}

\usepackage{pgfplots}
\usepgfplotslibrary{groupplots}
\pgfplotsset{compat=1.18}
\usepackage{natbib}
\usepackage{hyperref}
\usepackage[nameinlink,capitalize]{cleveref}
\usepackage{bibunits}                 % separate reference list for the supplement
\DeclareMathOperator{\Range}{Ran}     % matches the main text's \operatorname{Ran}
\newcommand{\cgf}{\Lambda}            % cumulant generating function
\newcommand{\lamcrit}{\lambda_{\mathrm{crit}}}
\definecolor{darkblue}{rgb}{0,0,0.55}

\hypersetup{
  colorlinks=true,
  linkcolor=darkblue,
  citecolor=blue,
  urlcolor=blue
}

\usepackage{authblk}

\usepackage{aliascnt}

\newtheorem{theorem}{Theorem}[section]

\newaliascnt{lemma}{theorem}
\newtheorem{lemma}[lemma]{Lemma}
\aliascntresetthe{lemma}

\newaliascnt{proposition}{theorem}
\newtheorem{proposition}[proposition]{Proposition}
\aliascntresetthe{proposition}

\newaliascnt{corollary}{theorem}
\newtheorem{corollary}[corollary]{Corollary}
\aliascntresetthe{corollary}

\theoremstyle{definition}

\newaliascnt{definition}{theorem}
\newtheorem{definition}[definition]{Definition}
\aliascntresetthe{definition}

\newaliascnt{assumption}{theorem}
\newtheorem{assumption}[assumption]{Assumption}
\aliascntresetthe{assumption}

\newaliascnt{condition}{theorem}

\aliascntresetthe{condition}

\newaliascnt{example}{theorem}
\newtheorem{example}[example]{Example}
\aliascntresetthe{example}

\theoremstyle{remark}

\newaliascnt{remark}{theorem}
\newtheorem{remark}[remark]{Remark}
\aliascntresetthe{remark}

\crefname{theorem}{Theorem}{Theorems}
\crefname{lemma}{Lemma}{Lemmas}
\crefname{proposition}{Proposition}{Propositions}
\crefname{corollary}{Corollary}{Corollaries}
\crefname{definition}{Definition}{Definitions}
\crefname{assumption}{Assumption}{Assumptions}
\crefname{condition}{Condition}{Conditions}
\crefname{example}{Example}{Examples}
\crefname{remark}{Remark}{Remarks}
\crefname{equation}{Equation}{Equations}
\crefname{figure}{Figure}{Figures}
\crefname{table}{Table}{Tables}
\crefname{section}{Section}{Sections}
\crefname{subsection}{Section}{Sections}
\crefname{subsubsection}{Section}{Sections}
\crefname{appendix}{Appendix}{Appendices}

\newcommand{\R}{\mathbb{R}}

\newcommand{\X}{\mathcal{X}}
\newcommand{\Y}{\mathcal{Y}}
\newcommand{\Th}{\Theta}

\newcommand{\E}{\mathbb{E}}
\newcommand{\Prob}{\mathbb{P}}

\DeclareMathOperator{\Var}{Var}
\DeclareMathOperator{\Cov}{Cov}

\newcommand{\Normal}{\mathcal{N}}

\DeclareMathOperator{\Lap}{Lap}
\DeclareMathOperator{\Logistic}{Logistic}
\DeclareMathOperator{\Unif}{Unif}

\DeclareMathOperator*{\argmin}{arg\,min}
\DeclareMathOperator{\supp}{supp}

\newcommand{\ind}{\mathbf{1}}

\newcommand{\varphiN}{\varphi}

\newcommand{\virt}{J}
\newcommand{\virtt}{J_\theta}
\newcommand{\Psival}{\Psi}

\newcommand{\qstar}{q^{*}}

\newcommand{\lamstar}{\lambda^{*}}

\newcommand{\xlo}{\underline{x}}
\newcommand{\xhi}{\bar{x}}

\title{\textbf{Public Good Provision under Locally Private Signals}}

\author{
Behrooz Moosavi Ramezanzadeh
\thanks{Department of Economics, University of Pittsburgh.
Email: \texttt{bem159@pitt.edu} (Corresponding Author)}
\quad
Jordan Awan%
\thanks{Department of Statistics, University of Pittsburgh.
Email: \texttt{jaa557@pitt.edu}}
}

\date{\today}

\begin{document}

\maketitle

\begin{abstract}
We study public-good provision when a planner observes agents' preferences only
through a fixed local-privacy channel that randomizes each report before it
reaches the planner. We characterize the optimal reduced-form allocation: the
project is implemented when an aggregate posterior score is positive, where each
agent's score combines the posterior expected valuation and posterior virtual
value. Privacy enters through these posterior objects, muting the responsiveness
of provision to private preferences and, under weak monotone likelihood ratios,
potentially generating pooling. We then distinguish the optimal reduced-form
allocation from its implementation through signal-measurable transfers: the
required transfers solve a Fredholm integral equation whose solution is unique
under completeness when it exists, while existence requires a separate range
condition. Maximum reduced-form revenue exhibits three population regimes: it is
asymptotically linear, of square-root order, or exponentially small according as
the lower endpoint of the valuation distribution is positive, zero, or negative.
Finally, welfare comparisons depend on the privacy calibration. At a common
noise scale, Laplace Blackwell-dominates logistic noise, while under a common
tight \(\mu\)-GDP calibration the ordering reverses for the maximally separated
binary endpoint experiment. Thus the preferred privacy channel depends on the
standard used to hold privacy fixed.
\end{abstract}



\bigskip
\noindent\textbf{Credit authorship contribution statement:}

\medskip

\noindent\textbf{Behrooz Moosavi:}
Review \& editing, Writing original draft, Methodology, Investigation,
Conceptualization.

\medskip

\noindent\textbf{Jordan Awan:}
Review \& editing, Writing original draft, Methodology, Investigation,
Conceptualization.

\medskip

\noindent\textbf{Conflict of Interest:}
The authors declare that they have no conflict of interest.

\medskip

\noindent\textbf{Funding:}
This work was supported in part by NSF award number SES-2610910.

\medskip

\newpage

% ============================================================================
\section{Introduction}
\label{sec:introduction}
Classical mechanism design relies on the sensitivity of outcomes to individual
reports. Vickrey--Clarke--Groves payments
\citep{vickrey1961counterspeculation,clarke1971multipart,groves1973incentives}
charge each agent the externality it imposes on the others; the Myerson envelope
theorem \citep{myerson1981optimal} recovers information rents from the interim
allocation rule, since an agent's equilibrium utility has derivative equal to its
allocation probability and rents are the integral of that allocation over types.
Differential privacy demands the opposite: the distribution of a released output
must change only slightly when any single input changes
\citep{dwork2006calibrating,DworkRoth2014}. Local differential privacy enforces
this by randomizing each input before it reaches the collector.

\medskip

This architecture is already deployed at the scale of the world's largest data
collectors. Apple has used local differential privacy in iOS to gather usage
statistics from hundreds of millions of devices, with randomization performed on
the user's device before any data reaches its servers
\citep{apple2017learning,apple2017dp}, and Google's RAPPOR system applied the
same local model in Chrome \citep{erlingsson2014rappor}. In both, the collector
is untrusted and observes only a randomized signal, never the true input. These
systems perform statistical estimation, with no transfers, incentive
constraints, or collective decision at stake; the conflict with mechanism design
arises only when such a channel is placed in front of a planner who must elicit
preferences, make a public decision, and raise revenue from privatized signals
alone. There, neither the standard allocation rule nor the standard payment
identity survives unchanged.

\medskip

This paper characterizes the optimal reduced-form mechanism and the asymptotic
behavior of maximum reduced-form revenue in such an environment. A planner must
decide whether to implement a binary public project---build or not build---when
\(n\) agents hold privately known valuations \(X_1,\ldots,X_n\). The planner does
not observe these valuations. Each agent submits a value to a fixed trusted mediator (privacy channel), which randomizes the submission before transmitting it to the planner;
the agent controls the value it submits but neither observes nor controls the
realized noise, and cannot bypass the protocol. The planner observes only the
privatized signal profile \((Y_1,\ldots,Y_n)\), so both the project decision and
the transfer schedule must be measurable functions of these signals alone. The
mechanism must be Bayesian incentive compatible, so that truthful submission is
optimal for every type; we analyze the resulting truthful equilibrium, along
which the submitted value equals the valuation and the privatized signal is
distributed as
\[
Y_i\sim K(\,\cdot\mid X_i).
\]
It must also be individually rational and meet an ex-ante revenue target.

\medskip

This formulation fits settings in which preference-relevant information is
generated or measured locally and passed through a prescribed privacy device
before reaching an untrusted aggregator: federated rollout decisions based on
privatized quality or demand reports, consortium or standard-adoption decisions
in which members privatize their preferences from the aggregator, and more
generally any binary collective decision mediated by a fixed local randomization
protocol. Within this class we ask three questions. How strong can the privacy
guarantee become before the project can no longer be financed, and how does this
relate to the classical difficulty of public-good provision with many privately
informed agents \citep{mailath1990asymmetric}? What form does the optimal
decision rule take when the planner observes only privatized signals? And, once
channels are calibrated to a common privacy benchmark, does the choice of noise
distribution---Gaussian, Laplace, or logistic---affect welfare and revenue, or
only the statistical precision of the planner's information?

\bigskip

\paragraph{Our Contributions.}
This paper makes four contributions.

\begin{itemize}[leftmargin=1.5em,itemsep=4pt]

\item \textbf{Privacy reshapes how the good is provided.}
The planner never sees agents' true values, only reports passed through a
fixed, privacy-protecting noise channel. We show that the planner should
nonetheless act on a simple per-agent index that blends how much each agent is
expected to value the good with how much revenue can be collected from them,
and should provide the good when these indices sum to a positive total. The
amount of noise the channel adds controls how strongly provision responds to
people's underlying preferences; coarser channels can leave the planner unable
to tell apart agents whose reports are extreme.

\item \textbf{Knowing what to provide is not the same as being able to pay for
it.} We distinguish the allocation the planner would choose from the payments
needed to make truth-telling optimal. Recovering payments that depend only on
the noisy reports is a separate and harder problem---formally, solving an
integral (Fredholm) equation---that may have no exact solution. When an exact
solution exists it is unique, but existence is not automatic. The Online
Supplement therefore also constructs payments that work approximately, with
explicit bounds on how far incentives, participation, and revenue can be off.

\item \textbf{Privacy makes public goods harder to finance as the market
grows.} The revenue a privacy-respecting mechanism can raise is governed by the
lower endpoint of the valuation distribution---the smallest value an agent
could have. As the number of agents grows, revenue grows in proportion to the
market when that lower endpoint is positive, grows only at a slow square-root
rate when it is exactly zero, and shrinks to almost nothing when it is
negative. The change between a fundable and an unfundable
good is sharp, so a small tightening of the revenue requirement can switch
provision off.

\item \textbf{The better privacy technology depends on how privacy is measured.}
Comparing two natural noise designs, Laplace and logistic, we find no design is
better in every case. Holding the noise scale fixed---a common pure-\(\epsilon\)
calibration---the Laplace design gives the planner more usable information;
holding fixed instead a stricter hypothesis-testing privacy standard (Gaussian
differential privacy), the ranking reverses in the hardest case, that of telling
apart the highest- and lowest-value agents. Which
design a planner prefers thus depends on the privacy standard used to compare
them, and this reversal does not extend to a blanket ranking across all agents.

\end{itemize}

\bigskip
\paragraph{Related literature.}
Our work lies at the intersection of differential privacy, mechanism design,
and public-good provision. The privacy-aware mechanism-design literature begins
with \citet{McSherryTalwar2007}, who use the stability created by differential
privacy to obtain approximate truthfulness. We study a different problem: the
local privacy channel is fixed exogenously, privacy does not enter agents'
payoffs, and the planner must achieve exact Bayesian incentive compatibility
using allocation and transfer rules that depend only on privatized signals. The
classical transfer identity \citep{myerson1981optimal} therefore becomes an
inverse problem: recovering a signal-contingent transfer requires solving a
Fredholm equation of the first kind, where completeness of the channel kernel
gives uniqueness when a solution exists and existence is a separate range
condition. Recent work also studies privacy-constrained mechanism and
information design: \citet{EilatEliazMu2021} incorporate belief-based privacy
costs into mechanism design and show that privacy concerns can rationalize
coarse mechanisms, while \citet{StrackYang2024} characterize signals that
conceal a protected attribute through garbling and mean-preserving contraction
and \citet{HeSandomirskiyTamuz2026} study signals that reveal information about a
common state without revealing other agents' signals. We share their use of
garbling and the Blackwell order, but our planner does not design the
information structure: it faces a fixed local differential-privacy channel and
optimizes a public-good mechanism conditional on the information that channel
produces. Two further strands bear on our setting and are developed in the
Online Supplement: privacy as a payoff argument or a good to be compensated
\citep{GhoshRoth2011,NissimOrlandiSmorodinsky2012,ChenChongKashMoranVadhan2013},
and the statistical limits of local privacy
\citep{duchi2013local,duchi2018minimax}, where our concern is the
privacy--revenue rather than the privacy--estimation frontier.

\medskip

The paper most closely related to ours is \citet{PourbabaeeEchenique2025}, who
study binary public-good provision with \(\pm1\) types transmitted through a
randomized-response flip channel. They characterize Boolean decision rules
trading off revenue, surplus, and noise sensitivity. We allow continuous
valuations and a broader class of channels, so monotone likelihood ratios govern
implementability and generate features without a Boolean analogue, including the
distinction between strict and weak MLRP and Laplace tail pooling. We also
characterize transfers through a Fredholm equation, allow a hierarchical prior,
and derive three large-population revenue regimes---linear, square-root, and
exponentially small---according to the sign of the lower endpoint of the
valuation distribution.

\medskip

Finally, our asymptotic analysis relates to the literature on public-good
provision in large economies
\citep{GuthHellwig1986,mailath1990asymmetric,Hellwig2003}, where private
information obstructs efficient provision and, under budget balance, the
obstruction intensifies as the population grows. \citet{XiXie2023} identify a
square-root boundary in the growth rate of project costs for asymptotically
efficient provision. Their boundary concerns cost scaling and is therefore
analogous to, but distinct from, our knife-edge regime, in which maximal
reduced-form revenue is itself of square-root order. In our model the local
privacy channel further determines the speed at which revenue feasibility
deteriorates. The hierarchical extension is also related to
\citet{CremerMcLean1988}: a common latent parameter generates dependence across
types, while privacy weakens the information each signal conveys about that
parameter without eliminating the induced correlation.
\bigskip
\paragraph{Outline}
The paper is organized as follows. Section~\ref{sec:background} reviews the
public-good problem and the privacy channels; Section~\ref{sec:setup} sets up
the mediated reporting model and the posterior-score objects it relies on.
Section~\ref{sec:main_results} gives the main results: the optimal threshold
rule, the Fredholm implementation problem, the linear/square-root/exponential
revenue regimes, the hierarchical extension, and the channel welfare reversal.
Section~\ref{sec:numerics} illustrates some of the theoretical results derived in section~\ref{sec:main_results}. The appendices collect the
proofs and additional numerical results, and the online supplement develops
exact and approximate signal-measurable implementation, the revelation and
completeness arguments, and the full numerical methodology.

\section{Background}
\label{sec:background}
% ============================================================================

This section sets out the economic environment, the reporting and privacy
architecture, and the privacy criteria. The planner's mechanism-design
problem is formulated in Section~\ref{sec:setup}. Trade-off functions,
privacy calibrations, and channel-specific likelihood-ratio calculations are
collected in Appendix~\ref{app:background_derivations}.

% ------------------------------------------------------------------
\subsection{Agents and Preferences}
\label{subsec:agents}
% ------------------------------------------------------------------

A planner faces \(n\) agents indexed by \(i\in\{1,\ldots,n\}\). Agent \(i\)
privately observes a valuation
\[
X_i\in\X=[\xlo,\xhi]\subset\R.
\]
In the known-prior benchmark, the valuations are independently and
identically distributed according to a commonly known distribution \(F\)
with density \(f>0\) on \(\X\).

The planner decides whether to implement a binary public project,
\(q\in\{0,1\}\), and assigns transfers \(t_i\in\R\). Preferences are
quasilinear:
\begin{equation}
\label{eq:utility}
u_i(q,t_i\mid x_i)
=
q\,x_i-t_i.
\end{equation}
The valuation \(x_i\) is the net value of implementation, inclusive of any
exogenously assigned cost share. If the planner observed the valuation
profile \(x=(x_1,\ldots,x_n)\), the efficient allocation would be
\begin{equation}
\label{eq:first_best}
q^{\mathrm{FB}}(x)
=
\ind\!\left\{
\sum_{i=1}^{n}x_i\ge0
\right\}.
\end{equation}

\begin{example}[Shared community project]
\label{ex:community}
Suppose households decide whether to build a shared facility, such as a
community solar array. The value \(x_i\) is household \(i\)'s benefit net of
its assigned cost share. Thus \(x_i>0\) for a net beneficiary and \(x_i<0\)
for a net payer, while \(\xlo<0\) permits some households to be net losers.
The decision \(q=1\) means that the project is built, and \(t_i\) is the fee
charged to household \(i\). By \eqref{eq:first_best}, the efficient rule
builds exactly when aggregate net value is nonnegative. The planner does not
observe these valuations directly, which motivates the privacy channel
below.
\end{example}

% ------------------------------------------------------------------
\subsection{Strategic Reports and the Privacy Channel}
\label{subsec:channel}
% ------------------------------------------------------------------

After observing \(X_i=x_i\), agent \(i\) submits a report
\(R_i=r_i\in\X\) to a trusted local mediator. Conditional on the report, the
mediator generates a privatized signal through a fixed Markov kernel:
\begin{equation}
\label{eq:channel}
Y_i\mid R_i=r_i
\sim
K(\cdot\mid r_i).
\end{equation}
When the kernel admits a density, it is denoted by \(k(y_i\mid r_i)\).
Signals take values in a measurable space \(\Y\subseteq\R\), with
\(\Y=\R\) for the additive channels considered below.

The agent controls the report entering the channel but not the realized
channel noise. The planner observes only the signal profile
\[
Y=(Y_1,\ldots,Y_n)
\]
and observes neither \(X\) nor \(R\).

A randomized signal-measurable mechanism is a measurable map
\begin{equation}
\label{eq:mechanism_def}
(q,t):
\Y^n
\longrightarrow
[0,1]\times\R^n,
\qquad
y
\longmapsto
\bigl(q(y),t_1(y),\ldots,t_n(y)\bigr),
\end{equation}
where \(q(y)\) is the probability of implementation and \(t_i(y)\) is
agent \(i\)'s payment. Transfers are assumed measurable and integrable under
the induced signal laws. The realized project decision is binary;
\(q(y)\in[0,1]\) denotes its implementation probability conditional on the
signal profile.

We study direct reporting mechanisms, in which the value an agent submits is
itself the input to the fixed privacy channel. Bayesian incentive
compatibility, imposed in Section~\ref{subsec:planner_problem}, requires
truthful reporting \(R_i=X_i\) to be optimal within this reporting game.
Along the resulting truthful equilibrium path,
\begin{equation}
\label{eq:truthful_signal_law}
Y_i\mid X_i=x_i
\sim
K(\cdot\mid x_i).
\end{equation}
The report variable is retained to formulate deviations, whereas posterior
beliefs, welfare, revenue, and allocation are evaluated under the truthful
law \eqref{eq:truthful_signal_law}.

Conditional independence of the mediator's randomizations gives
\begin{equation}
\label{eq:truthful_joint_channel}
Y\mid X=x
\sim
\prod_{i=1}^{n}K(\cdot\mid x_i).
\end{equation}
The truthful-path marginal density of one signal is%
\footnote{Throughout, \(m(\cdot)\) denotes the one-signal marginal density.
The mean of the posterior score introduced later is denoted by
\(\mu_S(\lambda)\).}
\begin{equation}
\label{eq:marginal_signal_density}
m(y)
=
\int_{\X}k(y\mid x)f(x)\,dx.
\end{equation}
Under the independent known-prior benchmark,
\begin{equation}
\label{eq:joint_signal_density}
f_Y(y)
=
\prod_{i=1}^{n}m(y_i),
\qquad
f_{Y_{-i}}(y_{-i})
=
\prod_{j\ne i}m(y_j).
\end{equation}

% ------------------------------------------------------------------
\subsection{Privacy Criteria}
\label{subsec:privacy_defs}
% ------------------------------------------------------------------

Privacy is measured through hypothesis-testing trade-off functions, with
Gaussian differential privacy providing the common cross-family benchmark
\citep{dong2022gaussian}. For probability measures \(P\) and \(Q\) on a
common measurable signal space, define
\begin{equation}
\label{eq:tradeoff_function}
\mathcal T(P,Q)(\alpha)
=
\inf_{\psi}
\left\{
1-\E_Q[\psi]:
\E_P[\psi]\le\alpha
\right\},
\qquad
\alpha\in[0,1],
\end{equation}
where the infimum is over measurable randomized tests
\(\psi:\Y\to[0,1]\). Thus, \(\mathcal T(P,Q)(\alpha)\) is the smallest
attainable type-II error among tests whose type-I error is at most
\(\alpha\). A larger trade-off function corresponds to a less informative
binary experiment and therefore stronger hypothesis-testing privacy.

\begin{definition}[\(\mu\)-Gaussian local differential privacy]
\label{def:gdp}
For \(\mu\ge0\), define
\begin{equation}
\label{eq:gaussian_tradeoff}
G_\mu(\alpha)
=
\Phi\!\left(
\Phi^{-1}(1-\alpha)-\mu
\right),
\qquad
\alpha\in[0,1].
\end{equation}
A channel \(K\) satisfies \(\mu\)-Gaussian local differential privacy,
abbreviated \(\mu\)-GDP, if
\begin{equation}
\label{eq:gdp}
\mathcal T\!\left(
K(\cdot\mid r),
K(\cdot\mid r')
\right)(\alpha)
\ge
G_\mu(\alpha)
\end{equation}
for every \(r,r'\in\X\) and every \(\alpha\in[0,1]\).
\end{definition}

Smaller values of \(\mu\) correspond to stronger privacy. The Gaussian
shift experiment
\[
\Normal(0,1)
\quad\text{versus}\quad
\Normal(\mu,1)
\]
has trade-off function \(G_\mu\).

\begin{definition}[Approximate local differential privacy]
\label{def:approx_ldp}
For \(\epsilon\ge0\) and \(\delta\in[0,1]\), a channel \(K\) satisfies
\((\epsilon,\delta)\)-local differential privacy if
\begin{equation}
\label{eq:approx_ldp}
K(S\mid r)
\le
e^\epsilon K(S\mid r')
+
\delta
\end{equation}
for every \(r,r'\in\X\) and every measurable \(S\subseteq\Y\).
\end{definition}

The parameter \(\epsilon\) bounds multiplicative distinguishability, while
\(\delta\) allows an additive exceptional probability
\citep{dwork2006calibrating,DworkRoth2014}. A \(\mu\)-GDP channel satisfies
\((\epsilon,\delta_\mu(\epsilon))\)-LDP for every \(\epsilon\ge0\), where
\begin{equation}
\label{eq:gdp_to_approx}
\delta_\mu(\epsilon)
=
\Phi\!\left(
-\frac{\epsilon}{\mu}+\frac{\mu}{2}
\right)
-
e^\epsilon
\Phi\!\left(
-\frac{\epsilon}{\mu}-\frac{\mu}{2}
\right)
\end{equation}
\citep{dong2022gaussian}.

\begin{remark}[Pure local differential privacy]
\label{rem:pure_ldp}
Pure \(\epsilon\)-LDP is the case \(\delta=0\):
\begin{equation}
\label{eq:pure_ldp}
K(S\mid r)
\le
e^\epsilon K(S\mid r')
\qquad
\forall r,r'\in\X,
\quad
\forall\text{ measurable }S\subseteq\Y.
\end{equation}
When conditional densities exist, this is equivalent to
\begin{equation}
\label{eq:pure_ldp_density}
k(y\mid r)
\le
e^\epsilon k(y\mid r')
\qquad
\text{for a.e. }y,
\quad
\forall r,r'\in\X
\end{equation}
\citep{AwanKenneyReimherrSlavkovic2019}.

Pure LDP, approximate LDP, and GDP are distinct restrictions. Equal values
of \(\epsilon\), equal noise variances, or equal noise scales do not
generally equalize the complete testing trade-off functions of different
channel families.
\end{remark}

% ------------------------------------------------------------------
\subsection{Channel Families}
\label{subsec:channel_families}
% ------------------------------------------------------------------

We study additive location channels of the form
\begin{equation}
\label{eq:additive_channel}
Y_i
=
R_i+\eta_i,
\end{equation}
where the noise variables are independent across agents and independent of
\((X,R)\). Since reports lie in \([\xlo,\xhi]\), the largest separation
between any two reports is the type-space width
\begin{equation}
\label{eq:type_sensitivity}
\Delta_{\X}
=
\xhi-\xlo,
\end{equation}
which plays the role of the sensitivity in the privacy calibrations below.

\begin{table}[htbp]
\centering
\renewcommand{\arraystretch}{1.35}
\setlength{\tabcolsep}{6pt}
\begin{tabularx}{\textwidth}{@{}l *{3}{>{\raggedright\arraybackslash}X}@{}}
\toprule
\textbf{Channel}
&
\textbf{Signal model}
&
\textbf{Conditional density}
&
\textbf{Standard calibration}
\\
\midrule
Gaussian
&
\(Y_i^{G}=R_i+\eta_i^{G}\)\newline
\(\eta_i^{G}\sim\Normal(0,\sigma^2)\)
&
\(\dfrac{1}{\sigma}\,
\varphiN\!\left(\dfrac{y-r}{\sigma}\right)\)
&
Exact \(\mu_G(\sigma)\)-GDP,\newline
\(\mu_G(\sigma)=\Delta_{\X}/\sigma\);\newline
no finite pure \(\epsilon\)
\\
\addlinespace
Laplace
&
\(Y_i^{L}=R_i+\eta_i^{L}\)\newline
\(\eta_i^{L}\sim\Lap(0,b_L)\)
&
\(\dfrac{1}{2b_L}e^{-|y-r|/b_L}\)
&
Pure \(\epsilon\)-LDP,\newline
\(b_L=\Delta_{\X}/\epsilon\)
\\
\addlinespace
Logistic
&
\(Y_i^{\mathrm{Log}}=R_i+\eta_i^{\mathrm{Log}}\)\newline
\(\eta_i^{\mathrm{Log}}\sim\Logistic(0,\beta)\)
&
\(\dfrac{e^{-(y-r)/\beta}}
{\beta[1+e^{-(y-r)/\beta}]^2}\)
&
Pure \(\epsilon\)-LDP,\newline
\(\beta=\Delta_{\X}/\epsilon\)
\\
\bottomrule
\end{tabularx}
\caption{
Additive location channels, with
\(\Delta_{\X}=\xhi-\xlo\). The Gaussian channel is naturally calibrated
through GDP, whereas the Laplace and logistic channels admit finite
pure-LDP calibrations. The derivations are in
Appendix~\ref{app:background_derivations}.
}
\label{tab:channel_families}
\end{table}

Along the truthful path,
\begin{equation}
\label{eq:truthful_channel_summary}
Y_i^{G}
=
X_i+\eta_i^{G},
\qquad
Y_i^{L}
=
X_i+\eta_i^{L},
\qquad
Y_i^{\mathrm{Log}}
=
X_i+\eta_i^{\mathrm{Log}}.
\end{equation}

Privacy limits the distinguishability of reports. The monotone likelihood
ratio property introduced in Section~\ref{subsec:assumptions} is a distinct
property: it governs how signals order posterior beliefs and therefore how
the mechanism responds to signal realizations.

\newpage
\section{Setup}
\label{sec:setup}

This section formulates the planner's problem, states the maintained
assumptions, and records the reduced-form identities used in
Section~\ref{sec:main_results}. General incentive, integrability, posterior,
and monotonicity arguments are collected in
Appendix~\ref{app:auxiliary}; mechanism-specific proofs are collected in
Appendix~\ref{app:main_results_proofs}.

% ------------------------------------------------------------------
\subsection{The Planner's Problem}
\label{subsec:planner_problem}
% ------------------------------------------------------------------

Figure~\ref{fig:ldp_architecture} summarizes the timing. Agent \(i\)
observes \(X_i\), submits \(R_i\), the mediator draws
\(Y_i\sim K(\cdot\mid R_i)\), and the planner chooses \(q(Y)\) and \(t(Y)\).

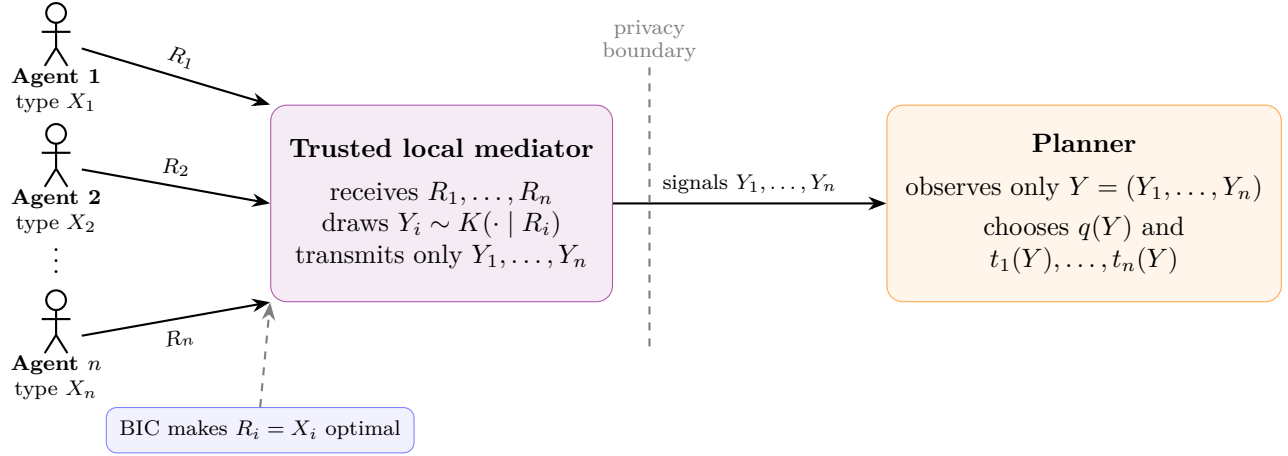
\begin{figure}[htbp]
\centering

\tikzset{
  person/.pic={
    \draw[thick] (0,0.52) circle[radius=0.13];
    \draw[thick] (0,0.39) -- (0,0.06);
    \draw[thick] (-0.19,0.27) -- (0.19,0.27);
    \draw[thick] (0,0.06) -- (-0.15,-0.18);
    \draw[thick] (0,0.06) -- (0.15,-0.18);
  }
}

\begin{tikzpicture}[
  >={Stealth[length=2.4mm]},
  font=\small,
  mediator/.style={
    rectangle,
    rounded corners=6pt,
    draw=violet!65,
    fill=violet!8,
    align=center,
    inner sep=7pt,
    minimum width=3.8cm,
    minimum height=2.6cm
  },
  planner/.style={
    rectangle,
    rounded corners=6pt,
    draw=orange!70,
    fill=orange!8,
    align=center,
    inner sep=7pt,
    minimum width=3.8cm,
    minimum height=2.6cm
  },
  note/.style={
    rectangle,
    rounded corners=4pt,
    draw=blue!50,
    fill=blue!5,
    align=center,
    font=\scriptsize,
    inner sep=5pt
  },
  arr/.style={->,thick},
  darr/.style={->,thick,dashed,gray}
]

\pic at (0,3.0) {person};
\node[font=\scriptsize,align=center] at (0,2.50)
{\textbf{Agent 1}\\ type \(X_1\)};

\pic at (0,1.4) {person};
\node[font=\scriptsize,align=center] at (0,0.90)
{\textbf{Agent 2}\\ type \(X_2\)};

\node at (0,0.35) {\(\vdots\)};

\pic at (0,-0.8) {person};
\node[font=\scriptsize,align=center] at (0,-1.30)
{\textbf{Agent \(n\)}\\ type \(X_n\)};

\node[mediator] (med) at (5.1,1.0)
{\textbf{Trusted local mediator}\\[4pt]
 receives \(R_1,\ldots,R_n\)\\
 draws \(Y_i\sim K(\cdot\mid R_i)\)\\
 transmits only \(Y_1,\ldots,Y_n\)};

\node[planner] (plan) at (13.6,1.0)
{\textbf{Planner}\\[4pt]
 observes only \(Y=(Y_1,\ldots,Y_n)\)\\[3pt]
 chooses \(q(Y)\) and\\
 \(t_1(Y),\ldots,t_n(Y)\)};

\draw[arr]
(0.34,3.05)
--
node[above,sloped,font=\scriptsize]{\(R_1\)}
(med.north west);

\draw[arr]
(0.34,1.45)
--
node[above,font=\scriptsize]{\(R_2\)}
(med.west);

\draw[arr]
(0.34,-0.75)
--
node[below,sloped,font=\scriptsize]{\(R_n\)}
(med.south west);

\draw[arr]
(med.east)
--
node[above,font=\scriptsize]{signals \(Y_1,\ldots,Y_n\)}
(plan.west);

\draw[dashed,thick,gray]
(7.85,2.8)
--
(7.85,-0.9);

\node[gray,font=\scriptsize,align=center] at (7.85,3.15)
{privacy\\[-1pt]boundary};

\node[note] (bicnote) at (2.7,-2.0)
{BIC makes \(R_i=X_i\) optimal};

\draw[darr]
(bicnote.north)
--
(med.south west);

\end{tikzpicture}

\caption{
Strategic local-privacy architecture. Agent \(i\) observes \(X_i\), chooses
a report \(R_i\), and is randomized by the trusted mediator. The planner
observes only the privatized signals and conditions the allocation and
transfers on those signals.
}
\label{fig:ldp_architecture}
\end{figure}
\clearpage
For a signal-measurable mechanism, define agent \(i\)'s interim allocation
and interim payment after report \(r_i\), holding the other agents to
truthful reporting, by
\begin{align}
Q_i(r_i)
&=
\int_{\Y^n}
q(y)\,
k(y_i\mid r_i)\,
f_{Y_{-i}}(y_{-i})\,dy,
\label{eq:interim_allocation}
\\
T_i(r_i)
&=
\int_{\Y^n}
t_i(y)\,
k(y_i\mid r_i)\,
f_{Y_{-i}}(y_{-i})\,dy.
\label{eq:interim_transfer}
\end{align}
A type \(x_i\) reporting \(r_i\) obtains
\begin{equation}
\label{eq:deviation_utility}
U_i(x_i,r_i)
=
x_iQ_i(r_i)-T_i(r_i),
\end{equation}
and truthful interim utility is
\begin{equation}
\label{eq:truthful_interim_utility}
U_i(x_i)
=
U_i(x_i,x_i)
=
x_iQ_i(x_i)-T_i(x_i).
\end{equation}

The planner maximizes expected surplus subject to incentive, participation,
and revenue constraints:
\begin{subequations}
\label{eq:primal}
\begin{align}
\sup_{q,t}\quad
&
\E\!\left[
q(Y)\sum_{i=1}^{n}X_i
\right]
\notag
\\
\text{subject to}\quad
&
x_iQ_i(x_i)-T_i(x_i)
\ge
x_iQ_i(r_i)-T_i(r_i),
&&
\forall x_i,r_i\in\X,
\label{eq:bic}
\\
&
U_i(x_i)\ge0,
&&
\forall x_i\in\X,
\label{eq:ir}
\\
&
\E\!\left[
\sum_{i=1}^{n}t_i(Y)
\right]
=
\sum_{i=1}^{n}\E[T_i(X_i)]
\ge R,
&&
\label{eq:revenue}
\\
&
q:\Y^n\to[0,1],
\qquad
t_i:\Y^n\to\R
\text{ measurable and integrable.}
&&
\notag
\end{align}
\end{subequations}
The constraints are Bayesian incentive compatibility, interim individual
rationality, and an ex-ante revenue requirement. The case \(R=0\) imposes
expected non-deficit
\citep{mas1995microeconomic,myerson1981optimal}.

% ------------------------------------------------------------------
\subsection{Structural Assumptions}
\label{subsec:assumptions}
% ------------------------------------------------------------------

\begin{assumption}[Regularity]
\label{ass:regular}
The distribution \(F\) is absolutely continuous on
\(\X=[\xlo,\xhi]\), with continuously differentiable density \(f>0\).
The virtual value
\begin{equation}
\label{eq:virtual_value}
\virt(x)
=
x-\frac{1-F(x)}{f(x)}
\end{equation}
is weakly increasing.
\end{assumption}

This is the standard one-dimensional regularity condition
\citep{myerson1981optimal}. Log-concavity of \(1-F\) is sufficient
\citep{bagnoli2005log}.

\begin{assumption}[Monotone likelihood ratio property]
\label{ass:mlrp}
The kernel has a strictly positive density \(k(y\mid x)\) on a common
support \(\Y\), and for every \(x_2>x_1\),
\begin{equation}
\label{eq:mlrp}
y
\longmapsto
\frac{k(y\mid x_2)}
{k(y\mid x_1)}
\end{equation}
is weakly increasing.
\end{assumption}

\begin{remark}[Strict and weak MLRP]
\label{rmk:strict_weak_mlrp}
Strict MLRP requires the likelihood ratio in \eqref{eq:mlrp} to be strictly
increasing, whereas weak MLRP permits flat regions. Posterior expectations
of increasing functions are weakly increasing under MLRP and strictly
increasing under strict MLRP under the corresponding nondegeneracy
conditions \citep{milgrom1981good}.

The Gaussian and logistic location channels satisfy strict MLRP. The
Laplace location channel satisfies weak MLRP because its likelihood ratio
is constant in both outer tails. On a bounded type space, this produces
exact posterior pooling in the global signal tails. The resulting posterior
score remains weakly increasing, so the allocation threshold remains
monotone; flat regions generate pooling or ties rather than requiring
ironing. The general argument is given in
Lemma~\ref{lem:mlrp_fosd} and
Remark~\ref{rmk:weak_mlrp_no_ironing} of
Appendix~\ref{app:auxiliary}. The channel-specific calculations are in
Appendix~\ref{app:background_derivations}.
\end{remark}

\begin{assumption}[Statistical completeness and transfer space]
\label{ass:completeness}
Let
\[
\mathcal H
\subseteq
L^1\!\bigl(\Y,m(y)\,dy\bigr)
\]
be a Banach space of admissible opponent-averaged transfers. Suppose that
there exists \(C_{\mathcal H}<\infty\) such that
\begin{equation}
\label{eq:transfer_space_domination}
\sup_{x\in\X}
\int_{\Y}
|g(y)|k(y\mid x)\,dy
\le
C_{\mathcal H}\|g\|_{\mathcal H}
\qquad
\forall g\in\mathcal H.
\end{equation}
Suppose further that bounded truncations of every \(g\in\mathcal H\)
belong to \(\mathcal H\) and converge to \(g\) in the
\(\mathcal H\)-norm.

If \(g\in\mathcal H\) satisfies
\begin{equation}
\label{eq:completeness}
\int_{\Y}
g(y)k(y\mid x)\,dy
=
0
\qquad
\forall x\in\X,
\end{equation}
then
\[
g=0
\qquad
m(y)\,dy\text{-almost everywhere}.
\]
\end{assumption}

The domination condition
\eqref{eq:transfer_space_domination} makes the conditional-expectation
operator bounded on the maintained transfer space, while completeness
\eqref{eq:completeness} makes that operator injective. These are distinct
properties. Boundedness guarantees that the operator is well defined and
continuous; completeness identifies the opponent-averaged transfer whenever
an implementation exists. Neither property implies that the desired interim
payment belongs to the range of the operator, so existence remains a separate
range condition.

Assumption~\ref{ass:completeness} is used only for signal-measurable
transfer implementation. It is not required for the reduced-form allocation
or revenue characterizations. Completeness is maintained here as a
high-level condition on the chosen transfer space \(\mathcal H\); verifying
it for particular channel families and function spaces is a separate
operator-specific question.

\begin{assumption}[Compactness and kernel continuity]
\label{ass:compact_spaces}
The type space \(\X=[\xlo,\xhi]\) is compact, \(\Y\subseteq\R\) is Borel,
and the kernel has a jointly measurable density satisfying
\begin{equation}
\label{eq:l1_kernel_continuity}
\lim_{x'\to x}
\int_{\Y}
\left|
k(y\mid x')-k(y\mid x)
\right|\,dy
=
0
\qquad
\forall x\in\X.
\end{equation}
Allocations satisfy \(0\le q\le1\), transfers are integrable under every
admissible report profile, and the feasible set is nonempty.
\end{assumption}

The signal space may be unbounded. The continuity and Fubini--Tonelli
consequences of this assumption are collected in
Appendix~\ref{app:auxiliary}. Additional boundedness and strict-feasibility
conditions used only for existence and multiplier arguments are stated in
Appendix~\ref{app:main_results_proofs}.

% ------------------------------------------------------------------
\subsection{Interim Reduction and Posterior Objects}
\label{subsec:myerson_reduction}
% ------------------------------------------------------------------

Although the privacy channel changes the mapping from reports to the
planner's observations, each agent's deviation problem remains
one-dimensional. By \eqref{eq:deviation_utility}, a type \(x_i\) that
reports \(r_i\) obtains
\[
x_iQ_i(r_i)-T_i(r_i),
\]
which is the standard one-dimensional screening payoff associated with
interim allocation \(Q_i\) and interim payment \(T_i\).

Bayesian incentive compatibility is therefore equivalent to weak
monotonicity of \(Q_i\) together with the envelope identity
\begin{equation}
\label{eq:envelope}
U_i(x_i)
=
U_i(\xlo)
+
\int_{\xlo}^{x_i}Q_i(z)\,dz,
\end{equation}
or, equivalently,
\begin{equation}
\label{eq:transfer_identity}
T_i(x_i)
=
x_iQ_i(x_i)
-
\int_{\xlo}^{x_i}Q_i(z)\,dz
-
U_i(\xlo).
\end{equation}
Since \(Q_i\ge0\), truthful utility is weakly increasing, and interim IR is
equivalent to \(U_i(\xlo)\ge0\). For a fixed allocation, expected payments
are maximized by imposing the zero-rent normalization
\begin{equation}
\label{eq:lowest_type_zero_rent}
U_i(\xlo)
=
0.
\end{equation}
These claims are established in
Proposition~\ref{prop:app_interim_bic} of
Appendix~\ref{app:auxiliary}.

Under the zero-rent normalization, expected payment equals expected
allocated virtual surplus:
\begin{equation}
\label{eq:revenue_virtual}
\E[T_i(X_i)]
=
\E\!\left[
Q_i(X_i)\virt(X_i)
\right].
\end{equation}
The proof is given in
Lemma~\ref{lem:expected_transfer_identity} of
Appendix~\ref{app:auxiliary}.

The planner conditions on privatized signals rather than valuations. Define
the posterior mean and posterior virtual value by
\begin{equation}
\label{eq:posterior_objects_setup}
\widehat x_i(y_i)
=
\E[X_i\mid Y_i=y_i],
\qquad
\widehat J_i(y_i)
=
\E[\virt(X_i)\mid Y_i=y_i].
\end{equation}
These functions are common across agents under symmetry. Independence of
types and channel draws implies
\[
\E[X_i\mid Y]
=
\E[X_i\mid Y_i],
\qquad
\E[\virt(X_i)\mid Y]
=
\E[\virt(X_i)\mid Y_i].
\]
Consequently,
\begin{align}
\E\!\left[
q(Y)\sum_{i=1}^{n}X_i
\right]
&=
\E\!\left[
q(Y)\sum_{i=1}^{n}\widehat x_i(Y_i)
\right],
\label{eq:posterior_welfare}
\\
\E\!\left[
\sum_{i=1}^{n}t_i(Y)
\right]
&=
\E\!\left[
q(Y)\sum_{i=1}^{n}\widehat J_i(Y_i)
\right]
\label{eq:posterior_revenue}
\end{align}
under the zero-rent normalization. For arbitrary lowest-type rents, the
right-hand side of \eqref{eq:posterior_revenue} is reduced by
\(\sum_iU_i(\xlo)\). These identities are proved in
Lemma~\ref{lem:posterior_representations} of
Appendix~\ref{app:auxiliary}.

Privacy therefore enters the reduced-form problem through the posterior
expectations induced by the channel. The planner evaluates each signal using
the posterior mean \(\widehat x_i(Y_i)\) and posterior virtual value
\(\widehat J_i(Y_i)\).

For \(\lambda\ge0\), define the generalized virtual value
\begin{equation}
\label{eq:generalized_virtual_value}
\Psival(x,\lambda)
=
x+\lambda\virt(x)
\end{equation}
and the posterior score
\begin{equation}
\label{eq:posterior_score_setup}
S_i(y_i,\lambda)
=
\widehat x_i(y_i)
+
\lambda\widehat J_i(y_i)
=
\E\!\left[
\Psival(X_i,\lambda)
\mid
Y_i=y_i
\right].
\end{equation}
The score is the per-agent contribution to the planner's weighted
welfare--revenue objective. Regularity makes
\(\Psival(\cdot,\lambda)\) weakly increasing, and MLRP implies that
\(S_i(\cdot,\lambda)\) is weakly increasing. Under strict MLRP, if
\(\Psival(\cdot,\lambda)\) is nonconstant on a set of positive prior
probability, then \(S_i(\cdot,\lambda)\) is strictly increasing. The general
argument is given in Lemma~\ref{lem:mlrp_fosd} of
Appendix~\ref{app:auxiliary}.

A useful identity for the large-population analysis is
\begin{equation}
\label{eq:lower_endpoint}
\E[\virt(X)]
=
\E[\widehat J_i(Y_i)]
=
\xlo.
\end{equation}
The proof is given in Lemma~\ref{lem:lower_endpoint} of
Appendix~\ref{app:auxiliary}. Thus, \(\xlo\) is the mean posterior virtual
value. When \(\xlo>0\), maximum reduced-form revenue has positive drift.
When \(\xlo<0\), positive aggregate virtual surplus is a large-deviation
event. The case \(\xlo=0\) is the central-limit boundary between these two
regimes.

Define the reduced-form monotone allocation class by
\begin{equation}
\label{eq:monotone_allocation_class}
\mathcal Q^{\mathrm{mon}}
=
\left\{
q:\Y^n\to[0,1]:
q\text{ is measurable and }
Q_i^q\text{ is weakly increasing on }\X
\text{ for every }i
\right\}.
\end{equation}
When the dependence on the channel matters, the same class is denoted by
\(\mathcal Q^{\mathrm{mon}}(K)\). For \(\lambda\ge0\), define the aggregate
posterior score
\begin{equation}
\label{eq:aggregate_score_def}
G_\lambda(y)
=
\sum_{i=1}^{n}S_i(y_i,\lambda).
\end{equation}
The threshold characterization in Section~\ref{sec:main_results} selects
the project when \(G_\lambda>0\). If the aggregate score places positive
probability on zero, the allocation on the zero-score set creates a separate
selection problem: a tie-breaking rule must satisfy both the revenue
requirement and the monotonicity restrictions. To obtain an unambiguous
almost-sure threshold characterization, we impose the following condition.

\begin{assumption}[No mass at the score threshold]
\label{ass:aggregate_score_atomless}
For every \(\lambda\ge0\) that arises as a supporting multiplier of the
reduced-form problem formally stated in
\eqref{eq:main_reduced_problem},
\begin{equation}
\label{eq:aggregate_score_atomless}
\Prob\!\left(G_\lambda(Y)=0\right)
=
0.
\end{equation}
\end{assumption}

A sufficient condition is that, at each relevant multiplier, at least one
coordinate score \(S_i(Y_i,\lambda)\) have an atomless distribution and be
independent of the sum of the remaining coordinate scores.

Continuity of the
underlying signals alone is not sufficient. In particular, Laplace tail
pooling creates atoms in the distribution of each individual posterior score
(Remark~\ref{rmk:strict_weak_mlrp}). Although these atoms need not generate
positive mass at zero in the aggregate score, they prevent aggregate
atomlessness from being inferred automatically.

Assumption~\ref{ass:aggregate_score_atomless} must therefore either be
verified for the channel and prior under study or maintained as an additional
regularity condition.

The preceding identities characterize the reduced-form allocation and
interim-payment requirements. Whether the required interim payment schedule
can be generated by an actual signal-contingent transfer is a separate inverse
problem, formulated as a Fredholm equation in
Proposition~\ref{prop:fredholm}.
\bigskip

\section{Main Results}
\label{sec:main_results}

This section characterizes the optimal reduced-form allocation under a known
prior, the implementation of its interim payment schedule through
signal-measurable transfers, the asymptotic behavior of maximum reduced-form
revenue, the hierarchical-prior extension, and welfare comparisons across
privacy channels. The supporting functional-analytic and probabilistic
arguments are collected in Appendix~\ref{app:main_results_proofs}.

% ------------------------------------------------------------------
\subsection{Optimal Allocation under a Known Prior}
\label{subsec:known_prior}
% ------------------------------------------------------------------

Under the zero-rent normalization \(U_i(\xlo)=0\), the posterior welfare and
revenue representations give
\begin{align}
\mathsf W(q)
&=
\E\!\left[
q(Y)\sum_{i=1}^{n}\widehat x_i(Y_i)
\right],
\label{eq:main_welfare_functional}
\\
\mathsf V(q)
&=
\E\!\left[
q(Y)\sum_{i=1}^{n}\widehat J_i(Y_i)
\right].
\label{eq:main_revenue_functional}
\end{align}
The reduced-form problem is therefore
\begin{equation}
\label{eq:main_reduced_problem}
\sup_{q\in\mathcal Q^{\mathrm{mon}}}
\left\{
\mathsf W(q):
\mathsf V(q)\ge R
\right\}.
\end{equation}
The Lagrangian associated with the revenue constraint is
\begin{equation}
\label{eq:main_known_prior_lagrangian}
\mathcal L(q,\lambda)
=
\E[q(Y)G_\lambda(Y)]
-
\lambda R,
\end{equation}
with \(G_\lambda\) the aggregate posterior score
\eqref{eq:aggregate_score_def}. Regularity makes
\(x\mapsto x+\lambda\virt(x)\) weakly increasing, and MLRP makes each
\(S_i(\cdot,\lambda)\) weakly increasing, so \(G_\lambda\) is coordinatewise
weakly increasing.

\begin{theorem}[Optimal posterior-score allocation]
\label{thm:optimal_allocation}
Suppose Assumptions~\ref{ass:regular}, \ref{ass:mlrp},
\ref{ass:compact_spaces}, and \ref{ass:aggregate_score_atomless} hold,
together with the strict reduced-form revenue-feasibility condition
\begin{equation}
\label{eq:main_strict_feasibility}
\exists\,q^\circ\in\mathcal Q^{\mathrm{mon}}
\quad\text{such that}\quad
\mathsf V(q^\circ)>R.
\end{equation}
Then there exists a finite multiplier \(\lamstar\ge0\) such that
\begin{equation}
\label{eq:optimal_allocation}
\qstar(y)
=
\ind\!\left\{
G_{\lamstar}(y)>0
\right\}
\qquad
\text{almost surely under the truthful signal law}
\end{equation}
solves the reduced-form problem \eqref{eq:main_reduced_problem}. The
allocation \(\qstar\) is coordinatewise weakly increasing in the signal
profile and induces weakly increasing interim allocations; it is therefore
reduced-form Bayesian incentive compatible. Moreover,
\begin{equation}
\label{eq:optimal_complementary_slackness}
\lamstar
\left(
\mathsf V(\qstar)-R
\right)
=
0.
\end{equation}
\end{theorem}

\noindent
\textit{Proof.} Appendix~\ref{app:proof_optimal_allocation}.

The theorem has a simple interpretation. The planner selects the project
whenever the sum of the agents' posterior welfare--revenue scores is
positive. If the revenue constraint is slack, complementary slackness
implies \(\lamstar=0\), and the rule reduces to the
posterior-surplus-maximizing allocation. When \(\lamstar>0\), the revenue
constraint binds and posterior virtual values receive positive weight.
Assumption~\ref{ass:aggregate_score_atomless} ensures that the zero-score
set is null, so no separate tie-breaking rule is needed. A binding
constraint may, in a degenerate case, still be supported by a zero
multiplier, so the converse implication need not hold.

\begin{remark}[Tie sets]
\label{rmk:optimal_allocation_ties}
Under Assumption~\ref{ass:aggregate_score_atomless} the threshold set
\(\{G_{\lamstar}=0\}\) is null, so \(\qstar=\ind\{G_{\lamstar}\ge0\}\) almost
surely and the value assigned on the threshold is immaterial. If instead
\(\Prob(G_{\lamstar}(Y)=0)>0\), a constant randomization on the tie set need
not move revenue in the required direction; exact revenue selection then
requires a separately constructed measurable tie rule that also preserves
coordinatewise monotonicity, as discussed in
Remark~\ref{rmk:positive_mass_tie}. No such conclusion follows from convex
duality alone.
\end{remark}

\begin{remark}[Strict score monotonicity]
\label{rmk:strict_score_monotonicity}
Suppose the channel satisfies strict MLRP and
\(x\mapsto x+\lamstar\virt(x)\) is nonconstant on a set of positive prior
probability. Then \(S_i(\cdot,\lamstar)\) is strictly increasing. If, in
addition, the distribution of the opponents' aggregate score assigns
positive probability to every relevant threshold neighborhood, then the
induced interim allocation is strictly increasing over the corresponding
report region.
\end{remark}

\begin{example}[Uniform prior with Gaussian noise]
\label{ex:uniform_gaussian_main}
Suppose
\[
X_i\sim\Unif[-1,1],
\qquad
Y_i=X_i+\eta_i,
\qquad
\eta_i\sim\Normal(0,\sigma^2),
\]
independently across agents. Then \(\virt(x)=2x-1\). Let
\(h(y,\sigma)=\E[X_i\mid Y_i=y]\). It follows that
\[
\widehat x_i(y)=h(y,\sigma),
\qquad
\widehat J_i(y)=2h(y,\sigma)-1,
\]
and therefore
\begin{equation}
\label{eq:uniform_gaussian_score}
S_i(y,\lambda)
=
(1+2\lambda)h(y,\sigma)-\lambda.
\end{equation}
Consequently, the optimal rule thresholds the aggregate posterior mean,
\begin{equation}
\label{eq:uniform_gaussian_threshold}
\qstar(y)
=
\ind\!\left\{
\sum_{i=1}^{n}h(y_i,\sigma)
>
\frac{n\lamstar}{1+2\lamstar}
\right\}.
\end{equation}
The explicit formula for \(h(y,\sigma)\) and its strict monotonicity are
derived in Appendix~\ref{app:background_derivations}.
\end{example}

% ------------------------------------------------------------------
\subsection{Implementation through Signal-Measurable Transfers}
\label{subsec:transfer_implementation}
% ------------------------------------------------------------------

Theorem~\ref{thm:optimal_allocation} characterizes the optimal allocation and
its interim envelope payments. The existence of an actual transfer
\(t_i:\Y^n\to\R\) generating those interim payments is a distinct inverse
problem. Let \(\mathcal H\) be the Banach space of admissible
opponent-averaged transfers specified in
Appendix~\ref{app:channel_operator}, and define the channel operator
\begin{equation}
\label{eq:main_channel_operator}
\mathcal K:\mathcal H\to C(\X),
\qquad
(\mathcal Kg)(x)
=
\int_{\Y}g(y)k(y\mid x)\,dy.
\end{equation}

\begin{proposition}[Fredholm equation for implementing transfers]
\label{prop:fredholm}
Suppose Assumptions~\ref{ass:completeness} and \ref{ass:compact_spaces}
hold, let \(q\) induce a Bayesian incentive-compatible interim allocation,
and impose the zero-rent normalization \(U_i(\xlo)=0\). A signal-measurable
transfer \(t_i:\Y^n\to\R\) implements the envelope payment associated with
\(q\) if and only if, for every \(x_i\in\X\),
\begin{equation}
\label{eq:fredholm}
\begin{aligned}
&
\int_{\Y^n}
t_i(y)\,
k(y_i\mid x_i)\,
f_{Y_{-i}}(y_{-i})\,dy
\\
&\qquad =
\int_{\Y^n}
q(y)
\left[
x_i k(y_i\mid x_i)
-
\int_{\xlo}^{x_i}
k(y_i\mid z)\,dz
\right]
f_{Y_{-i}}(y_{-i})\,dy.
\end{aligned}
\end{equation}
Equivalently, define the opponent-averaged transfer under the truthful
signal law by
\[
\bar t_i(y_i)
=
\E\!\left[
t_i(Y)\mid Y_i=y_i
\right],
\]
and let \(\tau_i^q\) denote the required interim payment in
\eqref{eq:fredholm}, viewed as a function of \(x_i\). Then
\begin{equation}
\label{eq:fredholm_operator_form}
\mathcal K\bar t_i
=
\tau_i^q.
\end{equation}
An opponent-averaged transfer exists if and only if
\[
\tau_i^q\in\operatorname{Ran}(\mathcal K).
\]
When it exists, completeness makes \(\bar t_i\) unique up to
\(m(y_i)\,dy_i\)-almost-everywhere equality. In particular, the optimal
allocation \(\qstar\) of Theorem~\ref{thm:optimal_allocation} admits a
signal-measurable transfer implementation if and only if
\[
\tau_i^{\qstar}
\in
\operatorname{Ran}(\mathcal K)
\qquad
\text{for every }i.
\]
The full ex-post transfer is not unique. If \(t_i^0\) is one admissible
implementation, then every other admissible implementation has the form
\begin{equation}
\label{eq:main_expost_transfer_class}
t_i(y)
=
t_i^0(y)+r_i(y),
\qquad
\E\!\left[
r_i(Y)\mid Y_i
\right]
=
0
\quad
\text{under the truthful signal law},
\end{equation}
where \(r_i\) must remain measurable and integrable under every admissible
report profile. Conversely, every such admissible perturbation leaves the
interim payment schedule unchanged.

\end{proposition}

\noindent
\textit{Proof.} Appendix~\ref{app:channel_operator}.

\begin{remark}[Existence versus identification]
\label{rmk:fredholm_existence_identification}
Completeness gives injectivity of the channel operator and therefore
identification of the opponent-averaged transfer. It does not imply that the
desired interim payment lies in the range of the operator. Transfer existence
is a separate range condition, so reduced-form revenue and revenue in the
original mechanism problem coincide only when that condition holds.
\end{remark}

\begin{remark}[Ill-posedness]
\label{rmk:fredholm_illposedness}
If \(\operatorname{Ran}(\mathcal K)\) is not closed in \(C(\X)\), the inverse
is unbounded on its range, so small perturbations of the interim payment
schedule may require large changes in the implementing transfer. This is the
standard ill-posedness of a first-kind Fredholm equation
\citep{kress2014linear}.
\end{remark}

% ------------------------------------------------------------------
\subsection{Maximum Reduced-Form Revenue}
\label{subsec:revenue_behavior}
% ------------------------------------------------------------------

Define aggregate posterior virtual surplus by
\begin{equation}
\label{eq:aggregate_posterior_virtual_surplus}
V_n(Y)
=
\sum_{i=1}^{n}\widehat J_i(Y_i).
\end{equation}
Since each \(\widehat J_i\) is weakly increasing under regularity and MLRP,
the allocation
\begin{equation}
\label{eq:revenue_maximizing_rule}
q^{\mathrm{rev}}(y)
=
\ind\!\left\{
V_n(y)>0
\right\}
\end{equation}
is coordinatewise weakly increasing and therefore belongs to
\(\mathcal Q^{\mathrm{mon}}\). Because the pointwise maximizer of
\(aV_n(y)\) over \(a\in[0,1]\) equals one when \(V_n(y)>0\), zero when
\(V_n(y)<0\), and is immaterial to the objective when \(V_n(y)=0\), it
follows that
\begin{equation}
\label{eq:max_reduced_revenue}
R_n^{*,\mathrm{red}}(K)
=
\sup_{q\in\mathcal Q^{\mathrm{mon}}}
\E[q(Y)V_n(Y)]
=
\E[(V_n(Y))_+].
\end{equation}
By the lower-endpoint identity \(\E[\widehat J_i(Y_i)]=\xlo\), the asymptotic
behavior of maximum reduced-form revenue is governed by the sign of \(\xlo\).

\begin{theorem}[Maximum reduced-form revenue: three regimes]
\label{thm:three_regimes_known}
Suppose Assumptions~\ref{ass:regular}, \ref{ass:mlrp}, and
\ref{ass:compact_spaces} hold, and suppose
\[
\sigma_J^2
=
\Var(\widehat J_i(Y_i))
\in(0,\infty).
\]
Then the following statements hold.
\begin{enumerate}[label=(\roman*),leftmargin=2em,itemsep=4pt]
\item If \(\xlo>0\), there exist constants \(c,C>0\) such that
\[
\left|
R_n^{*,\mathrm{red}}(K)-n\xlo
\right|
\le
Cne^{-cn}.
\]
\item If \(\xlo=0\), then
\[
R_n^{*,\mathrm{red}}(K)
=
\frac{\sigma_J}{\sqrt{2\pi}}\sqrt n
+
o(\sqrt n).
\]
\item If \(\xlo<0\), define the Cram\'er rate function
\begin{equation}
\label{eq:known_prior_rate_function}
I_K(a)
=
\sup_{t\in\R}
\left\{
ta-\log\E[e^{t\widehat J_i(Y_i)}]
\right\}.
\end{equation}
If \(0\) lies in the interior of the convex hull of
\(\supp\widehat J_i(Y_i)\), then
\[
\lim_{n\to\infty}
\frac1n
\log R_n^{*,\mathrm{red}}(K)
=
-I_K(0)
<
0.
\]
\end{enumerate}
\end{theorem}

\noindent
\textit{Proof.} Appendix~\ref{app:proof_three_regimes_known}.

The theorem concerns maximum reduced-form revenue. Equality with revenue in
the original mechanism problem requires that the envelope payment generated by
\(q^{\mathrm{rev}}\) satisfy the range condition in
Proposition~\ref{prop:fredholm}.

\begin{remark}[Probability and revenue in the negative-drift regime]
\label{rmk:probability_revenue_rate}
When \(\xlo<0\), Cram\'er's theorem also gives
\[
\lim_{n\to\infty}
\frac1n
\log\Prob(V_n\ge0)
=
-I_K(0).
\]
Thus the probability of positive aggregate posterior virtual surplus and
maximum reduced-form revenue have the same exponential decay exponent
\(I_K(0)\), although they are distinct finite-sample quantities.
\end{remark}

% ------------------------------------------------------------------
\subsection{Revenue of a Fixed Posterior-Score Rule}
\label{subsec:fixed_score_revenue}
% ------------------------------------------------------------------

For \(\lambda\ge0\), let
\(W_i(\lambda)=\widehat x_i(Y_i)+\lambda\widehat J_i(Y_i)=S_i(Y_i,\lambda)\)
and let
\begin{equation}
\label{eq:mean_score_def}
\mu_S(\lambda)
=
\E[S_i(Y_i,\lambda)]
=
\E[X]+\lambda\,\xlo
\end{equation}
denote the per-agent mean posterior score; consider
\(q_\lambda(Y)=\ind\{\sum_iW_i(\lambda)\ge0\}\).

\begin{proposition}[Revenue of a fixed posterior-score rule]
\label{prop:asymptotic_revenue}
Fix \(\lambda\ge0\) with \(\mu_S(\lambda)>0\). Then there exist constants
\(c_\lambda,C_\lambda>0\) such that
\begin{equation}
\label{eq:fixed_rule_acceptance}
\Prob(q_\lambda(Y)=0)
\le
C_\lambda e^{-c_\lambda n},
\end{equation}
and
\begin{equation}
\label{eq:fixed_rule_revenue_asymptotics}
\E[V_n(Y)q_\lambda(Y)]
=
n\xlo
+
\mathcal O\!\left(ne^{-c_\lambda n}\right).
\end{equation}
Consequently, the rule generates positive linear-order revenue when
\(\xlo>0\) and a negative linear-order deficit when \(\xlo<0\).
\end{proposition}

\noindent
\textit{Proof.} Appendix~\ref{app:proof_asymptotic_revenue}.

\begin{remark}[Knife-edge multiplier]
\label{rmk:knife_edge_multiplier}
Suppose \(\xlo\ne0\). The mean score vanishes at
\(\lambda_{\mathrm{crit}}=-\E[X]/\xlo\), provided this is nonnegative. If
\(\xlo<0\), then \(\mu_S(\lambda)\) is decreasing, so implementation converges
to one for fixed \(\lambda<\lambda_{\mathrm{crit}}\) and to zero for fixed
\(\lambda>\lambda_{\mathrm{crit}}\); if \(\xlo>0\) the direction is reversed.
A nondegenerate local transition requires a sequence
\(\lambda_n-\lambda_{\mathrm{crit}}=\mathcal O(n^{-1/2})\) and follows from a
separate central-limit calculation rather than directly from
Proposition~\ref{prop:asymptotic_revenue}.
\end{remark}

\begin{proposition}[Local transition at the mean-score boundary]
\label{prop:local_multiplier_transition}
Suppose \(\xlo\ne0\), let
\[
\lambda_{\mathrm{crit}}
=
-\frac{\E[X]}{\xlo}
\ge0,
\]
and assume
\[
\sigma_{\mathrm{crit}}^2
=
\Var\!\left(
S_i(Y_i,\lambda_{\mathrm{crit}})
\right)
\in(0,\infty).
\]
For a fixed \(h\in\R\), let
\[
\lambda_n
=
\lambda_{\mathrm{crit}}
+
\frac{h}{\sqrt n}.
\]
Then
\begin{equation}
\label{eq:local_multiplier_transition}
\Prob\!\left(
\sum_{i=1}^{n}
S_i(Y_i,\lambda_n)
\ge0
\right)
\longrightarrow
\Phi\!\left(
\frac{h\xlo}{\sigma_{\mathrm{crit}}}
\right).
\end{equation}
Equivalently, defining
\[
u_n
=
-\frac{
\sqrt n\,\xlo
(\lambda_n-\lambda_{\mathrm{crit}})
}{
\sigma_{\mathrm{crit}}
},
\]
the implementation probability converges to \(\Phi(-u)\) whenever
\(u_n\to u\).
\end{proposition}

\noindent
\textit{Proof.} Appendix~\ref{app:proof_local_multiplier_transition}.
\subsection{Hierarchical Prior}
\label{subsec:hierarchical}
% ------------------------------------------------------------------

Let \(\theta\sim\pi\). Conditional on \(\theta\), suppose
\(X_1,\ldots,X_n\overset{\mathrm{iid}}{\sim}F_\theta\) on the common support
\(\X=[\xlo,\xhi]\). Agents observe \(\theta\), while the planner observes only
\(Y\); incentive compatibility and individual rationality are imposed
conditional on \(\theta\). For every \(\theta\), define
\(\widehat x^\theta(y_i)=\E[X_i\mid\theta,Y_i=y_i]\),
\(\widehat J^\theta(y_i)=\E[\virt_\theta(X_i)\mid\theta,Y_i=y_i]\), and the
hierarchical per-agent and aggregate scores
\begin{equation}
\label{eq:hierarchical_score}
\Psi_i(y,\lambda)
=
\E_{\theta\mid y}
\left[
\widehat x^\theta(y_i)
+
\lambda\widehat J^\theta(y_i)
\right],
\qquad
G_\lambda^{\mathrm H}(y)
=
\sum_{i=1}^{n}\Psi_i(y,\lambda).
\end{equation}
Unlike the known-prior score, \(\Psi_i(y,\lambda)\) generally depends on the
entire signal profile through the posterior of \(\theta\).
We first study a conditional reduced-form relaxation. In this relaxation,
for each value of \(\theta\), the allocation must induce a conditionally
monotone interim allocation and therefore admits a
\(\theta\)-indexed envelope-payment schedule. Because the planner does not
observe \(\theta\), these conditional payment schedules need not be jointly
implementable by a single signal-measurable ex-post transfer. Full
implementation therefore requires the additional simultaneous Fredholm
compatibility condition described in
Remark~\ref{rmk:hierarchical_transfer}.
\begin{theorem}[Hierarchical posterior-score allocation]
\label{thm:hierarchical_allocation}
Suppose Assumptions~\ref{ass:hierarchical_regular},
\ref{ass:hierarchical_score_monotonicity}, and
\ref{ass:hierarchical_score_atomless} of
Appendix~\ref{app:hierarchical_proof} hold. Suppose further that the
hierarchical reduced-form revenue requirement is
strictly feasible: there exists an allocation
\(q_{\mathrm H}^{\circ}\) that induces a weakly increasing conditional
interim allocation for every \(\theta\) and satisfies
\[
\mathsf V_{\mathrm H}(q_{\mathrm H}^{\circ})
>
R.
\]
Then there exists a finite multiplier \(\lamstar\ge0\) such that
\begin{equation}
\label{eq:hierarchical_optimal_allocation}
q_{\mathrm H}^*(y)
=
\ind\!\left\{
G_{\lamstar}^{\mathrm H}(y)>0
\right\}
\qquad
\text{almost surely under the hierarchical truthful signal law}
\end{equation}
solves the hierarchical conditional reduced-form allocation problem.
For every \(\theta\), it induces a weakly increasing conditional interim
allocation and therefore admits a \(\theta\)-indexed conditional envelope
payment schedule. Moreover,
\[
\lamstar
\left(
\mathsf V_{\mathrm H}(q_{\mathrm H}^*)-R
\right)
=
0.
\]
\end{theorem}

\noindent
\textit{Proof.} Appendix~\ref{app:hierarchical_proof}.

\begin{remark}[Scope of the hierarchical theorem]
\label{rmk:hierarchical_scope}
Theorem~\ref{thm:hierarchical_allocation} solves a conditional reduced-form
relaxation. It does not by itself establish the existence of a single
signal-measurable transfer \(t_i(Y)\) that implements all of the
\(\theta\)-indexed conditional envelope-payment schedules simultaneously.
Accordingly, the value of the relaxation is an upper bound on the value of
the fully implementable hierarchical mechanism problem unless the
simultaneous range condition in
Remark~\ref{rmk:hierarchical_transfer} is verified.

The theorem also assumes a common conditional support
\([\xlo,\xhi]\) for all \(\theta\). A model in which \(\theta\) is itself
an unknown endpoint, such as
\(X_i\mid\theta\sim\Unif[\xlo,\theta]\), has
\(\theta\)-dependent support and is nonregular; it may be studied
numerically but is not covered by
Theorem~\ref{thm:hierarchical_allocation}.
\end{remark}
\begin{remark}[Hierarchical transfer implementation]
\label{rmk:hierarchical_transfer}
Full implementation requires a single signal-measurable transfer
\(t_i:\Y^n\to\R\) to satisfy the conditional Fredholm equation for every
\(\theta\). Conditional completeness identifies the corresponding
conditional opponent average whenever a solution exists, but it does not
imply that the family of required conditional interim payments lies in the
joint range generated by one common ex-post transfer. Thus conditional
envelope implementability for each \(\theta\) separately is necessary but
not sufficient for implementation in the original hierarchical mechanism.
\end{remark}

% ------------------------------------------------------------------
\subsection{Welfare Comparisons across Channels}
\label{subsec:channel_ordering}
% ------------------------------------------------------------------

Let \(\mathcal Q^{\mathrm{mon}}(K)\) denote the reduced-form allocation rules
that induce weakly increasing interim allocations under channel \(K\), and
define
\begin{equation}
\label{eq:reduced_value_channel}
W^{\mathrm{red}}(R;K)
=
\sup_{q\in\mathcal Q^{\mathrm{mon}}(K)}
\left\{
\E\!\left[
q(Y)\sum_{i=1}^{n}\widehat x_i(Y_i)
\right]:
\E\!\left[
q(Y)\sum_{i=1}^{n}\widehat J_i(Y_i)
\right]
\ge R
\right\}.
\end{equation}

\begin{proposition}[Blackwell monotonicity]
\label{prop:blackwell_monotone}
If \(K_1\succeq_{\mathrm B}K_2\), then
\begin{equation}
\label{eq:blackwell_value_monotonicity}
W^{\mathrm{red}}(R;K_1)
\ge
W^{\mathrm{red}}(R;K_2)
\end{equation}
for every prior, population size, and revenue target that is reduced-form
feasible under \(K_2\).
\end{proposition}

\noindent
\textit{Proof.} Appendix~\ref{app:proof_blackwell_monotone}.

\begin{remark}[Reduced-form versus full implementation]
\label{rmk:blackwell_reduced_form_scope}
Proposition~\ref{prop:blackwell_monotone} compares reduced-form values.
Equality with the value of the original mechanism problem additionally
requires transfer implementation under the relevant channel.
\end{remark}

% ..................................................................
\subsubsection{Equal-Scale Laplace--Logistic Comparison}
\label{subsubsec:equal_scale}
% ..................................................................

\begin{proposition}[Laplace dominates logistic at equal scale]
\label{prop:laplace_logistic_blackwell}
For every \(s>0\),
\begin{equation}
\label{eq:equal_scale_blackwell_order}
K_{\mathrm{Lap},s}
\succeq_{\mathrm B}
K_{\mathrm{Log},s},
\end{equation}
and hence
\(W^{\mathrm{red}}(R;K_{\mathrm{Lap},s})\ge
W^{\mathrm{red}}(R;K_{\mathrm{Log},s})\). Since both channels have pure-LDP
frontier scale \(s=\Delta_{\X}/\epsilon\), the same ordering holds under a
common pure-\(\epsilon\) calibration.
\end{proposition}

\noindent
\textit{Proof.} Appendix~\ref{app:equal_scale_blackwell}.

% ..................................................................
\subsubsection{Common Tight-\texorpdfstring{\(\mu\)}{mu} Calibration}
\label{subsubsec:common_mu}
% ..................................................................

Let \(\mu_K(s)\) denote the smallest GDP parameter for which the channel at
scale \(s\) satisfies \(\mu_K(s)\)-GDP.

\begin{proposition}[Common tight-\(\mu\) frontier scales]
\label{prop:common_mu_calibration}
For every \(\mu>0\), the Gaussian, Laplace, and logistic frontier scales are
\begin{equation}
\label{eq:common_mu_scales}
\sigma^*(\mu)
=
\frac{\Delta_{\X}}{\mu},
\qquad
b_L^*(\mu)
=
\frac{\Delta_{\X}}
{-2\log\!\bigl(2\Phi(-\mu/2)\bigr)},
\qquad
\beta^*(\mu)
=
\frac{\Delta_{\X}}
{2\log\!\bigl(
\Phi(\mu/2)/\Phi(-\mu/2)
\bigr)}.
\end{equation}
They satisfy
\(\mu_G(\sigma^*(\mu))=\mu_L(b_L^*(\mu))=\mu_{\mathrm{Log}}(\beta^*(\mu))=\mu\)
and
\begin{equation}
\label{eq:main_scale_ordering}
\beta^*(\mu)
<
b_L^*(\mu).
\end{equation}
For each family, the least-private pair over \(\X\) is \((\xlo,\xhi)\).
\end{proposition}

\noindent
\textit{Proof.} Appendix~\ref{app:tight_gdp_indices}.

Because \(\beta^*(\mu)<b_L^*(\mu)\), the common-\(\mu\) comparison is not an
equal-scale comparison, so Proposition~\ref{prop:laplace_logistic_blackwell}
does not determine the ordering under this calibration.

% ..................................................................
\subsubsection{Endpoint Dominance under Common Tight \texorpdfstring{\(\mu\)}{mu}}
\label{subsubsec:endpoint}
% ..................................................................

Restrict the report space to the binary endpoint state space
\(\{\xlo,\xhi\}\), and let
\(K_{\mathrm{Lap},\mu}^{\mathrm{end}}\) and
\(K_{\mathrm{Log},\mu}^{\mathrm{end}}\) denote the tightly calibrated endpoint
experiments.

\begin{theorem}[Endpoint dominance at common \(\mu\)-GDP]
\label{thm:endpoint_roc}
For every \(\mu>0\),
\begin{equation}
\label{eq:endpoint_tradeoff_order}
\mathcal T_{\mathrm{Log},\mu}^{\mathrm{end}}(\alpha)
\le
\mathcal T_{\mathrm{Lap},\mu}^{\mathrm{end}}(\alpha)
\qquad
\forall\alpha\in[0,1],
\end{equation}
equivalently
\(K_{\mathrm{Log},\mu}^{\mathrm{end}}\succeq_{\mathrm B}
K_{\mathrm{Lap},\mu}^{\mathrm{end}}\). Therefore, for every endpoint-state
decision problem,
\begin{equation}
\label{eq:endpoint_welfare_order}
W_{\mathrm{end}}^{\mathrm{red}}
\!\left(
R;K_{\mathrm{Log},\mu}^{\mathrm{end}}
\right)
\ge
W_{\mathrm{end}}^{\mathrm{red}}
\!\left(
R;K_{\mathrm{Lap},\mu}^{\mathrm{end}}
\right)
\end{equation}
whenever the revenue target is feasible under the Laplace endpoint
experiment.
\end{theorem}

\noindent
\textit{Proof.} Appendix~\ref{app:endpoint_common_mu}.

\begin{remark}[Scope of the endpoint theorem]
\label{rmk:endpoint_scope}
Theorem~\ref{thm:endpoint_roc} applies only to the binary subexperiment
generated by the maximally separated reports \(\xlo\) and \(\xhi\). It does
not establish a Blackwell ordering between the tightly calibrated Laplace and
logistic channels on the full continuous report space.
\end{remark}

\begin{remark}[Gaussian benchmark]
\label{rmk:channel_scope}
At \(\sigma^*(\mu)=\Delta_{\X}/\mu\), the Gaussian endpoint experiment has
trade-off function \(G_\mu\) at every testing level. The tightly calibrated
Laplace and logistic endpoint trade-off functions satisfy
\(\mathcal T(\alpha)\ge G_\mu(\alpha)\), with equality at
\(\alpha\in\{0,\Phi(-\mu/2),1\}\) and strict inequality at every other
interior testing level.
\end{remark}

Outside the exact comparisons above, no universal ranking of the Gaussian,
Laplace, and logistic channels is asserted. The numerical comparisons in
Section~\ref{subsec:sim_channel_welfare} are therefore specification dependent
unless supported by Proposition~\ref{prop:laplace_logistic_blackwell} or
Theorem~\ref{thm:endpoint_roc}.

\section{Numerical Illustrations}
\label{sec:numerics}

This section~\footnote{\noindent\textbf{Note:}
 The code for all numerical results, together with the
additional material, is available at
\url{https://github.com/BehroozMoosavi/Codes/tree/main/private_mechanism_design}.} connects the main results to simulation. The four experiments
trace, in turn, the price the planner pays for privacy in the responsiveness of provision,
the knife-edge at which budget pressure switches a public good from always
funded to never funded, the difficulty of financing provision when the marginal
agent is break-even, and the channel-choice problem a regulator faces once
privacy is budgeted in an economically meaningful unit. Throughout we use the
three additive channels of Table~\ref{tab:channel_families}---Gaussian
\(Y=X+\eta^{G}\), Laplace \(Y=X+\eta^{L}\), and logistic
\(Y=X+\eta^{\mathrm{Log}}\)---and a uniform type distribution, with the support
chosen in each experiment to probe the relevant regime. Posterior moments
\(\widehat x(y)=\E[X\mid Y=y]\) and \(\widehat J(y)=\E[\virt(X)\mid Y=y]\) are
evaluated by quadrature against the prior; implementation frequencies and
revenue are estimated by Monte Carlo with common random numbers across channels,
so that channel comparisons are paired. The priors and privacy calibrations used
in each experiment are recorded in Appendix~\ref{app:numerics}.

% ---------------------------------------------------------------------
\subsection{The Price of Privacy in the Provision Rule}
\label{subsec:sim_scores}
% ---------------------------------------------------------------------
The optimal rule of Theorem~\ref{thm:optimal_allocation} provides the public
good when the aggregate posterior score \(G_\lambda=\sum_i S_i\) is positive
(under the no-mass condition the threshold set is null, so the value assigned on
\(\{G_\lambda=0\}\) is immaterial). Economically, each per-agent score
\[
  S_i(y_i,\lambda)
  =
  \underbrace{\widehat x_i(y_i)}_{\text{expected value}}
  +
  \lambda\,
  \underbrace{\widehat J_i(y_i)}_{\text{revenue net of rents}},
  \qquad
  \virt(x)=x-\frac{1-F(x)}{f(x)},
\]
is the planner's case for building attributable to agent \(i\): the agent's
expected value for the good plus the revenue that can be extracted from that
agent net of the information rent the agent commands, the two priced against each
other by the shadow value \(\lambda\ge0\) of the revenue requirement. The
privacy channel enters only through how sharply this case responds to the
agent's noisy report.

Figure~\ref{fig:scores} plots \(S(\cdot,\lambda)\) for each channel on a uniform
prior. The Gaussian and logistic scores rise strictly with the signal: every
report, however extreme, still moves the provision decision, so the good remains
responsive to the private intensity of an agent's preference. Under Laplace noise, the
score varies over the interior region \(y\in(\xlo,\xhi)\) and is exactly flat on
the global tails \(y\le\xlo\) and \(y\ge\xhi\): the privacy mechanism has pooled all sufficiently
enthusiastic (or sufficiently reluctant) reporters into a single
indistinguishable group, because the bounded prior makes the Laplace likelihood
ratio constant in the tails (Remark~\ref{rmk:strict_weak_mlrp}). For those
agents, expressing more enthusiasm buys no additional weight in the decision. The
economic content of the channel is thus an elasticity: privacy noise governs how
strongly public provision can track private values, and a coarser channel caps
that responsiveness precisely for the types whose preferences are most extreme.

\begin{figure}[t]
  \centering
  \includegraphics[width=0.8\linewidth]{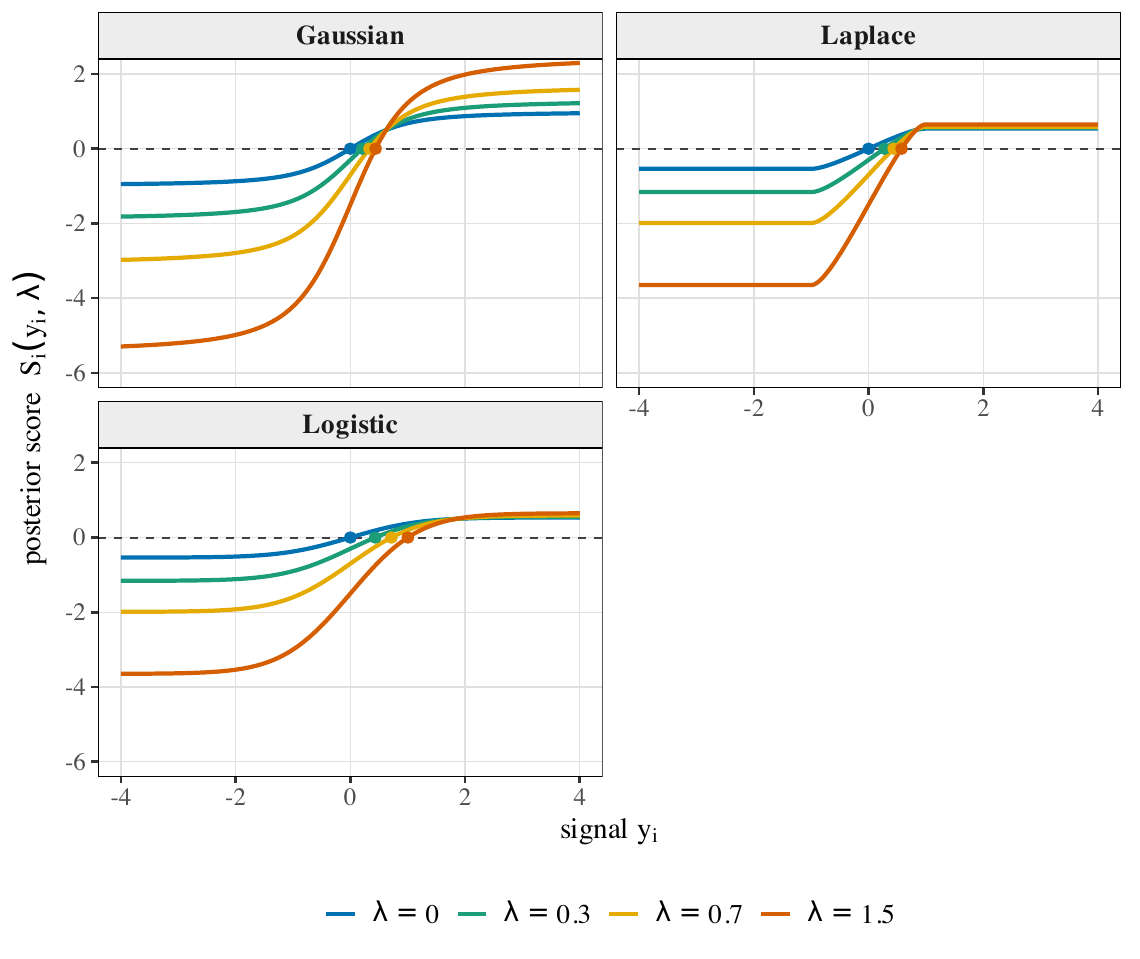}
  \caption{The per-agent score \(S(y,\lambda)=\widehat x(y)+\lambda\widehat
  J(y)\)---the planner's case for building attributable to one agent---against
  that agent's report \(y\), for the three channels and several values of the
  revenue price \(\lambda\). Gaussian and logistic respond to every report
  (strict MLRP), so provision stays sensitive to private value; Laplace is flat
  outside \([\xlo,\xhi]\) (tail pooling), so privacy caps how far provision can
  respond for the most extreme types. This score is the object thresholded by
  the optimal rule of Theorem~\ref{thm:optimal_allocation}.}
  \label{fig:scores}
\end{figure}

% ---------------------------------------------------------------------
\subsection{When Budget Pressure Switches the Good Off}
\label{subsec:sim_knife_edge}
% ---------------------------------------------------------------------
The multiplier \(\lambda\) is the weight the planner places on revenue relative
to welfare, and Proposition~\ref{prop:asymptotic_revenue} shows that whether the
good is funded is governed by the per-agent mean score
\[
  \mu_S(\lambda)=\E[X]+\lambda\,\xlo,
\]
the average net social value the planner perceives once revenue is priced at
\(\lambda\). The decisive primitive is \(\xlo=\E[\widehat J_i(Y_i)]\)
(Lemma~\ref{lem:lower_endpoint}), the lower endpoint of the type distribution,
equal to the mean posterior virtual value, i.e.\ the average virtual surplus
per agent. The mean score
changes sign at the critical price
\[
  \lambda_{\mathrm{crit}}
  =
  -\,\E[X]/\xlo,
\]
where \(\mu_S(\lambda_{\mathrm{crit}})=0\). In the numerical design \(X\sim\Unif[-0.5,1.5]\), so
\(\xlo<0\) and \(\mu_S'(\lambda)=\xlo<0\); hence \(\mu_S(\lambda)>0\) for
\(\lambda<\lambda_{\mathrm{crit}}\) and \(\mu_S(\lambda)<0\) for \(\lambda>\lambda_{\mathrm{crit}}\), and by the
law of large numbers the good is funded with probability tending to one for
\(\lambda<\lambda_{\mathrm{crit}}\) and to zero for \(\lambda>\lambda_{\mathrm{crit}}\) (the two directions
exchange when \(\xlo>0\)). The economic reading is stark: as the population grows
the public good is essentially always provided while the budget is loose and
essentially never provided once it tightens past \(\lambda_{\mathrm{crit}}\), with no smooth
interior trade-off. A small increase in fiscal pressure can flip a project from
``always built'' to ``never built.''

How privacy bears on this cliff is the content of the central-limit refinement.
A separate local calculation around \(\lambda_{\mathrm{crit}}\), applied to sequences satisfying
\(\sqrt n\,(\lambda_n-\lambda_{\mathrm{crit}})=O(1)\), gives a funding probability of
approximately \(\Phi\!\big(\sqrt n\,\mu_S(\lambda)/\sigma_S(\lambda_{\mathrm{crit}})\big)\),
where \(\sigma_S^2(\lambda_{\mathrm{crit}})=\Var\!\big(S_i(Y_i,\lambda_{\mathrm{crit}})\big)\). Writing the
standardized budget pressure
\[
  u
  =
  -\,\frac{\sqrt n\,\mu_S(\lambda)}{\sigma_S(\lambda_{\mathrm{crit}})}
  =
  -\,\frac{\sqrt n\,\xlo\,(\lambda-\lambda_{\mathrm{crit}})}{\sigma_S(\lambda_{\mathrm{crit}})},
\]
which absorbs the constant slope \(\mu_S'(\lambda_{\mathrm{crit}})=\xlo\), the probability
converges to \(\Phi(-u)\) for every channel. Figure~\ref{fig:knife} confirms
this: the curves for \(n\in\{25,\dots,500\}\) collapse onto the common Gaussian
limit. For nondegenerate channels satisfying
\(\sigma_S^2(\lambda_{\mathrm{crit}})>0\), the privacy channel enters the
local transition through the scale
\(\sigma_S(\lambda_{\mathrm{crit}})\). It stretches or compresses the
transition, sharpening or blurring how precisely the planner can locate the
funding margin, but it does not move the mean-score boundary
\(\lambda_{\mathrm{crit}}\). The boundary is channel invariant because
posterior expectations preserve the unconditional means
\(\E[X]\) and \(\E[\widehat J(Y)]=\xlo\). Whether a public good is fundable is set by the mean posterior virtual value
\(\xlo\), not by the privacy technology.

\begin{figure}[t]
  \centering
  \includegraphics[width=\linewidth]{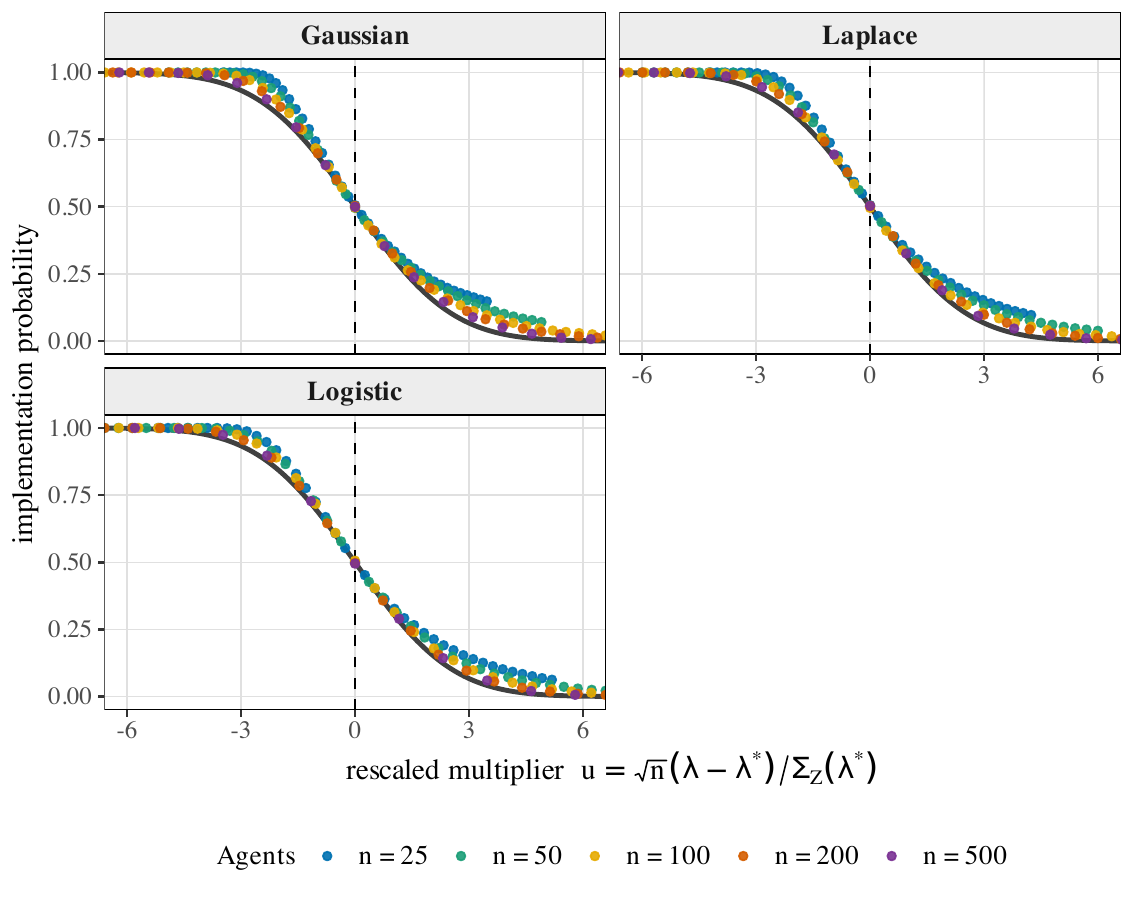}
  \caption{Probability the good is funded, against standardized budget pressure
  \(u=-\sqrt n\,\xlo\,(\lambda-\lambda_{\mathrm{crit}})/\sigma_S(\lambda_{\mathrm{crit}})=-\sqrt n\,
  \mu_S(\lambda)/\sigma_S(\lambda_{\mathrm{crit}})\). Across population sizes and channels the
  probabilities collapse onto \(\Phi(-u)\) (solid). Provision flips from
  near-certain to near-impossible across the knife-edge \(\lambda_{\mathrm{crit}}\) (the
  crossing at \(u=0\)); the privacy channel enters only through the scale
  \(\sigma_S\), which sets how sharp the transition is, not where it sits.}
  \label{fig:knife}
\end{figure}

% ---------------------------------------------------------------------
\subsection{Financing a Good at the Break-Even Margin}
\label{subsec:sim_sqrtn}
% ---------------------------------------------------------------------
Theorem~\ref{thm:three_regimes_known} classifies optimal reduced-form revenue
\(R_n^{*,\mathrm{red}}\) by the sign of the mean posterior virtual value \(\xlo\)
of the previous subsection: revenue grows linearly with the market when
\(\xlo>0\), decays when \(\xlo<0\), and---in the boundary case \(\xlo=0\)---grows
only at square-root order
\[
  R_n^{*,\mathrm{red}}
  \sim
  \frac{\sigma_J}{\sqrt{2\pi}}\,\sqrt n,
  \qquad
  \sigma_J^2=\Var\!\big(\widehat J_i(Y_i)\big).
\]
This boundary is the economically critical case: the lower endpoint equals
zero, so average posterior virtual value has zero drift and the mechanism
collects no systematic revenue. Per-capita revenue
\(R_n^{*,\mathrm{red}}/n\to0\), so a privacy-respecting public-good mechanism
cannot raise funds that scale with the market; what it collects comes entirely
from the statistical luck of drawing a realized virtual surplus above its mean.
The magnitude of that luck---and hence the financing the planner can hope
for---is set by \(\sigma_J\), the informativeness of the privacy channel about
types: a sharper channel raises the constant \(\sigma_J/\sqrt{2\pi}\), but no
channel changes the \(\sqrt n\) rate. Figure~\ref{fig:sqrtn} verifies both the
rate and the constant, with the normalized revenue
\(R_n^{*,\mathrm{red}}/\sqrt n\) flattening onto each channel's level
\(\sigma_J/\sqrt{2\pi}\). The message is that financing a public good from a
population whose marginal member breaks even is fundamentally hard: privacy
technology can buy a better constant, but not a better order of magnitude. We
display \(\xlo=0\) because it gives the sharpest quantitative check; the funded
(\(\xlo>0\)) regime is governed by its leading term \(n\xlo\), and the starved
(\(\xlo<0\)) regime follows from the same theorem and is illustrated in
Appendix~\ref{app:numerics}.

\begin{figure}[t]
  \centering
  \includegraphics[width=0.74\linewidth]{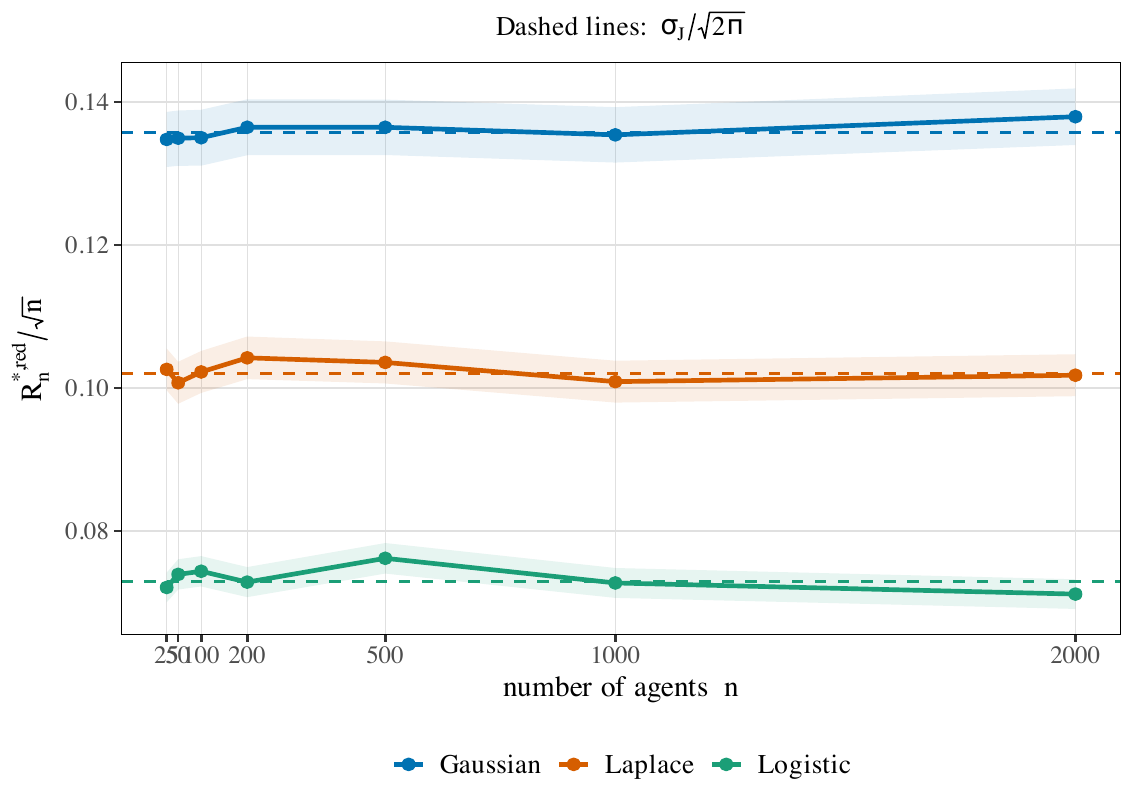}
  \caption{Normalized optimal revenue \(R_n^{*,\mathrm{red}}/\sqrt n\) at the
  break-even boundary \(\xlo=0\), where per-capita revenue vanishes and what is
  raised is pure statistical fluctuation in virtual surplus. Each channel
  converges to its constant \(\sigma_J/\sqrt{2\pi}\) (dashed): the channel's
  informativeness \(\sigma_J\) fixes how much can be financed, but the
  \(\sqrt n\) rate is common (Theorem~\ref{thm:three_regimes_known}).}
  \label{fig:sqrtn}
\end{figure}

% ---------------------------------------------------------------------
\subsection{Choosing a Privacy Technology: the Welfare Reversal}
\label{subsec:sim_channel_welfare}
% ---------------------------------------------------------------------
Which privacy technology should a planner adopt to maximize welfare?
Figure~\ref{fig:welfare} shows the answer is not a property of the noise
distribution alone---it depends on the privacy-accounting standard used to hold
privacy fixed. Raw noise scale, pure \(\epsilon\)-LDP, and \(\mu\)-GDP are
distinct standards. For the Laplace--logistic comparison, a common noise scale is equivalent to
a common tight pure-\(\epsilon\)-LDP calibration. Under this calibration,
the Laplace channel Blackwell-dominates the logistic channel
(Proposition~\ref{prop:laplace_logistic_blackwell}), so
\[
  W^{\mathrm{red}}(R;K_{\mathrm{Lap},s})
  \ge
  W^{\mathrm{red}}(R;K_{\mathrm{Log},s}).
\]
Recalibrating both channels under a common tight \(\mu\)-GDP guarantee changes
their relative scales and reverses the information ordering on the maximally
separated endpoint experiment \(\{\xlo,\xhi\}\): there the logistic channel takes
the smaller scale and is the more informative experiment, so the welfare ranking
flips (Theorem~\ref{thm:endpoint_roc}).

The left panel of Figure~\ref{fig:welfare} plots the equal-scale welfare gap
\((W_{\mathrm{Lap}}-W_{\mathrm{Log}})/W_{\mathrm{FB}}\ge0\) against the
pure-\(\epsilon\) budget; the right panel plots the common-\(\mu\) endpoint
welfare gap \((W_{\mathrm{Log}}-W_{\mathrm{Lap}})/W_{\mathrm{FB}}\ge0\) against
the high-type probability, for three privacy levels \(\mu\). In each calibration
the dominant channel's welfare advantage is nonnegative, and the dominant
channel flips from Laplace to logistic as the accounting standard changes from
scale to GDP. The comparative reading is that no privacy channel is universally
welfare-best; the welfare ranking depends on the privacy-accounting standard
under which the channels are compared. We stress the scope. The reversal is
established for the maximally separated endpoint experiment \(\{\xlo,\xhi\}\)
and does not lift to a Blackwell ranking over the full continuum of types
(Remark~\ref{rmk:endpoint_scope}); a planner therefore cannot conclude that
logistic dominates Laplace under GDP for an arbitrary preference distribution.
The sharp ROC-level form of the endpoint comparison, and the near-zero
continuous-type gap that delimits its scope, are reported in
Appendix~\ref{app:numerics}.

\begin{figure}[t]
  \centering
  \begin{minipage}{0.5\linewidth}\centering
    \includegraphics[width=\linewidth]{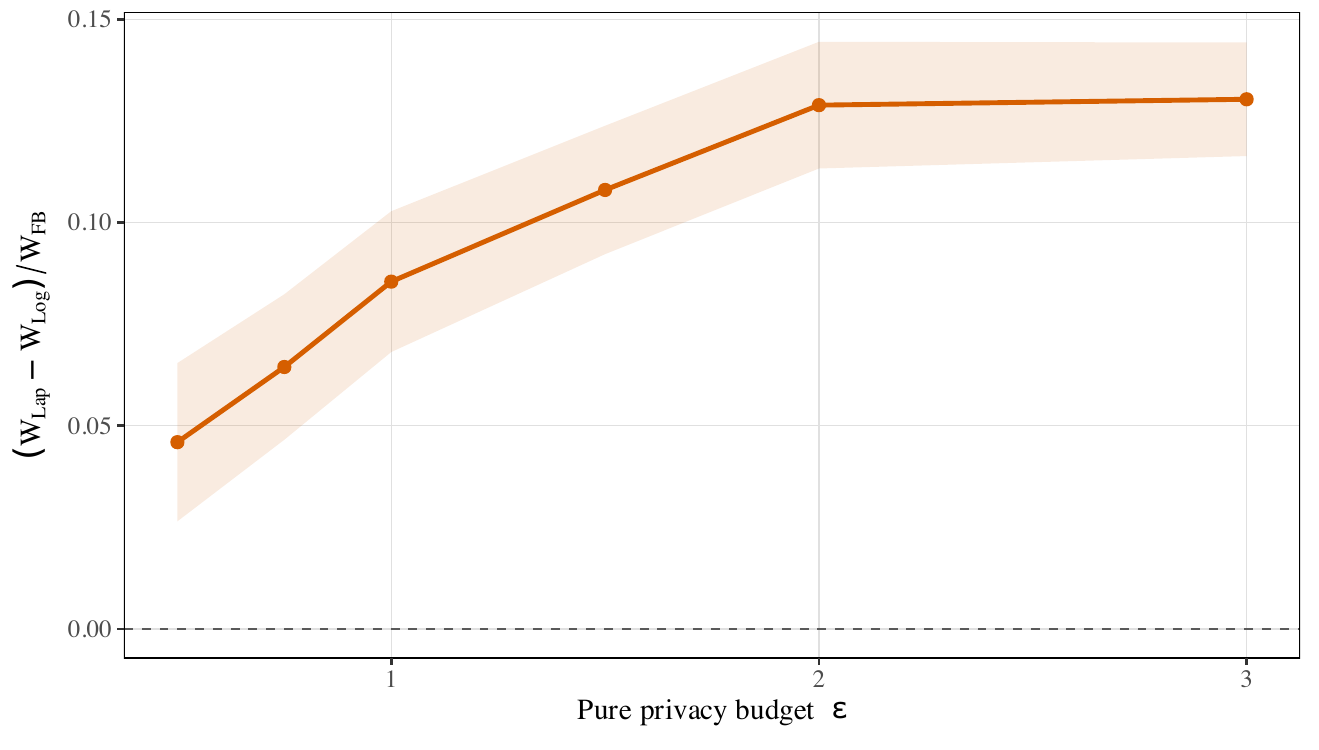}
  \end{minipage}\hfill
  \begin{minipage}{0.49\linewidth}\centering
    \includegraphics[width=\linewidth]{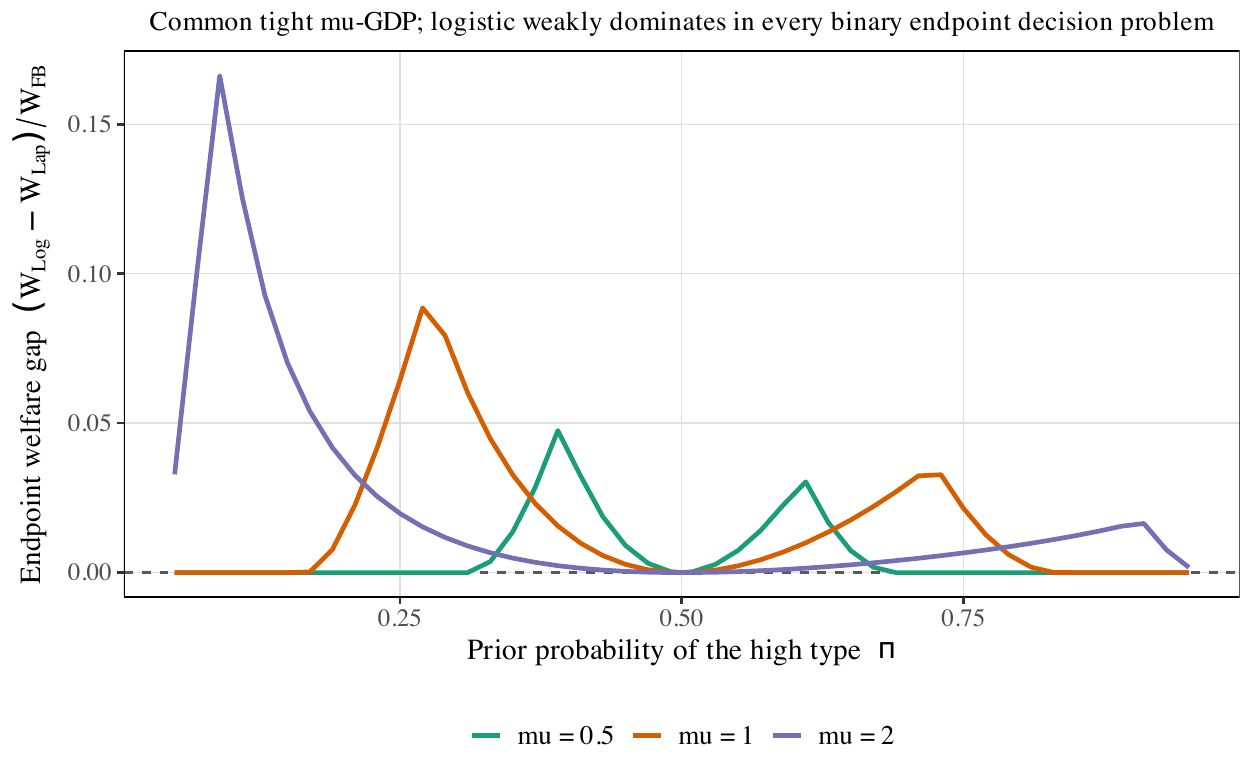}
  \end{minipage}
  \caption{The welfare-best channel depends on how privacy is budgeted. Left:
  under a common scale (pure \(\epsilon\)-LDP), Laplace holds a nonnegative
  welfare advantage \((W_{\mathrm{Lap}}-W_{\mathrm{Log}})/W_{\mathrm{FB}}\), the
  welfare form of the Blackwell dominance in
  Proposition~\ref{prop:laplace_logistic_blackwell}. Right: under a common
  \(\mu\)-GDP level, on the extreme-type experiment \(\{\xlo,\xhi\}\), logistic
  holds the advantage \((W_{\mathrm{Log}}-W_{\mathrm{Lap}})/W_{\mathrm{FB}}\)
  (Theorem~\ref{thm:endpoint_roc}). The dominant channel flips with the
  accounting standard; the effect is an endpoint statement
  (Remark~\ref{rmk:endpoint_scope}), not a ranking over all type profiles.}
  \label{fig:welfare}
\end{figure}

\section{Conclusion}
\label{sec:conclusion}

We have studied binary public-good provision when a planner observes agents
only through a fixed local-privacy channel. The privacy channel reshapes the
information on which optimal provision is based: the optimal reduced-form rule
thresholds an aggregate posterior score combining, for each agent, the
posterior expected valuation and posterior virtual value. The channel enters
through the responsiveness of this score to the agent's privatized report.

\medskip

Three economic conclusions emerge. First, the allocation the planner would
like to choose and the payments needed to implement it are distinct objects.
Implementing the envelope payments through transfers measurable in the
privatized signals is an inverse problem governed by a Fredholm range condition:
completeness gives uniqueness when a solution exists but does not guarantee
existence. The Online Supplement complements this exact characterization with
approximate implementations carrying explicit incentive, participation, and
revenue bounds.

\medskip

Second, maximum reduced-form revenue follows three asymptotic regimes governed
by the lower endpoint of the valuation distribution. It is asymptotically
linear, of square-root order, or exponentially small according as that endpoint
is positive, zero, or negative. The privacy channel affects the quantitative
difficulty within these regimes through the square-root constant and the
large-deviation rate. For a fixed posterior-score rule, provision also changes
sharply from asymptotically certain to asymptotically negligible at a
mean-score boundary.

\medskip

Third, welfare comparisons between privacy channels depend on how privacy is
calibrated. For Laplace and logistic noise, a common scale is equivalent to the common
tight pure-\(\epsilon\) calibration considered here, and Laplace
Blackwell-dominates logistic noise on the full report space. Under a common
tight \(\mu\)-GDP calibration, the ordering reverses for the maximally
separated binary endpoint experiment. Thus the preferred channel may depend on
the privacy standard used for comparison, while the endpoint reversal does not
imply a general ordering over the full continuous type space.

\medskip

The limits of these results also point to several directions for future work.
The allocation and revenue results are reduced-form statements; connecting them
to fully implemented mechanisms requires verifying the relevant transfer range
condition. Characterizing which economically important envelope-payment
schedules belong to the range of smoothing channel operators remains open. The
hierarchical extension relies on a high-level coordinatewise-monotonicity
condition that need not follow automatically once a latent common component
links agents' values. Identifying primitive conditions for this monotonicity,
and studying the scope for surplus extraction under privatized correlated
information, are natural next steps.

\medskip

The channel comparisons leave a further boundary. Equal-scale
Laplace-over-logistic dominance holds on the full report space, but the
common-\(\mu\) reversal is established only for the endpoint experiment.
Whether comparable welfare orderings hold on continuous type spaces, and how
the conclusions change for nonadditive or nonsmoothing channels, remain open.
More broadly, the framework suggests treating the privacy standard itself as
part of the design problem: when welfare comparisons depend on the calibration
used to hold privacy fixed, privacy accounting is not merely a technical
convention.

\bibliographystyle{apalike}
\bibliography{reference}

\newpage

\appendix

\section{Privacy and Channel Derivations}
\label{app:background_derivations}

This appendix develops the privacy, likelihood-ratio, posterior, and
comparison-of-experiments calculations used in the main text. We first
derive the complete Neyman--Pearson trade-off functions for the Gaussian,
Laplace, and logistic location channels. We then obtain their pure-LDP and
tight GDP calibrations, establish the ordering of the common-\(\mu\) scales,
describe posterior pooling under Laplace noise, derive the Gaussian posterior
formula used in Example~\ref{ex:uniform_gaussian_main}, and prove the two
Blackwell comparisons stated in Section~\ref{subsec:channel_ordering}.

Throughout,
\[
\Delta_{\X}
=
\xhi-\xlo.
\]
For two reports \(r<r'\), write
\[
d=r'-r\in[0,\Delta_{\X}].
\]
For an additive location family, translation invariance reduces the binary
experiment generated by \((r,r')\) to the experiment generated by locations
\((0,d)\).

For probability measures \(P\) and \(Q\) on a common measurable space,
recall the testing trade-off function
\begin{equation}
\label{eq:app_tradeoff_definition}
\mathcal T(P,Q)(\alpha)
=
\inf_{\psi}
\left\{
1-\E_Q[\psi]:
\E_P[\psi]\le\alpha
\right\},
\qquad
\alpha\in[0,1],
\end{equation}
where the infimum is taken over measurable randomized tests
\(\psi:\Y\to[0,1]\). The Neyman--Pearson lemma implies that, for each pair
of simple hypotheses, the lower envelope is generated by likelihood-ratio
tests.

% ---------------------------------------------------------------------
\subsection{Gaussian Location Channel}
\label{app:gaussian_channel}
% ---------------------------------------------------------------------

Let
\[
P_{0,\sigma}
=
\Normal(0,\sigma^2),
\qquad
P_{d,\sigma}
=
\Normal(d,\sigma^2).
\]
The log-likelihood ratio is
\begin{equation}
\label{eq:app_gaussian_log_lr}
\log
\frac{dP_{d,\sigma}}{dP_{0,\sigma}}(y)
=
\frac{d}{\sigma^2}y
-
\frac{d^2}{2\sigma^2},
\end{equation}
which is strictly increasing in \(y\). Hence the most powerful test of
\(P_{0,\sigma}\) against \(P_{d,\sigma}\) rejects for sufficiently large
values of \(Y\).

For a level-\(\alpha\) test, choose \(c_\alpha\) such that
\[
\alpha
=
P_{0,\sigma}(Y\ge c_\alpha)
=
1-\Phi(c_\alpha/\sigma).
\]
Thus
\[
c_\alpha
=
\sigma\Phi^{-1}(1-\alpha).
\]
The corresponding type-II error is
\[
P_{d,\sigma}(Y<c_\alpha)
=
\Phi\!\left(
\frac{c_\alpha-d}{\sigma}
\right).
\]
Therefore,
\begin{equation}
\label{eq:app_gaussian_tradeoff}
\mathcal T_{G,d,\sigma}(\alpha)
=
\Phi\!\left(
\Phi^{-1}(1-\alpha)-\frac{d}{\sigma}
\right)
=
G_{d/\sigma}(\alpha).
\end{equation}

\begin{proposition}[Gaussian privacy frontier]
\label{prop:app_gaussian_privacy}
For the Gaussian additive channel with scale \(\sigma\),
\begin{equation}
\label{eq:app_gaussian_tight_mu}
\mu_G(\sigma)
=
\frac{\Delta_{\X}}{\sigma}.
\end{equation}
The channel satisfies strict MLRP but does not satisfy finite pure
\(\epsilon\)-LDP.
\end{proposition}

\begin{proof}
Equation~\eqref{eq:app_gaussian_tradeoff} shows that the binary experiment
generated by reports separated by \(d\) has GDP parameter \(d/\sigma\).
The least-private report pair has the maximal separation
\(d=\Delta_{\X}\), proving \eqref{eq:app_gaussian_tight_mu}.

For \(x_2>x_1\),
\begin{equation}
\label{eq:app_gaussian_mlr}
\frac{k_G(y\mid x_2)}{k_G(y\mid x_1)}
=
\exp\!\left\{
\frac{x_2-x_1}{\sigma^2}y
-
\frac{x_2^2-x_1^2}{2\sigma^2}
\right\},
\end{equation}
which is strictly increasing in \(y\). Hence strict MLRP holds.

The likelihood ratio in \eqref{eq:app_gaussian_mlr} is unbounded as
\(y\to\infty\), so no finite constant \(\epsilon\) can satisfy the
pointwise likelihood-ratio bound required by pure LDP.
\end{proof}

% ---------------------------------------------------------------------
\subsection{Laplace Location Channel}
\label{app:laplace_channel}
% ---------------------------------------------------------------------

Let
\[
P_{0,b}
=
\Lap(0,b),
\qquad
P_{d,b}
=
\Lap(d,b),
\qquad
\rho
=
\frac{d}{b}.
\]
The likelihood ratio is
\begin{equation}
\label{eq:app_laplace_lr}
\frac{k_L(y\mid d)}{k_L(y\mid0)}
=
\exp\!\left\{
\frac{|y|-|y-d|}{b}
\right\}.
\end{equation}
Equivalently,
\begin{equation}
\label{eq:app_laplace_lr_piecewise}
\frac{k_L(y\mid d)}{k_L(y\mid0)}
=
\begin{cases}
e^{-\rho},
&
y\le0,
\\[3pt]
e^{(2y-d)/b},
&
0<y<d,
\\[3pt]
e^\rho,
&
y\ge d.
\end{cases}
\end{equation}
The ratio is weakly increasing but constant in each outer tail. Thus the
Laplace channel satisfies weak, but not strict, MLRP.

\begin{proposition}[Laplace trade-off function]
\label{prop:app_laplace_tradeoff}
For \(\rho=d/b\),
\begin{equation}
\label{eq:app_laplace_tradeoff}
\mathcal T_{L,\rho}(\alpha)
=
\begin{cases}
1-e^\rho\alpha,
&
0\le\alpha\le \frac12e^{-\rho},
\\[8pt]
\dfrac{e^{-\rho}}{4\alpha},
&
\frac12e^{-\rho}\le\alpha\le\frac12,
\\[10pt]
e^{-\rho}(1-\alpha),
&
\frac12\le\alpha\le1.
\end{cases}
\end{equation}
Its unique interior fixed point is
\begin{equation}
\label{eq:app_laplace_fixed_point}
p_L(\rho)
=
\frac12e^{-\rho/2}.
\end{equation}
\end{proposition}

\begin{proof}
By \eqref{eq:app_laplace_lr_piecewise}, the likelihood ratio is constant on
\((-\infty,0]\), strictly increasing on \((0,d)\), and constant on
\([d,\infty)\). The Neyman--Pearson tests therefore proceed by first using
the right tail, then moving the rejection threshold through \((0,d)\), and
finally randomizing in the left flat-likelihood-ratio tail.

If the rejection threshold satisfies \(c\ge d\), then
\[
\alpha
=
P_{0,b}(Y\ge c)
=
\frac12e^{-c/b},
\]
while
\[
1-\mathcal T_{L,\rho}(\alpha)
=
P_{d,b}(Y\ge c)
=
\frac12e^{-(c-d)/b}
=
e^\rho\alpha.
\]
This gives the first branch.

If \(0<c<d\), then
\[
\alpha
=
\frac12e^{-c/b},
\]
and
\[
\mathcal T_{L,\rho}(\alpha)
=
P_{d,b}(Y<c)
=
\frac12e^{-(d-c)/b}.
\]
Since \(e^{c/b}=1/(2\alpha)\), this becomes
\[
\mathcal T_{L,\rho}(\alpha)
=
\frac{e^{-\rho}}{4\alpha},
\]
which gives the middle branch.

Reflection around the midpoint \(d/2\) exchanges the two hypotheses
\(P_{0,b}\) and \(P_{d,b}\). The resulting binary experiment is symmetric,
so its testing trade-off function is self-inverse:
\[
\mathcal T_{L,\rho}
\bigl(
\mathcal T_{L,\rho}(\alpha)
\bigr)
=
\alpha.
\]
Applying this identity to the first branch gives the final branch
\[
\mathcal T_{L,\rho}(\alpha)
=
e^{-\rho}(1-\alpha),
\qquad
\frac12\le\alpha\le1.
\]
Solving \(T(\alpha)=\alpha\) on the middle branch yields
\[
\alpha^2=\frac14e^{-\rho},
\]
and hence \eqref{eq:app_laplace_fixed_point}.
\end{proof}

\begin{corollary}[Laplace pure-LDP frontier]
\label{cor:app_laplace_pure_ldp}
The Laplace location channel with scale \(b\) satisfies tight pure
\(\epsilon\)-LDP with
\[
\epsilon
=
\frac{\Delta_{\X}}{b}.
\]
Equivalently,
\begin{equation}
\label{eq:app_laplace_epsilon_scale}
b_L^*(\epsilon)
=
\frac{\Delta_{\X}}{\epsilon}.
\end{equation}
\end{corollary}

\begin{proof}
Equation~\eqref{eq:app_laplace_lr_piecewise} gives
\[
e^{-d/b}
\le
\frac{k_L(y\mid r')}{k_L(y\mid r)}
\le
e^{d/b}.
\]
The worst report separation is \(d=\Delta_{\X}\), and the upper bound is
attained in the right tail.
\end{proof}

% ---------------------------------------------------------------------
\subsection{Logistic Location Channel}
\label{app:logistic_channel}
% ---------------------------------------------------------------------

The centered logistic distribution with scale \(\beta>0\) has cdf
\[
F_\beta(y)
=
\frac{1}{1+e^{-y/\beta}}
\]
and density
\begin{equation}
\label{eq:app_logistic_density}
g_\beta(y)
=
\frac{e^{-y/\beta}}
{\beta(1+e^{-y/\beta})^2}.
\end{equation}
Let
\[
P_{0,\beta}
=
\Logistic(0,\beta),
\qquad
P_{d,\beta}
=
\Logistic(d,\beta),
\qquad
\rho
=
\frac{d}{\beta}.
\]

\begin{proposition}[Logistic trade-off function]
\label{prop:app_logistic_tradeoff}
For \(\rho=d/\beta\),
\begin{equation}
\label{eq:app_logistic_tradeoff}
\mathcal T_{\mathrm{Log},\rho}(\alpha)
=
\frac{1-\alpha}
{1-\alpha+e^\rho\alpha}.
\end{equation}
Its unique interior fixed point is
\begin{equation}
\label{eq:app_logistic_fixed_point}
p_{\mathrm{Log}}(\rho)
=
\frac{1}{1+e^{\rho/2}}.
\end{equation}
\end{proposition}

\begin{proof}
The logistic location family has monotone likelihood ratio, so the
Neyman--Pearson test rejects for \(Y\ge c\). Under the null,
\[
\alpha
=
1-F_\beta(c).
\]
Hence
\[
e^{c/\beta}
=
\frac{1-\alpha}{\alpha}.
\]
Under the alternative, the type-II error is
\[
F_\beta(c-d)
=
\frac{1}
{1+\exp\{-(c-d)/\beta\}}.
\]
Substituting the expression for \(e^{c/\beta}\) gives
\[
F_\beta(c-d)
=
\frac{1-\alpha}
{1-\alpha+e^\rho\alpha}.
\]
Solving \(T(\alpha)=\alpha\) yields
\[
\frac{1-\alpha}{1-\alpha+e^\rho\alpha}
=
\alpha,
\]
whose unique interior solution is
\[
\alpha
=
\frac{1}{1+e^{\rho/2}}.
\]
\end{proof}

\begin{proposition}[Logistic MLRP and pure-LDP frontier]
\label{prop:app_logistic_privacy}
The logistic location channel satisfies strict MLRP. Its tight pure-LDP
parameter over \(\X\) is
\[
\epsilon
=
\frac{\Delta_{\X}}{\beta},
\]
or equivalently
\begin{equation}
\label{eq:app_logistic_epsilon_scale}
\beta^*(\epsilon)
=
\frac{\Delta_{\X}}{\epsilon}.
\end{equation}
\end{proposition}

\begin{proof}
Differentiating the centered log density gives
\[
\frac{d}{dy}\log g_\beta(y)
=
-\frac1\beta
+
\frac{2}{\beta(1+e^{y/\beta})}.
\]
This derivative is strictly decreasing. Therefore, for \(x_2>x_1\),
\[
\frac{d}{dy}
\log
\frac{k_{\mathrm{Log}}(y\mid x_2)}
{k_{\mathrm{Log}}(y\mid x_1)}
=
\frac{d}{dy}\log g_\beta(y-x_2)
-
\frac{d}{dy}\log g_\beta(y-x_1)
>
0.
\]
Hence strict MLRP holds.

Moreover,
\[
\left|
\frac{d}{dy}\log g_\beta(y)
\right|
\le
\frac1\beta.
\]
By the mean-value theorem,
\[
\left|
\log
\frac{k_{\mathrm{Log}}(y\mid r)}
{k_{\mathrm{Log}}(y\mid r')}
\right|
\le
\frac{|r-r'|}{\beta}.
\]
The bound is approached in the tails, so it is tight. Maximizing
\(|r-r'|\) over \(\X\) gives the result.
\end{proof}

% ---------------------------------------------------------------------
\subsection{Tight GDP Indices and Common-\texorpdfstring{\(\mu\)}{mu} Scales}
\label{app:tight_gdp_indices}
% ---------------------------------------------------------------------

For a trade-off function \(T\), define its Gaussian separation profile by
\begin{equation}
\label{eq:app_gaussian_separation}
D_T(\alpha)
=
-\Phi^{-1}(\alpha)
-
\Phi^{-1}(T(\alpha)),
\qquad
\alpha\in(0,1).
\end{equation}
Because
\[
G_\mu(\alpha)
=
\Phi\!\left(
-\Phi^{-1}(\alpha)-\mu
\right),
\]
the inequality
\[
T(\alpha)\ge G_\mu(\alpha)
\]
is equivalent to
\[
D_T(\alpha)\le\mu.
\]
Consequently, the tight GDP index of \(T\) is
\begin{equation}
\label{eq:app_tight_mu_functional}
\mu(T)
=
\sup_{\alpha\in(0,1)}D_T(\alpha).
\end{equation}

\begin{lemma}[Normal-quantile monotonicity]
\label{lem:app_normal_quantile_ratios}
Let
\[
I(u)
=
\varphiN\!\left(\Phi^{-1}(u)\right),
\qquad
u\in(0,1),
\]
and define
\[
A(u)
=
\frac{u}{I(u)},
\qquad
B(u)
=
\frac{u(1-u)}{I(u)}.
\]
Then \(A\) is strictly increasing on \((0,1)\). Moreover, \(B\) is
symmetric about \(1/2\) and strictly increasing on \((0,1/2)\).

These properties are closely related to classical monotonicity and
functional inequalities for the Gaussian Mills ratio; see, for example,
\citet{baricz2008mills}.
\end{lemma}

\begin{proof}
Write
\[
z=\Phi^{-1}(u).
\]
Since
\[
\frac{dz}{du}
=
\frac{1}{\varphiN(z)}
=
\frac{1}{I(u)}
\]
and
\[
\varphiN'(z)=-z\varphiN(z),
\]
the Gaussian isoperimetric function satisfies
\begin{equation}
\label{eq:app_gaussian_isoperimetric_derivatives}
I'(u)=-z,
\qquad
I''(u)=-\frac{1}{I(u)}.
\end{equation}

We first prove the monotonicity of \(A\). Differentiation gives
\[
A'(u)
=
\frac{I(u)-uI'(u)}{I(u)^2}
=
\frac{\varphiN(z)+uz}{\varphiN(z)^2}.
\]
If \(z\ge0\), the numerator is strictly positive. If \(z=-x<0\), then
\(x>0\), \(u=\Phi(-x)\), and the numerator becomes
\[
\varphiN(x)-x\Phi(-x).
\]
Moreover,
\[
x\Phi(-x)
=
x\int_x^\infty \varphiN(t)\,dt
<
\int_x^\infty t\varphiN(t)\,dt
=
\varphiN(x),
\]
where the strict inequality follows from \(t>x\) on a set of positive
Lebesgue measure and the final equality follows from
\(\varphiN'(t)=-t\varphiN(t)\). Thus \(A'(u)>0\) for every
\(u\in(0,1)\).

We next consider \(B\). Since
\[
\Phi^{-1}(1-u)=-\Phi^{-1}(u)
\]
and \(\varphiN\) is even,
\[
I(1-u)=I(u).
\]
Consequently,
\[
B(1-u)=B(u),
\]
so \(B\) is symmetric about \(1/2\).

It remains to prove that \(B\) is strictly increasing on \((0,1/2)\).
Let
\[
J(u)=u(1-u)
\]
and define
\[
W(u)
=
J'(u)I(u)-J(u)I'(u).
\]
Because
\[
B'(u)
=
\frac{J'(u)I(u)-J(u)I'(u)}{I(u)^2}
=
\frac{W(u)}{I(u)^2},
\]
it is enough to prove
\[
W(u)>0,
\qquad
0<u<\frac12.
\]

Introduce the auxiliary function
\[
H(u)
=
2I(u)^2-J(u).
\]
Using \eqref{eq:app_gaussian_isoperimetric_derivatives}, we obtain
\[
H'(u)
=
4I(u)I'(u)-(1-2u)
\]
and
\begin{align}
H''(u)
&=
4\left\{
I'(u)^2+I(u)I''(u)
\right\}+2
\notag\\
&=
4\left\{
\Phi^{-1}(u)^2-1
\right\}+2
\notag\\
&=
4\Phi^{-1}(u)^2-2.
\label{eq:app_H_second_derivative}
\end{align}
Set
\[
u_0
=
\Phi\!\left(-\frac{1}{\sqrt2}\right).
\]
Equation~\eqref{eq:app_H_second_derivative} shows that \(H'\) is strictly
increasing on \((0,u_0)\) and strictly decreasing on
\((u_0,1/2)\).

As \(u\downarrow0\),
\[
I(u)\longrightarrow0
\]
and
\[
u\left|\Phi^{-1}(u)\right|
\longrightarrow0.
\]
Indeed, writing \(u=\Phi(-x)\) with \(x\to\infty\), the inequality already
proved above gives
\[
0\le x\Phi(-x)<\varphiN(x)\longrightarrow0.
\]
It follows that
\[
\lim_{u\downarrow0}H(u)=0
\]
and
\[
\lim_{u\downarrow0}H'(u)=-1.
\]
At the midpoint,
\[
H\!\left(\frac12\right)
=
2\varphiN(0)^2-\frac14
=
\frac1\pi-\frac14
>
0,
\]
and
\[
H'\!\left(\frac12\right)=0.
\]

Because \(H'\) is strictly decreasing on \((u_0,1/2)\) and ends at zero,
\[
H'(u)>0,
\qquad
u_0\le u<\frac12.
\]
Since \(H'\) starts at \(-1\) and is strictly increasing on
\((0,u_0)\), it has exactly one zero in \((0,u_0)\). Thus \(H\) first
decreases and then strictly increases on \((0,1/2)\). Since
\[
H(0+)=0,
\qquad
H\!\left(\frac12\right)>0,
\]
there exists a unique
\[
u_*\in\left(0,\frac12\right)
\]
such that
\[
H(u)<0
\quad\text{for }0<u<u_*,
\qquad
H(u)>0
\quad\text{for }u_*<u<\frac12.
\]

Differentiating \(W\) and using
\eqref{eq:app_gaussian_isoperimetric_derivatives} yields
\begin{align}
W'(u)
&=
J''(u)I(u)-J(u)I''(u)
\notag\\
&=
-2I(u)+\frac{J(u)}{I(u)}
\notag\\
&=
-\frac{H(u)}{I(u)}.
\label{eq:app_W_derivative}
\end{align}
Therefore,
\[
W'(u)>0
\quad\text{for }0<u<u_*,
\qquad
W'(u)<0
\quad\text{for }u_*<u<\frac12.
\]

Finally,
\[
\lim_{u\downarrow0}W(u)=0.
\]
Indeed, \(J'(u)I(u)\to0\), while
\[
|J(u)I'(u)|
=
u(1-u)\left|\Phi^{-1}(u)\right|
\longrightarrow0.
\]
Also,
\[
W\!\left(\frac12\right)=0,
\]
because
\[
J'\!\left(\frac12\right)=0
\qquad\text{and}\qquad
I'\!\left(\frac12\right)=0.
\]
Thus \(W\) increases strictly from zero on \((0,u_*)\) and then decreases
strictly back to zero on \((u_*,1/2)\). Hence
\[
W(u)>0,
\qquad
0<u<\frac12.
\]
Since \(I(u)>0\), it follows that
\[
B'(u)
=
\frac{W(u)}{I(u)^2}
>
0,
\qquad
0<u<\frac12.
\]
Therefore \(B\) is strictly increasing on \((0,1/2)\).
\end{proof}

\begin{lemma}[Tight endpoint GDP indices]
\label{lem:app_tight_endpoint_indices}
For the Laplace and logistic endpoint experiments,
\begin{align}
\mu_L(\rho)
&=
-2\Phi^{-1}
\left(
\frac12e^{-\rho/2}
\right),
\label{eq:app_laplace_mu_index}
\\
\mu_{\mathrm{Log}}(\rho)
&=
-2\Phi^{-1}
\left(
\frac{1}{1+e^{\rho/2}}
\right).
\label{eq:app_logistic_mu_index}
\end{align}
\end{lemma}

\begin{proof}
Both trade-off functions are self-inverse:
\[
T(T(\alpha))
=
\alpha.
\]
Hence
\[
D_T(T(\alpha))
=
D_T(\alpha).
\]
It is therefore enough to show that \(D_T\) increases up to the unique
interior fixed point.

At differentiability points,
\begin{equation}
\label{eq:app_D_derivative}
D_T'(\alpha)
=
-\frac{1}
{\varphiN(\Phi^{-1}(\alpha))}
-
\frac{T'(\alpha)}
{\varphiN(\Phi^{-1}(T(\alpha)))}.
\end{equation}

For the Laplace trade-off function, on the first branch define
\[
u=e^\rho\alpha=1-T_L(\alpha).
\]
Since \(u>\alpha\), normal symmetry and
Lemma~\ref{lem:app_normal_quantile_ratios} give
\[
D_{T_L}'(\alpha)
=
\frac1\alpha
\bigl(
A(u)-A(\alpha)
\bigr)
>
0.
\]
On the middle branch,
\[
T_L(\alpha)
=
\frac{p_L(\rho)^2}{\alpha},
\qquad
-T_L'(\alpha)
=
\frac{T_L(\alpha)}{\alpha},
\]
so
\[
D_{T_L}'(\alpha)
=
\frac1\alpha
\bigl(
A(T_L(\alpha))-A(\alpha)
\bigr)
>
0
\]
before the fixed point. Self-inversion then gives strict decrease after the
fixed point. Therefore the supremum occurs at \(p_L(\rho)\), giving
\eqref{eq:app_laplace_mu_index}.

For the logistic trade-off function,
\[
-T_{\mathrm{Log}}'(\alpha)
=
\frac{
T_{\mathrm{Log}}(\alpha)
\{1-T_{\mathrm{Log}}(\alpha)\}
}{
\alpha(1-\alpha)
}.
\]
Substitution into \eqref{eq:app_D_derivative} yields
\[
D_{T_{\mathrm{Log}}}'(\alpha)
=
\frac{
B(T_{\mathrm{Log}}(\alpha))-B(\alpha)
}{
\alpha(1-\alpha)
}.
\]
Before the fixed point,
\[
T_{\mathrm{Log}}(\alpha)>\alpha.
\]
Using the symmetry and strict increase of \(B\) on \((0,1/2)\), the
numerator is positive. Hence the supremum occurs at the fixed point
\(p_{\mathrm{Log}}(\rho)\), proving
\eqref{eq:app_logistic_mu_index}.
\end{proof}

Solving \(\mu_L(\rho)=\mu\) gives
\begin{equation}
\label{eq:app_laplace_rho_mu}
\rho_L^*(\mu)
=
-2\log\!\bigl(2\Phi(-\mu/2)\bigr),
\end{equation}
while solving \(\mu_{\mathrm{Log}}(\rho)=\mu\) gives
\begin{equation}
\label{eq:app_logistic_rho_mu}
\rho_{\mathrm{Log}}^*(\mu)
=
2\log\!\left(
\frac{\Phi(\mu/2)}
{\Phi(-\mu/2)}
\right).
\end{equation}

\begin{lemma}[Ordering of common-\(\mu\) scales]
\label{lem:app_scale_order}
For every \(\mu>0\),
\[
\beta^*(\mu)
<
b_L^*(\mu).
\]
\end{lemma}

\begin{proof}
Let
\[
p=\Phi(-\mu/2)\in(0,1/2).
\]
Then
\[
\rho_L^*(\mu)
=
-2\log(2p)
\]
and
\[
\rho_{\mathrm{Log}}^*(\mu)
=
2\log\!\left(
\frac{1-p}{p}
\right).
\]
Therefore,
\[
\rho_{\mathrm{Log}}^*(\mu)
-
\rho_L^*(\mu)
=
2\log\!\bigl(2(1-p)\bigr)
>
0.
\]
Since the physical scale equals \(\Delta_{\X}\) divided by the standardized
separation, the logistic scale is smaller.
\end{proof}

\begin{lemma}[Endpoint pair is least private]
\label{lem:app_endpoint_least_private}
For the Gaussian, Laplace, and logistic location families, the trade-off
function is pointwise nonincreasing in the report separation
\[
d=|r-r'|.
\]
Consequently, the least-private report pair over \(\X\) is
\((\xlo,\xhi)\).
\end{lemma}

\begin{proof}
For the Gaussian family, the trade-off function is \(G_{d/\sigma}\), which
decreases pointwise in \(d\). For the Laplace and logistic families, the
trade-off functions depend on \(d\) only through the standardized separation
\(\rho=d/s\), and each is pointwise decreasing in \(\rho\). The maximal
separation over \(\X\) is \(\Delta_{\X}\).
\end{proof}

\begin{proof}[Proof of Proposition~\ref{prop:common_mu_calibration}]
For the Gaussian family,
\[
\mu_G(\sigma)
=
\frac{\Delta_{\X}}{\sigma},
\]
so
\[
\sigma^*(\mu)
=
\frac{\Delta_{\X}}{\mu}.
\]
For the Laplace and logistic families, combine
Lemma~\ref{lem:app_tight_endpoint_indices} with
\eqref{eq:app_laplace_rho_mu} and
\eqref{eq:app_logistic_rho_mu}. Because scale equals
\(\Delta_{\X}/\rho\), this gives
\[
b_L^*(\mu)
=
\frac{\Delta_{\X}}
{-2\log(2\Phi(-\mu/2))}
\]
and
\[
\beta^*(\mu)
=
\frac{\Delta_{\X}}
{2\log(\Phi(\mu/2)/\Phi(-\mu/2))}.
\]
Lemma~\ref{lem:app_scale_order} gives the scale ordering, while
Lemma~\ref{lem:app_endpoint_least_private} identifies the least-private
report pair.
\end{proof}

% ---------------------------------------------------------------------
\subsection{Laplace Posterior Pooling}
\label{app:laplace_pooling}
% ---------------------------------------------------------------------

Suppose
\[
X\in[\xlo,\xhi],
\qquad
Y=X+\eta,
\qquad
\eta\sim\Lap(0,b).
\]

\begin{proposition}[Exact tail pooling under Laplace noise]
\label{prop:app_laplace_tail_pooling}
For every \(y\le\xlo\),
\begin{equation}
\label{eq:app_laplace_left_posterior}
\pi(dx\mid y)
=
\frac{
e^{-x/b}f(x)\,dx
}{
\int_{\X}e^{-z/b}f(z)\,dz
}.
\end{equation}
For every \(y\ge\xhi\),
\begin{equation}
\label{eq:app_laplace_right_posterior}
\pi(dx\mid y)
=
\frac{
e^{x/b}f(x)\,dx
}{
\int_{\X}e^{z/b}f(z)\,dz
}.
\end{equation}
Thus every posterior expectation is constant on each global signal tail.
\end{proposition}

\begin{proof}
If \(y\le\xlo\), then \(y\le x\) for every \(x\in\X\), and hence
\[
k_L(y\mid x)
=
\frac{1}{2b}e^{-(x-y)/b}
=
\frac{e^{y/b}}{2b}e^{-x/b}.
\]
The factor depending on \(y\) cancels in Bayes' formula, yielding
\eqref{eq:app_laplace_left_posterior}.

If \(y\ge\xhi\), then \(y\ge x\) for every \(x\in\X\), so
\[
k_L(y\mid x)
=
\frac{1}{2b}e^{-(y-x)/b}
=
\frac{e^{-y/b}}{2b}e^{x/b}.
\]
Again, the \(y\)-dependent factor cancels, yielding
\eqref{eq:app_laplace_right_posterior}.
\end{proof}

\begin{corollary}[Atoms in Laplace posterior scores]
\label{cor:app_laplace_score_atoms}
Let \(h:\X\to\R\) be integrable, and define
\[
M_h(y)
=
\E[h(X)\mid Y=y].
\]
Then \(M_h(Y)\) has an atom at each distinct pooled-tail value whose
corresponding signal tail has positive probability.
\end{corollary}

\begin{proof}
Proposition~\ref{prop:app_laplace_tail_pooling} shows that \(M_h(y)\) is
constant on \((-\infty,\xlo]\) and on \([\xhi,\infty)\). Each of these
signal events has positive probability under Laplace noise.
\end{proof}

% ---------------------------------------------------------------------
\subsection{Uniform Prior with Gaussian Noise}
\label{app:uniform_gaussian}
% ---------------------------------------------------------------------

Suppose
\[
X\sim\Unif[-1,1],
\qquad
Y=X+\eta,
\qquad
\eta\sim\Normal(0,\sigma^2).
\]
Conditional on \(Y=y\), \(X\) has a \(\Normal(y,\sigma^2)\) law truncated
to \([-1,1]\). Define
\[
a(y)
=
\frac{-1-y}{\sigma},
\qquad
b(y)
=
\frac{1-y}{\sigma}.
\]

\begin{proposition}[Gaussian posterior mean under a uniform prior]
\label{prop:app_uniform_gaussian_mean}
The posterior mean is
\begin{equation}
\label{eq:uniform_gaussian_posterior_mean}
h(y,\sigma)
=
y
+
\sigma
\frac{
\varphiN(a(y))-\varphiN(b(y))
}{
\Phi(b(y))-\Phi(a(y))
}.
\end{equation}
Moreover, \(h(\cdot,\sigma)\) is strictly increasing.
\end{proposition}

\begin{proof}
The first claim is the standard mean formula for a truncated normal
distribution. Strict monotonicity follows from strict MLRP of the Gaussian
location family and Lemma~\ref{lem:mlrp_fosd}.
\end{proof}

For \(X\sim\Unif[-1,1]\),
\[
F(x)=\frac{x+1}{2},
\qquad
f(x)=\frac12,
\]
so
\[
\virt(x)
=
x-\frac{1-F(x)}{f(x)}
=
2x-1.
\]
Therefore,
\begin{equation}
\label{eq:uniform_gaussian_posterior_virtual}
\widehat J(y)
=
2h(y,\sigma)-1,
\end{equation}
and
\begin{equation}
\label{eq:uniform_gaussian_posterior_score_appendix}
S(y,\lambda)
=
(1+2\lambda)h(y,\sigma)-\lambda.
\end{equation}

% ---------------------------------------------------------------------
\subsection{Equal-Scale Laplace--Logistic Ordering}
\label{app:equal_scale_blackwell}
% ---------------------------------------------------------------------

The characteristic functions of centered Laplace and logistic random
variables with scale \(s\) are
\[
\phi_{\mathrm{Lap},s}(t)
=
\frac{1}{1+s^2t^2}
\]
and
\[
\phi_{\mathrm{Log},s}(t)
=
\frac{\pi st}{\sinh(\pi st)}.
\]
Euler's product formula gives
\[
\frac{\sinh(\pi z)}{\pi z}
=
\prod_{k=1}^{\infty}
\left(
1+\frac{z^2}{k^2}
\right).
\]
Hence
\begin{equation}
\label{eq:app_logistic_cf_factorization}
\phi_{\mathrm{Log},s}(t)
=
\frac{1}{1+s^2t^2}
\prod_{k=2}^{\infty}
\left(
1+\frac{s^2t^2}{k^2}
\right)^{-1}.
\end{equation}

\begin{lemma}[Logistic noise as Laplace noise plus independent noise]
\label{lem:app_logistic_laplace_convolution}
There exists a centered random variable \(U_s\), independent of
\(L_s\sim\Lap(0,s)\), such that
\[
L_s+U_s
\sim
\Logistic(0,s).
\]
\end{lemma}

\begin{proof}
Let
\[
L_k\sim\Lap(0,s/k),
\qquad
k\ge2,
\]
be mutually independent. Since
\[
\sum_{k=2}^{\infty}\Var(L_k)
=
2s^2\sum_{k=2}^{\infty}\frac1{k^2}
<
\infty,
\]
the series
\[
U_s
=
\sum_{k=2}^{\infty}L_k
\]
converges in \(L^2\). Its characteristic function is
\[
\phi_{U_s}(t)
=
\prod_{k=2}^{\infty}
\left(
1+\frac{s^2t^2}{k^2}
\right)^{-1}.
\]
Multiplying by \(\phi_{\mathrm{Lap},s}(t)\) and using
\eqref{eq:app_logistic_cf_factorization} gives the logistic characteristic
function.
\end{proof}

\begin{proof}[Proof of Proposition~\ref{prop:laplace_logistic_blackwell}]
Let \(L_s\sim\Lap(0,s)\) be independent of the random variable \(U_s\) in
Lemma~\ref{lem:app_logistic_laplace_convolution}. A Laplace signal has the
form
\[
Y_L=x+L_s.
\]
Adding the report-independent noise \(U_s\) gives
\[
Y_L+U_s
=
x+L_s+U_s
\sim
x+\Logistic(0,s).
\]
Thus the logistic experiment is obtained by garbling the Laplace experiment.
Therefore,
\[
K_{\mathrm{Lap},s}
\succeq_{\mathrm B}
K_{\mathrm{Log},s}.
\]
The welfare conclusion follows from
Proposition~\ref{prop:blackwell_monotone}.
\end{proof}

% ---------------------------------------------------------------------
\subsection{Endpoint Ordering under Common Tight \texorpdfstring{\(\mu\)}{mu}}
\label{app:endpoint_common_mu}
% ---------------------------------------------------------------------

Fix \(\mu>0\), and define
\[
p
=
\Phi(-\mu/2)
\in(0,1/2).
\]
Under common tight-\(\mu\) calibration,
\[
e^{-\rho_L^*}
=
4p^2,
\qquad
e^{\rho_{\mathrm{Log}}^*}
=
\left(
\frac{1-p}{p}
\right)^2.
\]
Thus the calibrated endpoint trade-off functions are
\begin{equation}
\label{eq:app_calibrated_laplace_tradeoff}
T_L(\alpha)
=
\begin{cases}
1-\dfrac{\alpha}{4p^2},
&
0\le\alpha\le2p^2,
\\[8pt]
\dfrac{p^2}{\alpha},
&
2p^2\le\alpha\le\frac12,
\\[8pt]
4p^2(1-\alpha),
&
\frac12\le\alpha\le1,
\end{cases}
\end{equation}
and
\begin{equation}
\label{eq:app_calibrated_logistic_tradeoff}
T_{\mathrm{Log}}(\alpha)
=
\frac{1-\alpha}
{
1-\alpha+
\left(
\frac{1-p}{p}
\right)^2\alpha
}.
\end{equation}

\begin{lemma}[Endpoint trade-off ordering]
\label{lem:app_endpoint_tradeoff_order}
For every \(p\in(0,1/2)\),
\[
T_{\mathrm{Log}}(\alpha)
\le
T_L(\alpha)
\qquad
\forall\alpha\in[0,1].
\]
Equality holds exactly at
\[
\alpha\in\{0,p,1\}.
\]
\end{lemma}

\begin{proof}
Set
\[
B
=
\frac{(1-p)^2}{p^2}-1
=
\frac{1-2p}{p^2}.
\]
Then
\[
T_{\mathrm{Log}}(\alpha)
=
\frac{1-\alpha}{1+B\alpha}.
\]

For \(0\le\alpha\le2p^2\),
\[
T_L(\alpha)-T_{\mathrm{Log}}(\alpha)
=
\frac{
\alpha(1-2p)
\left(
3-2p-\alpha/p^2
\right)
}{
4p^2(1+B\alpha)
}
\ge0.
\]

For \(2p^2\le\alpha\le1/2\),
\[
T_L(\alpha)-T_{\mathrm{Log}}(\alpha)
=
\frac{
(\alpha-p)^2
}{
\alpha(1+B\alpha)
}
\ge0.
\]

For \(1/2\le\alpha\le1\),
\[
T_L(\alpha)-T_{\mathrm{Log}}(\alpha)
=
\frac{
(1-\alpha)(1-2p)
\left(
4\alpha-(1+2p)
\right)
}{
1+B\alpha
}
\ge0.
\]
The displayed factorizations identify the equality points.
\end{proof}

\begin{proof}[Proof of Theorem~\ref{thm:endpoint_roc}]
Lemma~\ref{lem:app_endpoint_tradeoff_order} gives pointwise ordering of the
complete Neyman--Pearson trade-off functions. For binary experiments, this
is equivalent to Blackwell dominance. Hence
\[
K_{\mathrm{Log},\mu}^{\mathrm{end}}
\succeq_{\mathrm B}
K_{\mathrm{Lap},\mu}^{\mathrm{end}}.
\]
The endpoint welfare conclusion follows by applying
Proposition~\ref{prop:blackwell_monotone} to the endpoint-state decision
problem.
\end{proof}

% =====================================================================
% APPENDIX B
% =====================================================================

\section{Incentive, Posterior, and Monotonicity Results}
\label{app:auxiliary}

This appendix collects the general one-dimensional incentive arguments and
posterior identities used throughout the paper. These results depend on the
reporting structure and the truthful signal law, but not on the particular
form of the optimal allocation.

% ---------------------------------------------------------------------
\subsection{Interim Incentive Compatibility}
% ---------------------------------------------------------------------

\begin{proposition}[Interim BIC and envelope representation]
\label{prop:app_interim_bic}
Let
\[
Q_i:\X\to[0,1],
\qquad
T_i:\X\to\R
\]
be measurable. The following statements are equivalent.

\begin{enumerate}[label=(\roman*),leftmargin=2em]

\item For all \(x,r\in\X\),
\[
xQ_i(x)-T_i(x)
\ge
xQ_i(r)-T_i(r).
\]

\item The function \(Q_i\) is weakly increasing and there exists a constant
\(U_i(\xlo)\) such that
\begin{align}
U_i(x)
&=
U_i(\xlo)
+
\int_{\xlo}^{x}Q_i(z)\,dz,
\label{eq:app_envelope}
\\
T_i(x)
&=
xQ_i(x)
-
\int_{\xlo}^{x}Q_i(z)\,dz
-
U_i(\xlo).
\label{eq:app_payment_identity}
\end{align}

\end{enumerate}

Under either condition, interim individual rationality is equivalent to
\[
U_i(\xlo)\ge0.
\]
For a fixed interim allocation \(Q_i\), expected payment is maximized by
setting
\[
U_i(\xlo)=0.
\]
\end{proposition}

\begin{proof}
Assume BIC. For \(x>x'\), incentive compatibility gives
\[
xQ_i(x)-T_i(x)
\ge
xQ_i(x')-T_i(x')
\]
and
\[
x'Q_i(x')-T_i(x')
\ge
x'Q_i(x)-T_i(x).
\]
Adding yields
\[
(x-x')
\bigl(
Q_i(x)-Q_i(x')
\bigr)
\ge0.
\]
Thus \(Q_i\) is weakly increasing.

Define truthful utility by
\[
U_i(x)
=
xQ_i(x)-T_i(x).
\]
The incentive inequalities imply
\[
(x-x')Q_i(x')
\le
U_i(x)-U_i(x')
\le
(x-x')Q_i(x)
\]
whenever \(x>x'\). The standard monotone-envelope argument gives
\[
U_i(x)-U_i(\xlo)
=
\int_{\xlo}^{x}Q_i(z)\,dz,
\]
which yields \eqref{eq:app_payment_identity} after solving for \(T_i(x)\).

Conversely, suppose \(Q_i\) is weakly increasing and the envelope identity
holds. For any \(x,r\in\X\),
\[
U_i(x)
-
\bigl(
xQ_i(r)-T_i(r)
\bigr)
=
\int_r^x
\bigl(
Q_i(z)-Q_i(r)
\bigr)\,dz.
\]
The right-hand side is nonnegative both when \(x\ge r\) and when \(x<r\),
because \(Q_i\) is weakly increasing. Hence truthful reporting is optimal.

Since \(Q_i\ge0\), the envelope formula implies that \(U_i\) is weakly
increasing, so individual rationality is equivalent to the lowest-type
condition. Finally, increasing \(U_i(\xlo)\) lowers every payment by the
same amount, proving the final claim.
\end{proof}

\begin{corollary}[Uniqueness of interim payments up to the lowest-type rent]
\label{cor:app_payment_uniqueness}
For a fixed weakly increasing interim allocation \(Q_i\), every BIC interim
payment schedule has the form
\[
T_i(x)
=
xQ_i(x)
-
\int_{\xlo}^{x}Q_i(z)\,dz
-
c_i
\]
for some constant \(c_i\). Interim IR requires \(c_i\ge0\).
\end{corollary}

% ---------------------------------------------------------------------
\subsection{Expected Payments}
% ---------------------------------------------------------------------

\begin{lemma}[Expected virtual-surplus identity]
\label{lem:expected_transfer_identity}
For every BIC interim allocation and payment pair,
\begin{equation}
\label{eq:app_expected_payment}
\E[T_i(X_i)]
=
\E[
Q_i(X_i)\virt(X_i)
]
-
U_i(\xlo).
\end{equation}
\end{lemma}

\begin{proof}
Integrating \eqref{eq:app_payment_identity} against \(f\) gives
\[
\begin{aligned}
\E[T_i(X_i)]
&=
\int_{\xlo}^{\xhi}
xQ_i(x)f(x)\,dx
\\
&\quad
-
\int_{\xlo}^{\xhi}
\left(
\int_{\xlo}^{x}Q_i(z)\,dz
\right)
f(x)\,dx
-
U_i(\xlo).
\end{aligned}
\]
By Tonelli's theorem,
\[
\int_{\xlo}^{\xhi}
\left(
\int_{\xlo}^{x}Q_i(z)\,dz
\right)
f(x)\,dx
=
\int_{\xlo}^{\xhi}
Q_i(z)\{1-F(z)\}\,dz.
\]
Therefore,
\[
\E[T_i(X_i)]
=
\int_{\xlo}^{\xhi}
Q_i(x)
\left[
x-\frac{1-F(x)}{f(x)}
\right]
f(x)\,dx
-
U_i(\xlo),
\]
which is \eqref{eq:app_expected_payment}.
\end{proof}

\begin{corollary}[Revenue-maximizing rent normalization]
\label{cor:app_zero_rent}
Among all BIC and interim-IR payment schedules implementing a fixed interim
allocation \(Q_i\), expected payment is maximized by
\[
U_i(\xlo)=0.
\]
\end{corollary}

% ---------------------------------------------------------------------
\subsection{Posterior Factorization and Reduced-Form Representations}
% ---------------------------------------------------------------------

\begin{lemma}[Posterior factorization]
\label{lem:app_posterior_factorization}
Under the independent known-prior model,
\[
\mathcal L(X\mid Y=y)
=
\bigotimes_{i=1}^{n}
\mathcal L(X_i\mid Y_i=y_i).
\]
Consequently,
\[
\E[X_i\mid Y]
=
\widehat x_i(Y_i),
\qquad
\E[\virt(X_i)\mid Y]
=
\widehat J_i(Y_i).
\]
\end{lemma}

\begin{proof}
The truthful joint density is
\[
f_{X,Y}(x,y)
=
\prod_{j=1}^{n}
f(x_j)k(y_j\mid x_j).
\]
The truthful signal density is
\[
f_Y(y)
=
\prod_{j=1}^{n}m(y_j).
\]
Dividing gives
\[
f_{X\mid Y}(x\mid y)
=
\prod_{j=1}^{n}
\frac{
f(x_j)k(y_j\mid x_j)
}{
m(y_j)
},
\]
which is the claimed factorization.
\end{proof}

\begin{lemma}[Posterior welfare and revenue representations]
\label{lem:posterior_representations}
For every bounded measurable allocation \(q\),
\begin{equation}
\label{eq:app_posterior_welfare}
\E\!\left[
q(Y)\sum_{i=1}^{n}X_i
\right]
=
\E\!\left[
q(Y)\sum_{i=1}^{n}\widehat x_i(Y_i)
\right].
\end{equation}
For every BIC signal-measurable mechanism,
\begin{equation}
\label{eq:app_posterior_revenue}
\E\!\left[
\sum_{i=1}^{n}t_i(Y)
\right]
=
\E\!\left[
q(Y)\sum_{i=1}^{n}\widehat J_i(Y_i)
\right]
-
\sum_{i=1}^{n}U_i(\xlo).
\end{equation}
\end{lemma}

\begin{proof}
For welfare, iterated expectations and
Lemma~\ref{lem:app_posterior_factorization} give
\[
\begin{aligned}
\E\!\left[
q(Y)\sum_iX_i
\right]
&=
\E\!\left[
q(Y)
\E\!\left[
\sum_iX_i
\mid Y
\right]
\right]
\\
&=
\E\!\left[
q(Y)\sum_i\widehat x_i(Y_i)
\right].
\end{aligned}
\]

For revenue,
Lemma~\ref{lem:expected_transfer_identity} gives
\[
\E[T_i(X_i)]
=
\E[
Q_i(X_i)\virt(X_i)
]
-
U_i(\xlo).
\]
By the definition of the interim allocation,
\[
\E[
Q_i(X_i)\virt(X_i)
]
=
\E[
q(Y)\virt(X_i)
].
\]
Applying iterated expectations yields
\[
\E[
q(Y)\virt(X_i)
]
=
\E[
q(Y)\widehat J_i(Y_i)
].
\]
Summing over agents proves \eqref{eq:app_posterior_revenue}.
\end{proof}

\begin{corollary}[Zero-rent reduced-form revenue]
\label{cor:app_zero_rent_revenue}
Under the normalization \(U_i(\xlo)=0\) for every \(i\),
\[
\E\!\left[
\sum_{i=1}^{n}t_i(Y)
\right]
=
\E\!\left[
q(Y)\sum_{i=1}^{n}\widehat J_i(Y_i)
\right].
\]
\end{corollary}

% ---------------------------------------------------------------------
\subsection{Posterior Monotonicity under MLRP}
% ---------------------------------------------------------------------

\begin{lemma}[MLRP and posterior stochastic order]
\label{lem:mlrp_fosd}
Suppose \(k(y\mid x)>0\) and the kernel satisfies MLRP. If \(y_2>y_1\),
then
\[
\mathcal L(X\mid Y=y_2)
\]
dominates
\[
\mathcal L(X\mid Y=y_1)
\]
in likelihood-ratio order and therefore in first-order stochastic order.
Hence, for every integrable weakly increasing function \(h\),
\[
y
\longmapsto
\E[h(X)\mid Y=y]
\]
is weakly increasing. If the posterior likelihood ratio is strictly increasing and \(h\) is
weakly increasing and nonconstant on a set receiving positive probability
under both posterior laws, then the posterior expectation inequality is
strict.
\end{lemma}

\begin{proof}
Bayes' formula gives
\[
\pi(x\mid y)
=
\frac{
k(y\mid x)f(x)
}{
\int_{\X}k(y\mid z)f(z)\,dz
}.
\]
Therefore,
\[
\frac{
\pi(x\mid y_2)
}{
\pi(x\mid y_1)
}
=
C(y_1,y_2)
\frac{
k(y_2\mid x)
}{
k(y_1\mid x)
},
\]
where \(C(y_1,y_2)>0\) does not depend on \(x\). By MLRP, the right-hand
side is weakly increasing in \(x\). Thus the posterior at \(y_2\) dominates
the posterior at \(y_1\) in likelihood-ratio order.
Likelihood-ratio order implies first-order stochastic dominance.
Consequently, the expectation of every integrable weakly increasing
function \(h\) is weakly larger under the posterior associated with
\(y_2\).
Under the stated strictness and common-positive-probability conditions,
strict likelihood-ratio ordering and nonconstancy of \(h\) imply a strict
expectation inequality.
\end{proof}

\begin{corollary}[Monotonicity of posterior scores]
\label{cor:app_score_monotonicity}
Under Assumptions~\ref{ass:regular} and \ref{ass:mlrp}, for every
\(\lambda\ge0\),
\[
y
\longmapsto
S_i(y,\lambda)
=
\E[
X_i+\lambda\virt(X_i)
\mid Y_i=y
]
\]
is weakly increasing.
\end{corollary}

\begin{proof}
Regularity implies that
\[
x
\longmapsto
x+\lambda\virt(x)
\]
is weakly increasing. Apply Lemma~\ref{lem:mlrp_fosd}.
\end{proof}

\begin{remark}[Weak MLRP, pooling, and ironing]
\label{rmk:weak_mlrp_no_ironing}
Weak MLRP may generate flat posterior-score regions. Such flatness produces
pooling and may generate atoms in the score distribution, but it does not
generate nonmonotonicity. Consequently, Myerson ironing is not required
unless the underlying generalized virtual value or another relevant
allocation index is nonmonotone.
\end{remark}

% ---------------------------------------------------------------------
\subsection{Lower-Endpoint Identity}
% ---------------------------------------------------------------------

\begin{lemma}[Mean virtual value]
\label{lem:lower_endpoint}
Under Assumption~\ref{ass:regular},
\[
\E[\virt(X)]
=
\xlo.
\]
Consequently,
\[
\E[\widehat J_i(Y_i)]
=
\xlo.
\]
\end{lemma}

\begin{proof}
By definition,
\[
\E[\virt(X)]
=
\E[X]
-
\int_{\xlo}^{\xhi}
\{1-F(x)\}\,dx.
\]
For a random variable supported on \([\xlo,\xhi]\),
\[
\E[X]
=
\xlo
+
\int_{\xlo}^{\xhi}
\{1-F(x)\}\,dx.
\]
Subtracting gives
\[
\E[\virt(X)]
=
\xlo.
\]
The posterior identity follows by iterated expectations:
\[
\E[\widehat J_i(Y_i)]
=
\E[
\E[\virt(X_i)\mid Y_i]
]
=
\E[\virt(X_i)].
\]
\end{proof}

\begin{corollary}[Mean posterior score]
\label{cor:app_mean_posterior_score}
For every \(\lambda\ge0\),
\[
\mu_S(\lambda)
=
\E[S_i(Y_i,\lambda)]
=
\E[X]+\lambda\xlo.
\]
\end{corollary}

% ---------------------------------------------------------------------
\subsection{Continuity of Interim Allocations}
% ---------------------------------------------------------------------

\begin{lemma}[Continuity of channel-induced interim allocations]
\label{lem:app_interim_continuity}
Under Assumption~\ref{ass:compact_spaces}, every bounded measurable
allocation \(q\) induces a continuous interim allocation \(Q_i^q\).
\end{lemma}

\begin{proof}
For \(x,x'\in\X\),
\[
\begin{aligned}
|Q_i^q(x')-Q_i^q(x)|
&\le
\int_{\Y^n}
q(y)
\left|
k(y_i\mid x')-k(y_i\mid x)
\right|
f_{Y_{-i}}(y_{-i})\,dy
\\
&\le
\int_{\Y}
\left|
k(y_i\mid x')-k(y_i\mid x)
\right|\,dy_i.
\end{aligned}
\]
The final term converges to zero by
\eqref{eq:l1_kernel_continuity}.
\end{proof}

% =====================================================================
% APPENDIX C
% =====================================================================

\section{Functional Analysis and Proofs of the Main Results}
\label{app:main_results_proofs}

This appendix establishes existence and strong duality for the reduced-form
problem, proves the posterior-score characterization, develops the Fredholm
implementation operator, proves the three revenue regimes and the fixed-score
asymptotics, provides the hierarchical extension, and establishes Blackwell
monotonicity.

% ---------------------------------------------------------------------
\subsection{Weak-\texorpdfstring{\(^*\)}{*} Compactness}
\label{app:functional_analysis}
% ---------------------------------------------------------------------

Let
\[
\nu(dy)
=
f_Y(y)\,dy
\]
denote the truthful signal law on \(\Y^n\). We identify allocations that
agree \(\nu\)-almost surely and equip \(L^\infty(\nu)\) with the weak-\(^*\)
topology
\[
\sigma(L^\infty(\nu),L^1(\nu)).
\]

Define the allocation cube
\begin{equation}
\label{eq:app_allocation_cube}
\mathcal A
=
\left\{
q\in L^\infty(\nu):
0\le q\le1
\quad
\nu\text{-a.e.}
\right\}.
\end{equation}

\begin{lemma}[Weak-\(^*\) compactness of the allocation cube]
\label{lem:app_allocation_cube_compact}
The set \(\mathcal A\) is convex and weak-\(^*\) compact.
\end{lemma}

\begin{proof}
The set \(\mathcal A\) lies in the closed unit ball of \(L^\infty(\nu)\),
which is weak-\(^*\) compact by the Banach--Alaoglu theorem. It remains to
show weak-\(^*\) closedness.

Let \(q_\gamma\in\mathcal A\) converge weak-\(^*\) to \(q\). For every
nonnegative \(g\in L^1(\nu)\),
\[
\int q_\gamma g\,d\nu
\ge0.
\]
Passing to the limit gives
\[
\int qg\,d\nu
\ge0,
\]
so \(q\ge0\) almost everywhere. Applying the same argument to
\(1-q_\gamma\) gives \(q\le1\) almost everywhere.
\end{proof}

For each \(i\) and \(x\in\X\), define
\[
\ell_{i,x}(y)
=
\frac{k(y_i\mid x)}
{m(y_i)}
\]
on the support of \(m\). Then
\[
Q_i^q(x)
=
\int_{\Y^n}
q(y)\ell_{i,x}(y)\,\nu(dy),
\]
and
\[
\|\ell_{i,x}\|_{L^1(\nu)}
=
1.
\]

\begin{lemma}[Weak-\(^*\) continuity of interim allocations]
\label{lem:app_interim_weakstar}
For fixed \(i\) and \(x\), the map
\[
q
\longmapsto
Q_i^q(x)
\]
is weak-\(^*\) continuous.
\end{lemma}

\begin{proof}
The map is the dual pairing with
\(\ell_{i,x}\in L^1(\nu)\).
\end{proof}

Define
\[
\mathcal Q^{\mathrm{mon}}
=
\left\{
q\in\mathcal A:
Q_i^q(x')\ge Q_i^q(x)
\text{ whenever }x'\ge x,
\ \forall i
\right\}.
\]

\begin{lemma}[Compactness of the monotone allocation class]
\label{lem:app_monotone_compact}
The set \(\mathcal Q^{\mathrm{mon}}\) is convex and weak-\(^*\) compact.
\end{lemma}

\begin{proof}
For each \(i\) and each \(x'\ge x\), the condition
\[
Q_i^q(x')-Q_i^q(x)
\ge0
\]
defines a weak-\(^*\) closed half-space by
Lemma~\ref{lem:app_interim_weakstar}. Therefore
\(\mathcal Q^{\mathrm{mon}}\) is weak-\(^*\) closed in the compact set
\(\mathcal A\). Convexity follows from linearity of \(Q_i^q\).
\end{proof}

Define
\[
w(y)
=
\sum_{i=1}^{n}\widehat x_i(y_i),
\qquad
v(y)
=
\sum_{i=1}^{n}\widehat J_i(y_i).
\]
Compactness of \(\X\) and boundedness of \(\virt\) imply
\[
w,v\in L^1(\nu).
\]

\begin{lemma}[Continuity of welfare and revenue]
\label{lem:app_objective_continuity}
The functionals
\[
\mathsf W(q)
=
\int q(y)w(y)\,\nu(dy)
\]
and
\[
\mathsf V(q)
=
\int q(y)v(y)\,\nu(dy)
\]
are weak-\(^*\) continuous.
\end{lemma}

\begin{proof}
Both are dual pairings with elements of \(L^1(\nu)\).
\end{proof}

\begin{proposition}[Existence of a reduced-form optimum]
\label{prop:app_reduced_form_existence}
If
\[
\mathcal F_R
=
\left\{
q\in\mathcal Q^{\mathrm{mon}}:
\mathsf V(q)\ge R
\right\}
\]
is nonempty, then the reduced-form problem
\eqref{eq:main_reduced_problem} has a solution.
\end{proposition}

\begin{proof}
The feasible set is weak-\(^*\) closed in the weak-\(^*\) compact set
\(\mathcal Q^{\mathrm{mon}}\), and is therefore weak-\(^*\) compact.
The weak-\(^*\) continuous functional \(\mathsf W\) attains its maximum.
\end{proof}

% ---------------------------------------------------------------------
\subsection{Supporting Multipliers and Strong Duality}
\label{app:multiplier_duality}
% ---------------------------------------------------------------------

Define the downward-comprehensive attainable set
\[
\mathcal C
=
\left\{
(r,z)\in\R^2:
\exists q\in\mathcal Q^{\mathrm{mon}}
\text{ with }
r\le\mathsf V(q),
\quad
z\le\mathsf W(q)
\right\}.
\]

\begin{lemma}[Attainable-set geometry]
\label{lem:app_attainable_set}
The set \(\mathcal C\) is convex, closed, and downward comprehensive.
\end{lemma}

\begin{proof}
Convexity follows from convexity of
\(\mathcal Q^{\mathrm{mon}}\) and linearity of \(\mathsf W\) and
\(\mathsf V\). Downward comprehensiveness is immediate.

Let
\[
(r_\gamma,z_\gamma)\to(r,z),
\qquad
(r_\gamma,z_\gamma)\in\mathcal C.
\]
Choose \(q_\gamma\in\mathcal Q^{\mathrm{mon}}\) satisfying
\[
r_\gamma\le\mathsf V(q_\gamma),
\qquad
z_\gamma\le\mathsf W(q_\gamma).
\]
By weak-\(^*\) compactness, a subnet converges to some
\(q\in\mathcal Q^{\mathrm{mon}}\). Continuity of \(\mathsf V\) and
\(\mathsf W\) gives
\[
r\le\mathsf V(q),
\qquad
z\le\mathsf W(q).
\]
Hence \((r,z)\in\mathcal C\).
\end{proof}

\begin{proposition}[Supporting multiplier and strong duality]
\label{prop:app_supporting_multiplier}
Suppose strict feasibility holds:
\[
\exists q^\circ\in\mathcal Q^{\mathrm{mon}}
\quad\text{such that}\quad
\mathsf V(q^\circ)>R.
\]
Then there exists a finite \(\lambda^*\ge0\) such that every primal optimum
maximizes
\[
q
\longmapsto
\mathsf W(q)+\lambda^*\mathsf V(q)
\]
over \(\mathcal Q^{\mathrm{mon}}\), and
\[
\lambda^*
\bigl(
\mathsf V(q^*)-R
\bigr)
=
0.
\]
Moreover,
\begin{equation}
\label{eq:app_strong_duality}
\sup_{\substack{q\in\mathcal Q^{\mathrm{mon}}\\
\mathsf V(q)\ge R}}
\mathsf W(q)
=
\inf_{\lambda\ge0}
\left[
\sup_{q\in\mathcal Q^{\mathrm{mon}}}
\{
\mathsf W(q)+\lambda\mathsf V(q)
\}
-\lambda R
\right].
\end{equation}
\end{proposition}

\begin{proof}
Let \(W^*(R)\) denote the primal value. Then
\[
(R,W^*(R))
\]
is a boundary point of the closed convex set \(\mathcal C\).Indeed, \((R,W^*(R))\in\mathcal C\), while
\((R,W^*(R)+\varepsilon)\notin\mathcal C\) for every
\(\varepsilon>0\) by the definition of \(W^*(R)\). The supporting-hyperplane theorem gives a nonzero normal vector
\((a,b)\in\R^2\). Downward comprehensiveness of \(\mathcal C\) implies
\(a,b\ge0\): if either component were negative, moving sufficiently far
downward in the corresponding coordinate would violate the supporting
inequality. Hence \((a,b)\in\R_+^2\), and
\[
ar+bz
\le
aR+bW^*(R)
\qquad
\forall(r,z)\in\mathcal C.
\]
Strict feasibility rules out \(b=0\), so \(b>0\). Define
\[
\lambda^*
=
\frac{a}{b}.
\]
Then
\[
\mathsf W(q)+\lambda^*\mathsf V(q)
\le
W^*(R)+\lambda^*R
\]
for all \(q\in\mathcal Q^{\mathrm{mon}}\).

For a primal optimum \(q^*\),
\[
\mathsf W(q^*)=W^*(R),
\qquad
\mathsf V(q^*)\ge R.
\]
Substituting into the support inequality gives complementary slackness.
Weak duality and evaluation at \(\lambda^*\) give
\eqref{eq:app_strong_duality}.
\end{proof}

% ---------------------------------------------------------------------
\subsection{Pointwise Lagrangian Maximization}
% ---------------------------------------------------------------------

\begin{lemma}[Pointwise maximization]
\label{lem:app_pointwise_max}
For fixed \(\lambda\ge0\), the maximizers of
\[
q
\longmapsto
\int q(y)G_\lambda(y)\,\nu(dy)
\]
over \(\mathcal A\) are precisely the measurable functions satisfying
\[
q(y)=1
\quad\text{on }\{G_\lambda>0\},
\]
\[
q(y)=0
\quad\text{on }\{G_\lambda<0\},
\]
with arbitrary values in \([0,1]\) on \(\{G_\lambda=0\}\).
\end{lemma}

\begin{proof}
For each \(y\), maximize
\[
a
\longmapsto
aG_\lambda(y)
\]
over \(a\in[0,1]\). Integrating the pointwise inequality proves the claim.
\end{proof}

\begin{remark}[Positive-mass tie sets]
\label{rmk:positive_mass_tie}
If
\[
\Prob(G_{\lambda^*}(Y)=0)>0,
\]
pointwise maximization alone does not determine the allocation on the
zero-score set. A constant tie probability preserves coordinatewise
monotonicity but need not move expected virtual surplus in the direction
required to meet the revenue constraint. Exact selection therefore requires
an additional measurable monotone tie rule.
\end{remark}

% ---------------------------------------------------------------------
\subsection{Proof of the Optimal-Allocation Theorem}
\label{app:proof_optimal_allocation}
% ---------------------------------------------------------------------

\begin{proof}[Proof of Theorem~\ref{thm:optimal_allocation}]
Existence follows from
Proposition~\ref{prop:app_reduced_form_existence}. The supporting multiplier
\(\lambda^*\) is supplied by
Proposition~\ref{prop:app_supporting_multiplier}.

Regularity implies that
\[
x
\longmapsto
x+\lambda^*\virt(x)
\]
is weakly increasing. By Lemma~\ref{lem:mlrp_fosd},
\[
S_i(y_i,\lambda^*)
=
\E[
X_i+\lambda^*\virt(X_i)
\mid
Y_i=y_i
]
\]
is weakly increasing in \(y_i\). Hence
\[
G_{\lambda^*}(y)
=
\sum_iS_i(y_i,\lambda^*)
\]
is coordinatewise weakly increasing.

Let \(\widetilde q\) be a primal optimum, whose existence follows from
Proposition~\ref{prop:app_reduced_form_existence}. By
Proposition~\ref{prop:app_supporting_multiplier}, \(\widetilde q\) also
maximizes the Lagrangian at \(\lambda^*\).

Lemma~\ref{lem:app_pointwise_max} implies that every Lagrangian maximizer
equals
\[
\ind\{G_{\lambda^*}>0\}
\]
outside the zero-score set. By
Assumption~\ref{ass:aggregate_score_atomless}, the zero-score set is
\(\nu\)-null, so the Lagrangian maximizer is unique up to
\(\nu\)-almost-everywhere equality. Consequently,
\[
\widetilde q(y)
=
\ind\{G_{\lambda^*}(y)>0\}
\qquad
\nu\text{-a.e.}
\]
Therefore the threshold rule
\[
q^*(y)
=
\ind\{G_{\lambda^*}(y)>0\}
\]
is itself primal feasible and optimal, not merely a maximizer of the unconstrained Lagrangian.

Because \(G_{\lambda^*}\) is coordinatewise weakly increasing, \(q^*\) is
coordinatewise weakly increasing. For each agent, MLRP then implies that the
interim allocation is weakly increasing in the report. Reduced-form BIC
follows from Proposition~\ref{prop:app_interim_bic}. Complementary slackness
follows from Proposition~\ref{prop:app_supporting_multiplier}.
\end{proof}

% ---------------------------------------------------------------------
\subsection{The Channel Operator and Fredholm Implementation}
\label{app:channel_operator}
% ---------------------------------------------------------------------

Let \(\mathcal H\) be the Banach space of admissible opponent-averaged
transfers specified in Assumption~\ref{ass:completeness}. Define
\[
\mathcal K:\mathcal H\to C(\X),
\qquad
(\mathcal Kg)(x)
=
\int_{\Y}g(y)k(y\mid x)\,dy.
\]

\begin{proposition}[Boundedness of the channel operator]
\label{prop:app_channel_operator}
The operator \(\mathcal K\) is linear and bounded. In particular,
\[
\|\mathcal Kg\|_\infty
\le
C_{\mathcal H}\|g\|_{\mathcal H}.
\]
Under \eqref{eq:l1_kernel_continuity},
\(\mathcal Kg\in C(\X)\).
\end{proposition}

\begin{proof}
Linearity is immediate. The norm bound follows from the domination condition
in Assumption~\ref{ass:completeness}. For bounded \(g\),
\[
|(\mathcal Kg)(x')-(\mathcal Kg)(x)|
\le
\|g\|_\infty
\int_{\Y}
|k(y\mid x')-k(y\mid x)|\,dy,
\]
which converges to zero by
\eqref{eq:l1_kernel_continuity}. For general \(g\in\mathcal H\), approximate
by bounded truncations and use boundedness of \(\mathcal K\).
\end{proof}

\begin{proposition}[Completeness and injectivity]
\label{prop:app_completeness_injectivity}
Assumption~\ref{ass:completeness} is equivalent to injectivity of
\(\mathcal K\), modulo \(m(y)\,dy\)-almost-everywhere equality.
\end{proposition}

\begin{proof}
The identity \(\mathcal Kg=0\) is exactly
\[
\int_{\Y}
g(y)k(y\mid x)\,dy
=
0
\qquad
\forall x\in\X.
\]
Thus the completeness condition is equivalent to a trivial kernel.
\end{proof}

\begin{theorem}[Range, uniqueness, and stability]
\label{thm:app_transfer_operator}
Let \(\tau\in C(\X)\).

\begin{enumerate}[label=(\roman*),leftmargin=2em]

\item A solution \(g\in\mathcal H\) of
\[
\mathcal Kg=\tau
\]
exists if and only if
\[
\tau\in\operatorname{Ran}(\mathcal K).
\]

\item Under Assumption~\ref{ass:completeness}, the solution is unique.

\item If \(\operatorname{Ran}(\mathcal K)\) is closed in \(C(\X)\), then
\[
\mathcal K^{-1}:
\operatorname{Ran}(\mathcal K)
\longrightarrow
\mathcal H
\]
is bounded.

\item If the range is not closed, the inverse is unbounded on its range.

\end{enumerate}
\end{theorem}

\begin{proof}
Parts (i) and (ii) follow from the definition of the range and injectivity.
If the range is closed, it is a Banach space under the inherited
\(C(\X)\)-norm. The bounded inverse theorem therefore applies.

Conversely, suppose the inverse were bounded while the range were not
closed. Then there would exist \(g_n\in\mathcal H\) and
\(\tau\notin\operatorname{Ran}(\mathcal K)\) such that
\[
\mathcal Kg_n\to\tau
\]
in \(C(\X)\). Boundedness of the inverse would imply that \(g_n\) is Cauchy
in \(\mathcal H\). Let \(g\) be its limit. Continuity of \(\mathcal K\)
would then give
\[
\mathcal Kg=\tau,
\]
a contradiction.
\end{proof}

\begin{proof}[Proof of Proposition~\ref{prop:fredholm}]
For a fixed allocation \(q\),
\[
Q_i^q(x)
=
\int_{\Y^n}
q(y)
k(y_i\mid x)
f_{Y_{-i}}(y_{-i})\,dy.
\]
Under zero lowest-type rent, the envelope payment is
\[
\tau_i^q(x)
=
xQ_i^q(x)
-
\int_{\xlo}^{x}Q_i^q(z)\,dz.
\]
Substituting the definition of \(Q_i^q\) and applying Fubini gives
\[
\begin{aligned}
\tau_i^q(x)
&=
\int_{\Y^n}
q(y)
\left[
xk(y_i\mid x)
-
\int_{\xlo}^{x}
k(y_i\mid z)\,dz
\right]
f_{Y_{-i}}(y_{-i})\,dy.
\end{aligned}
\]

An ex-post transfer \(t_i\) induces the interim payment
\[
T_i(x)
=
\int_{\Y^n}
t_i(y)
k(y_i\mid x)
f_{Y_{-i}}(y_{-i})\,dy.
\]
Define the truthful-law opponent average
\[
\bar t_i(y_i)
=
\E[t_i(Y)\mid Y_i=y_i]
=
\int_{\Y^{n-1}}
t_i(y_i,y_{-i})
f_{Y_{-i}}(y_{-i})\,dy_{-i}.
\]
Then
\[
T_i(x)
=
\int_{\Y}
\bar t_i(y_i)k(y_i\mid x)\,dy_i
=
(\mathcal K\bar t_i)(x).
\]
Thus implementation is equivalent to
\[
\mathcal K\bar t_i
=
\tau_i^q.
\]
Conversely, whenever \(\bar t_i\in\mathcal H\) solves this equation, an
ex-post implementation is obtained by setting
\[
t_i(y_i,y_{-i})=\bar t_i(y_i),
\]
provided this transfer satisfies the maintained integrability requirements.
Existence, uniqueness, and stability follow from
Theorem~\ref{thm:app_transfer_operator}.

Finally, two ex-post transfers implement the same interim payment schedule if
and only if their difference \(r_i\) satisfies
\[
\E[r_i(Y)\mid Y_i]
=
0
\]
under the truthful signal law.
\end{proof}

% ---------------------------------------------------------------------
\subsection{Proof of the Revenue Regimes}
\label{app:proof_three_regimes_known}
% ---------------------------------------------------------------------

Set
\[
Z_i
=
\widehat J_i(Y_i),
\qquad
S_n
=
\sum_{i=1}^{n}Z_i.
\]
The variables \(Z_i\) are i.i.d. and bounded, with
\[
\E[Z_i]
=
\xlo.
\]
Moreover,
\[
R_n^{*,\mathrm{red}}(K)
=
\E[(S_n)_+].
\]

\begin{proof}[Proof of Theorem~\ref{thm:three_regimes_known}]
Suppose first that \(\xlo>0\). Since
\[
(S_n)_+
=
S_n+(-S_n)_+,
\]
we have
\[
\E[(S_n)_+]
=
n\xlo+\E[(-S_n)_+].
\]
If \(|Z_i|\le M\), then
\[
0
\le
\E[(-S_n)_+]
\le
nM\Prob(S_n<0).
\]
Hoeffding's inequality gives constants \(c,C>0\) such that
\[
\Prob(S_n<0)
\le
Ce^{-cn}.
\]
Hence
\[
\left|
R_n^{*,\mathrm{red}}(K)-n\xlo
\right|
\le
CMne^{-cn},
\]
after absorbing constants.

Suppose next that \(\xlo=0\). The central limit theorem gives
\[
\frac{S_n}{\sigma_J\sqrt n}
\Rightarrow
N(0,1).
\]
Furthermore,
\[
\sup_n
\E\!\left[
\left(
\frac{S_n}{\sigma_J\sqrt n}
\right)^2
\right]
=
1.
\]
Thus the positive parts are uniformly integrable. Therefore,
\[
\frac{\E[(S_n)_+]}{\sigma_J\sqrt n}
\longrightarrow
\E[N_+]
=
\frac{1}{\sqrt{2\pi}}.
\]

Finally, suppose \(\xlo<0\). Since \(Z_i\) is bounded, its moment-generating
function is finite on all of \(\R\). Cram\'er's theorem applies with rate
function
\[
I_K(a)
=
\sup_{t\in\R}
\left\{
ta-\log\E[e^{tZ_i}]
\right\}.
\]

For the upper bound,
\[
\E[(S_n)_+]
\le
nM\Prob(S_n/n\ge0).
\]
Cram\'er's upper bound gives
\[
\limsup_{n\to\infty}
\frac1n
\log\E[(S_n)_+]
\le
-I_K(0).
\]

For the lower bound, fix \(\eta>0\) sufficiently small that
\((0,\eta)\) lies in the interior of the effective domain. Then
\[
\E[(S_n)_+]
\ge
\frac{n\eta}{2}
\Prob\!\left(
\frac{S_n}{n}\in(\eta/2,\eta)
\right).
\]
Cram\'er's lower bound gives
\[
\liminf_{n\to\infty}
\frac1n
\log\E[(S_n)_+]
\ge
-\inf_{a\in(\eta/2,\eta)}I_K(a).
\]
Because \(0\) lies in the interior of the effective domain, \(I_K\) is
continuous at zero. Letting \(\eta\downarrow0\) yields
\[
\liminf_{n\to\infty}
\frac1n
\log\E[(S_n)_+]
\ge
-I_K(0).
\]
Combining the bounds proves the logarithmic limit. Since
\[
0\neq\E[Z_i]=\xlo,
\]
the rate satisfies \(I_K(0)>0\).
\end{proof}

\begin{corollary}[Exponential representation in the negative-drift regime]
\label{cor:app_negative_drift_exponential}
If \(\xlo<0\), then
\[
R_n^{*,\mathrm{red}}(K)
=
\exp\!\left\{
-nI_K(0)+o(n)
\right\}.
\]
\end{corollary}

% ---------------------------------------------------------------------
\subsection{Proof for a Fixed Posterior-Score Rule}
\label{app:proof_asymptotic_revenue}
% ---------------------------------------------------------------------

\begin{proof}[Proof of Proposition~\ref{prop:asymptotic_revenue}]
Let
\[
A_n
=
\left\{
\sum_{i=1}^{n}W_i(\lambda)\ge0
\right\}.
\]
The variables \(W_i(\lambda)\) are i.i.d. and bounded, and
\[
\E[W_i(\lambda)]
=
\mu_S(\lambda)
>
0.
\]
Hoeffding's inequality therefore gives constants
\(c_\lambda,C_\lambda>0\) such that
\[
\Prob(A_n^c)
\le
C_\lambda e^{-c_\lambda n}.
\]

Since
\[
V_n
=
\sum_{i=1}^{n}\widehat J_i(Y_i)
\]
and \(|\widehat J_i(Y_i)|\le M\),
\[
\begin{aligned}
\E[V_n\ind_{A_n}]
&=
\E[V_n]
-
\E[V_n\ind_{A_n^c}]
\\
&=
n\xlo
+
\mathcal O
\left(
n\Prob(A_n^c)
\right)
\\
&=
n\xlo
+
\mathcal O
\left(
ne^{-c_\lambda n}
\right).
\end{aligned}
\]
\end{proof}

\subsection{Local Transition at the Mean-Score Boundary}
\label{app:proof_local_multiplier_transition}

\begin{proof}[Proof of Proposition~\ref{prop:local_multiplier_transition}]
Because
\[
S_i(Y_i,\lambda)
=
\widehat x_i(Y_i)
+
\lambda\widehat J_i(Y_i),
\]
we have
\[
S_i(Y_i,\lambda_n)
=
S_i(Y_i,\lambda_{\mathrm{crit}})
+
\frac{h}{\sqrt n}\widehat J_i(Y_i).
\]
Therefore,
\[
\frac{1}{\sqrt n}
\sum_{i=1}^{n}
S_i(Y_i,\lambda_n)
=
\frac{1}{\sqrt n}
\sum_{i=1}^{n}
S_i(Y_i,\lambda_{\mathrm{crit}})
+
h\frac1n
\sum_{i=1}^{n}
\widehat J_i(Y_i).
\]
At the critical multiplier,
\[
\E[
S_i(Y_i,\lambda_{\mathrm{crit}})
]
=
0.
\]
Hence the central limit theorem gives
\[
\frac{1}{\sqrt n}
\sum_{i=1}^{n}
S_i(Y_i,\lambda_{\mathrm{crit}})
\Rightarrow
\Normal(0,\sigma_{\mathrm{crit}}^2).
\]
By the law of large numbers and
Lemma~\ref{lem:lower_endpoint},
\[
\frac1n
\sum_{i=1}^{n}
\widehat J_i(Y_i)
\longrightarrow
\xlo
\qquad
\text{in probability}.
\]
Slutsky's theorem therefore yields
\[
\frac{1}{\sqrt n}
\sum_{i=1}^{n}
S_i(Y_i,\lambda_n)
\Rightarrow
\Normal(
h\xlo,
\sigma_{\mathrm{crit}}^2
).
\]
Since the limiting distribution is continuous at zero,
\[
\Prob\!\left(
\sum_{i=1}^{n}
S_i(Y_i,\lambda_n)
\ge0
\right)
\longrightarrow
\Phi\!\left(
\frac{h\xlo}{\sigma_{\mathrm{crit}}}
\right).
\]
\end{proof}
\subsection{Hierarchical Prior}
\label{app:hierarchical_proof}
% ---------------------------------------------------------------------

\begin{assumption}[Hierarchical regularity]
\label{ass:hierarchical_regular}
For every \(\theta\), \(F_\theta\) has a strictly positive continuously
differentiable density \(f_\theta\) on the common compact support
\(\X=[\xlo,\xhi]\). Its virtual value
\[
\virt_\theta(x)
=
x-\frac{1-F_\theta(x)}{f_\theta(x)}
\]
is weakly increasing. All posterior moments used below are jointly measurable and integrable.
The channel satisfies MLRP and \eqref{eq:l1_kernel_continuity}. The family
of conditional truthful signal laws is dominated by a common sigma-finite
measure, and the corresponding likelihood-ratio representers used in the
conditional interim allocations belong to \(L^1(\nu_{\mathrm H})\).
\end{assumption}

\medskip

Define the hierarchical monotone allocation class by
\[
\mathcal Q_{\mathrm H}^{\mathrm{mon}}
=
\left\{
q:
Q_{i,\theta}^q
\text{ is weakly increasing in }x
\text{ for every }i,\theta
\right\}.
\]

\begin{assumption}[Hierarchical score monotonicity]
\label{ass:hierarchical_score_monotonicity}
At every supporting multiplier \(\lambda\), the aggregate hierarchical score
\[
G_\lambda^{\mathrm H}(y)
=
\sum_{i=1}^{n}\Psi_i(y,\lambda)
\]
is coordinatewise weakly increasing.
\end{assumption}

\begin{assumption}[Hierarchical score atomlessness]
\label{ass:hierarchical_score_atomless}
At every supporting multiplier \(\lambda\),
\[
\Prob\!\left(
G_\lambda^{\mathrm H}(Y)=0
\right)
=
0
\]
under the hierarchical truthful signal law.
\end{assumption}

Conditional on \(\theta\), the envelope theorem gives
\[
U_i^\theta(x)
=
U_i^\theta(\xlo)
+
\int_{\xlo}^{x}
Q_i^\theta(z)\,dz.
\]
Under conditional zero rents,
\[
\E[T_i(X_i)\mid\theta]
=
\E[
Q_i^\theta(X_i)\virt_\theta(X_i)
\mid\theta
].
\]
Moreover,
\[
\E[X_i\mid\theta,Y]
=
\widehat x^\theta(Y_i),
\qquad
\E[\virt_\theta(X_i)\mid\theta,Y]
=
\widehat J^\theta(Y_i).
\]
\begin{lemma}[Hierarchical posterior representations]
\label{lem:app_hierarchical_posterior}
Under Assumption~\ref{ass:hierarchical_regular},
\begin{align}
\E[X_i\mid Y=y]
&=
\E_{\theta\mid y}
\left[
\widehat x^\theta(y_i)
\right],
\label{eq:app_hierarchical_posterior_mean}
\\
\E[\virt_\theta(X_i)\mid Y=y]
&=
\E_{\theta\mid y}
\left[
\widehat J^\theta(y_i)
\right].
\label{eq:app_hierarchical_posterior_virtual}
\end{align}
Consequently, the per-agent contribution to the hierarchical
welfare--revenue Lagrangian is
\[
\Psi_i(y,\lambda)
=
\E_{\theta\mid y}
\left[
\widehat x^\theta(y_i)
+
\lambda\widehat J^\theta(y_i)
\right].
\]
\end{lemma}

\begin{proof}
By iterated expectations,
\[
\E[X_i\mid Y=y]
=
\E_{\theta\mid y}
\left[
\E[X_i\mid\theta,Y=y]
\right].
\]
Conditional on \(\theta\), the types are independent and the channel acts
independently across agents, so
\[
\E[X_i\mid\theta,Y=y]
=
\E[X_i\mid\theta,Y_i=y_i]
=
\widehat x^\theta(y_i).
\]
This proves \eqref{eq:app_hierarchical_posterior_mean}. The same argument
applied to \(\virt_\theta(X_i)\) proves
\eqref{eq:app_hierarchical_posterior_virtual}. Summing the two terms with
weight \(\lambda\) yields the stated score.
\end{proof}

\begin{proof}[Proof of Theorem~\ref{thm:hierarchical_allocation}]
The hierarchical Lagrangian is
\[
\mathcal L_{\mathrm H}(q,\lambda)
=
\E[
q(Y)G_\lambda^{\mathrm H}(Y)
]
-\lambda R.
\]
Let \(\nu_{\mathrm H}\) denote the hierarchical truthful signal law. For
each \(i,\theta\), and \(x\), the conditional interim allocation can be
written as a dual pairing
\[
Q_{i,\theta}^q(x)
=
\int_{\Y^n}
q(y)\ell_{i,\theta,x}(y)\,
\nu_{\mathrm H}(dy)
\]
for an appropriate likelihood-ratio representer
\(\ell_{i,\theta,x}\in L^1(\nu_{\mathrm H})\), under the maintained
domination and integrability assumptions. Hence, for every
\(x'\ge x\), the restriction
\[
Q_{i,\theta}^q(x')-Q_{i,\theta}^q(x)\ge0
\]
defines a weak-\(^*\) closed half-space. The hierarchical monotone class is
therefore an intersection of weak-\(^*\) closed half-spaces inside the
weak-\(^*\) compact allocation cube. It is consequently weak-\(^*\)
compact and convex. The hierarchical welfare and revenue functionals are
weak-\(^*\) continuous by their posterior representations.
Hierarchical strict feasibility
and the supporting-hyperplane argument give a finite supporting multiplier
\(\lambda^*\).

Pointwise maximization and
Assumption~\ref{ass:hierarchical_score_atomless} imply
\[
q_{\mathrm H}^*(y)
=
\ind\{
G_{\lambda^*}^{\mathrm H}(y)>0
\}
\]
almost surely under the hierarchical truthful signal law. By
Assumption~\ref{ass:hierarchical_score_monotonicity}, this allocation is
coordinatewise weakly increasing.

Fix \(\theta\). Conditional MLRP implies that
\(Y_i\mid X_i=x_i,\theta\) is stochastically increasing in \(x_i\).
Therefore, the conditional interim allocation is weakly increasing.
Applying Proposition~\ref{prop:app_interim_bic} conditional on \(\theta\)
gives conditional BIC. Complementary slackness follows from the supporting
multiplier argument.
\end{proof}

A common ex-post transfer must satisfy, for every \(\theta\) and \(x_i\),
\[
\begin{aligned}
&
\int_{\Y^n}
t_i(y)
k(y_i\mid x_i)
\prod_{j\ne i}
m_\theta(y_j)\,dy
\\
&\qquad =
\int_{\Y^n}
q_{\mathrm H}^*(y)
\left[
x_i k(y_i\mid x_i)
-
\int_{\xlo}^{x_i}
k(y_i\mid z)\,dz
\right]
\prod_{j\ne i}
m_\theta(y_j)\,dy.
\end{aligned}
\]
Conditional completeness identifies the conditional opponent average if a
solution exists. It does not imply that a single ex-post transfer belongs to
the intersection of the ranges of all conditional channel operators.

% ---------------------------------------------------------------------
\subsection{Blackwell Monotonicity}
\label{app:proof_blackwell_monotone}
% ---------------------------------------------------------------------

\begin{proof}[Proof of Proposition~\ref{prop:blackwell_monotone}]
Suppose
\[
K_1
\succeq_{\mathrm B}
K_2.
\]
By Blackwell's theorem, there exists a report-independent Markov kernel \(L\)
such that
\[
K_2(A\mid x)
=
\int
L(A\mid y)
K_1(dy\mid x)
\]
for every measurable set \(A\).

Let \(q_2\) be reduced-form feasible under \(K_2\). Under \(K_1\), after
observing the signal profile \(Y\), generate conditionally independent
garbled signals
\[
Z_i
\sim
L(\cdot\mid Y_i),
\qquad
i=1,\ldots,n,
\]
and define
\[
q_1(Y)
=
\E[
q_2(Z)
\mid Y
].
\]
Conditional on the valuation profile, \(Z\) has exactly the signal law
generated by \(K_2\). Hence \(q_1\) and \(q_2\) induce identical interim
allocation schedules. Reduced-form monotonicity is therefore preserved.

Likewise,
\[
\E\!\left[
q_1(Y)\sum_iX_i
\right]
=
\E\!\left[
q_2(Z)\sum_iX_i
\right]
\]
and
\[
\E\!\left[
q_1(Y)\sum_i\virt(X_i)
\right]
=
\E\!\left[
q_2(Z)\sum_i\virt(X_i)
\right].
\]
Thus every reduced-form feasible welfare--revenue pair under \(K_2\) is
attainable under \(K_1\). Taking suprema proves
\[
W^{\mathrm{red}}(R;K_1)
\ge
W^{\mathrm{red}}(R;K_2).
\]
\end{proof}
\section{Additional Numerical Results}
\label{app:numerics}
% ============================================================================

This appendix records the numerical illustrations not shown in
Section~\ref{sec:numerics}. Posterior quantities are computed by quadrature and
Monte Carlo comparisons use common random numbers across channels, so that all
channel comparisons are paired.

% ------------------------------------------------------------------
\subsection{Simulation Design and Calibration}
\label{app:numerics:design}
% ------------------------------------------------------------------

\paragraph{Priors.}
The posterior-score panel uses
\[
X\sim\Unif[-1,1].
\]
The phase-transition, budget-tradeoff, and approximate-LDP experiments use
\[
X\sim\Unif[-0.5,1.5],
\]
for which
\[
\E[X]=0.5,
\qquad
\xlo=-0.5,
\qquad
\lambda_{\mathrm{crit}}=-\,\E[X]/\xlo=1.
\]
The square-root revenue experiment uses
\[
X\sim\Unif[0,1],
\]
so that \(\xlo=0\). The exponential-regime experiment uses a prior with
\(\xlo<0\).

\paragraph{Common tight-\(\mu\) calibration.}
For support width \(\Delta_{\X}\), the frontier scales are
\[
\sigma
=
\frac{\Delta_{\X}}{\mu},
\qquad
b_L
=
\frac{\Delta_{\X}}
{-2\log(2\Phi(-\mu/2))},
\qquad
\beta
=
\frac{\Delta_{\X}}
{2\log\!\big(\Phi(\mu/2)/\Phi(-\mu/2)\big)},
\]
as in Proposition~\ref{prop:common_mu_calibration}.

\paragraph{Equal-scale calibration.}
For a common pure-\(\epsilon\) frontier,
\[
b_L=\beta=\frac{\Delta_{\X}}{\epsilon}.
\]

% ------------------------------------------------------------------
\subsection{The Exponential Revenue Regime}
\label{app:numerics:exponential}
% ------------------------------------------------------------------

Of the three regimes in Theorem~\ref{thm:three_regimes_known}, the linear case
\(\xlo>0\) is governed by its leading term \(n\xlo\) and needs no separate plot;
the boundary case \(\xlo=0\) is shown in Section~\ref{sec:numerics}. Here we
illustrate the remaining case. When \(\xlo<0\),
Theorem~\ref{thm:three_regimes_known} predicts
\[
R_n^{*,\mathrm{red}}
=
\exp\{-nI_K(0)+o(n)\},
\]
and the probability of positive aggregate posterior virtual surplus obeys the
same exponential decay exponent \(I_K(0)\)
(Remark~\ref{rmk:probability_revenue_rate}).

\begin{figure}[t]
  \centering
  \begin{minipage}{0.49\linewidth}
    \centering
    \includegraphics[width=\linewidth]
    {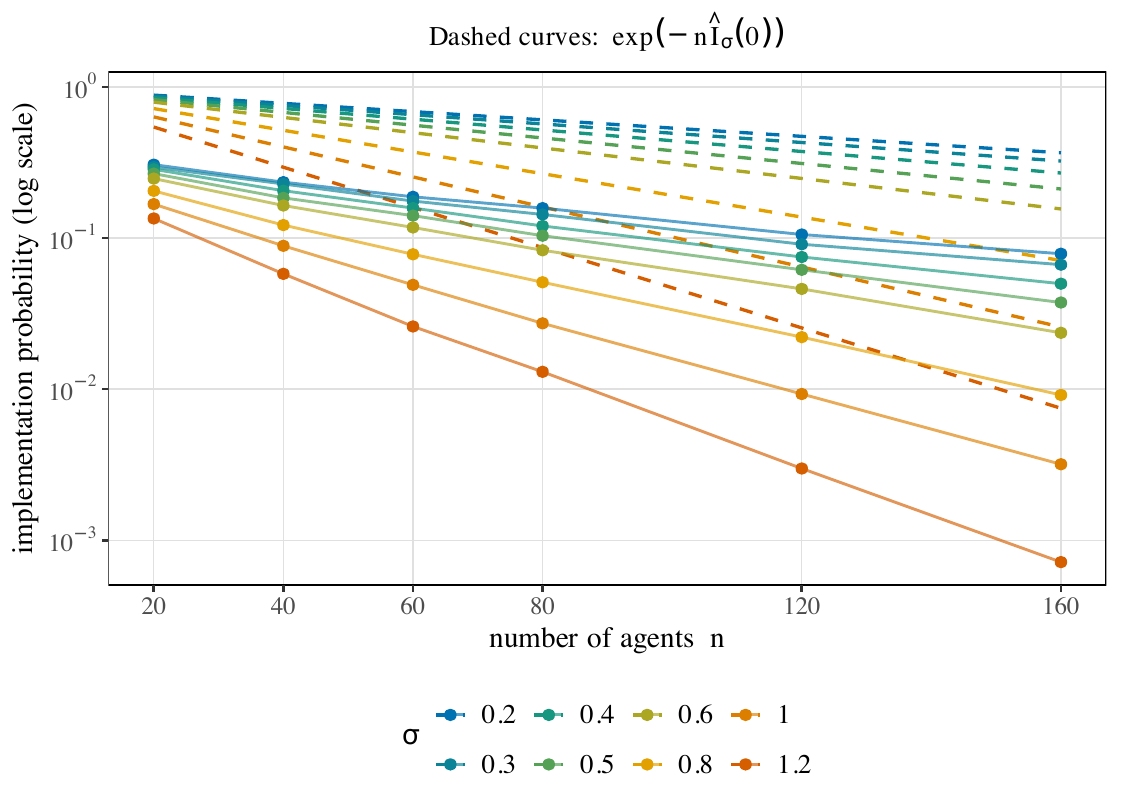}
  \end{minipage}
  \hfill
  \begin{minipage}{0.49\linewidth}
    \centering
    \includegraphics[width=\linewidth]
    {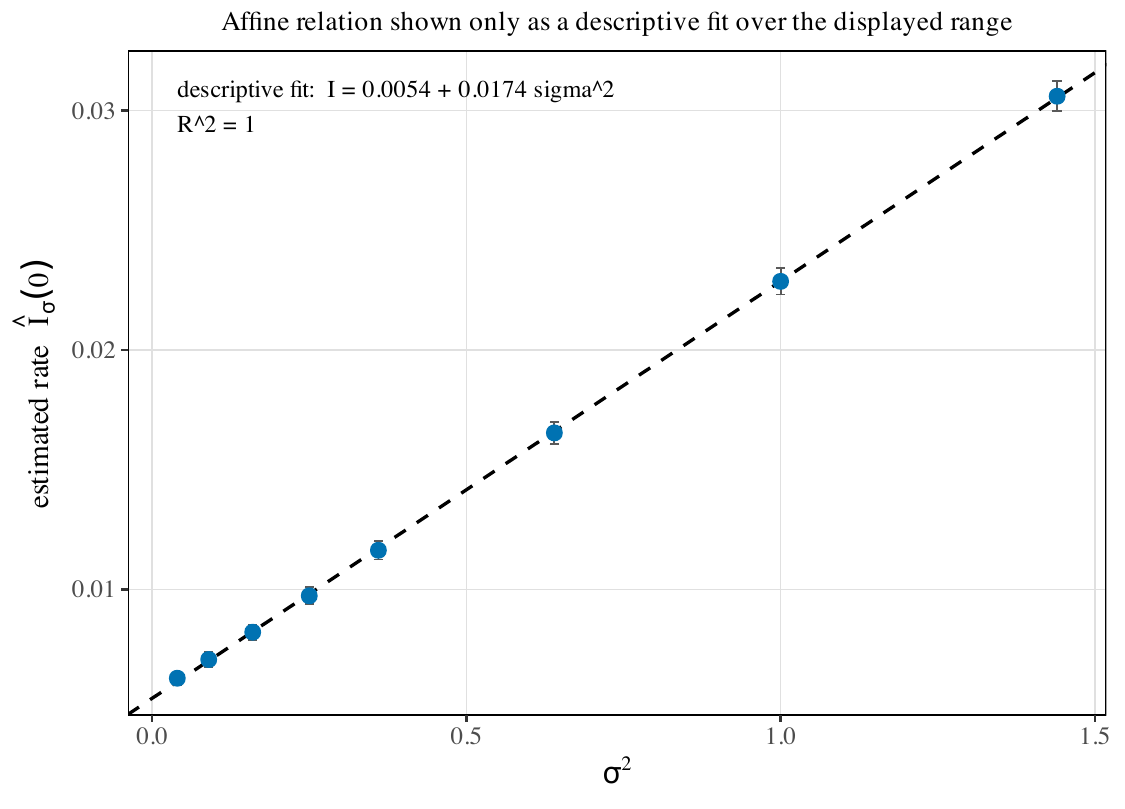}
  \end{minipage}
  \caption{
  Exponential regime (\(\xlo<0\)). Left: implementation probability on a
  logarithmic scale against \(n\), with the Cram\'er-rate prediction overlaid.
  Right: estimated Gaussian rate against the noise variance. The fitted
  relationship in the right panel is descriptive and is not asserted as a
  theorem.
  }
  \label{fig:exp4}
\end{figure}

% ------------------------------------------------------------------
\subsection{Hierarchical Unknown-Endpoint Diagnostic}
\label{app:numerics:hierarchical}
% ------------------------------------------------------------------

The hierarchical theorem assumes a common conditional support. The numerical
specification
\[
X_i\mid\theta\sim\Unif[\xlo,\theta]
\]
instead has a parameter-dependent support and is nonregular, so it falls outside
Theorem~\ref{thm:hierarchical_allocation} (Remark~\ref{rmk:hierarchical_scope}).
It is included only as a descriptive diagnostic.

\begin{figure}[t]
  \centering
  \begin{minipage}{0.49\linewidth}
    \centering
    \includegraphics[width=\linewidth]
    {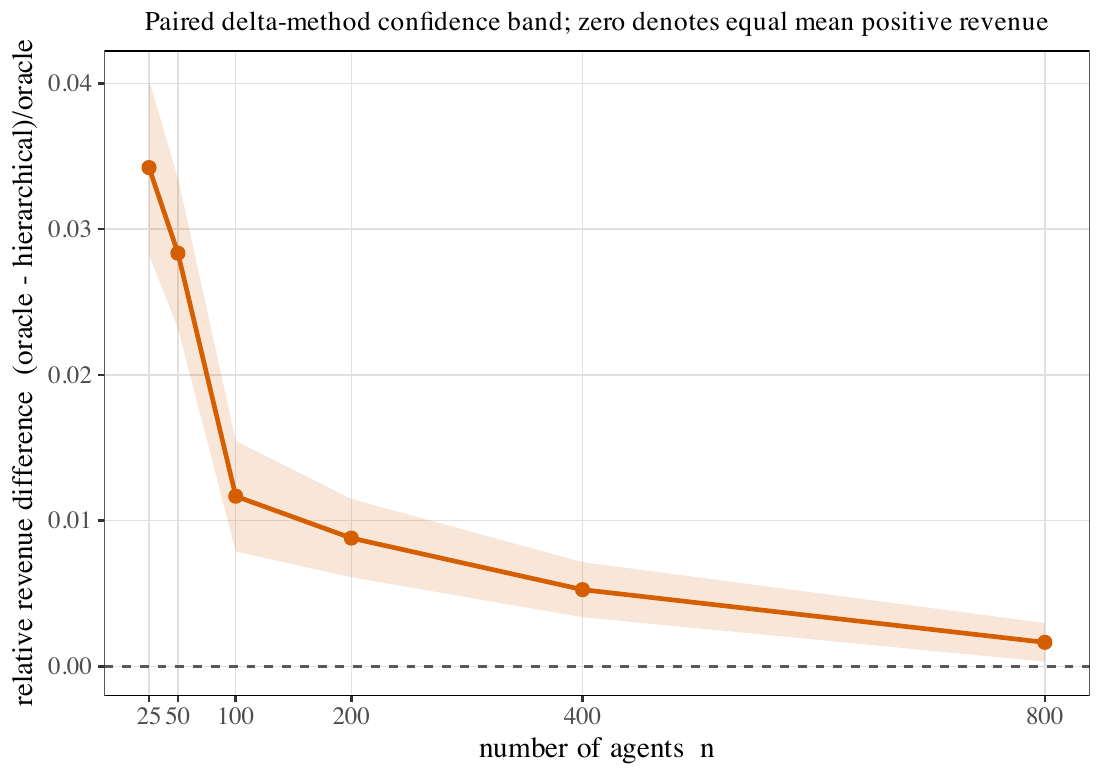}
  \end{minipage}
  \hfill
  \begin{minipage}{0.49\linewidth}
    \centering
    \includegraphics[width=\linewidth]
    {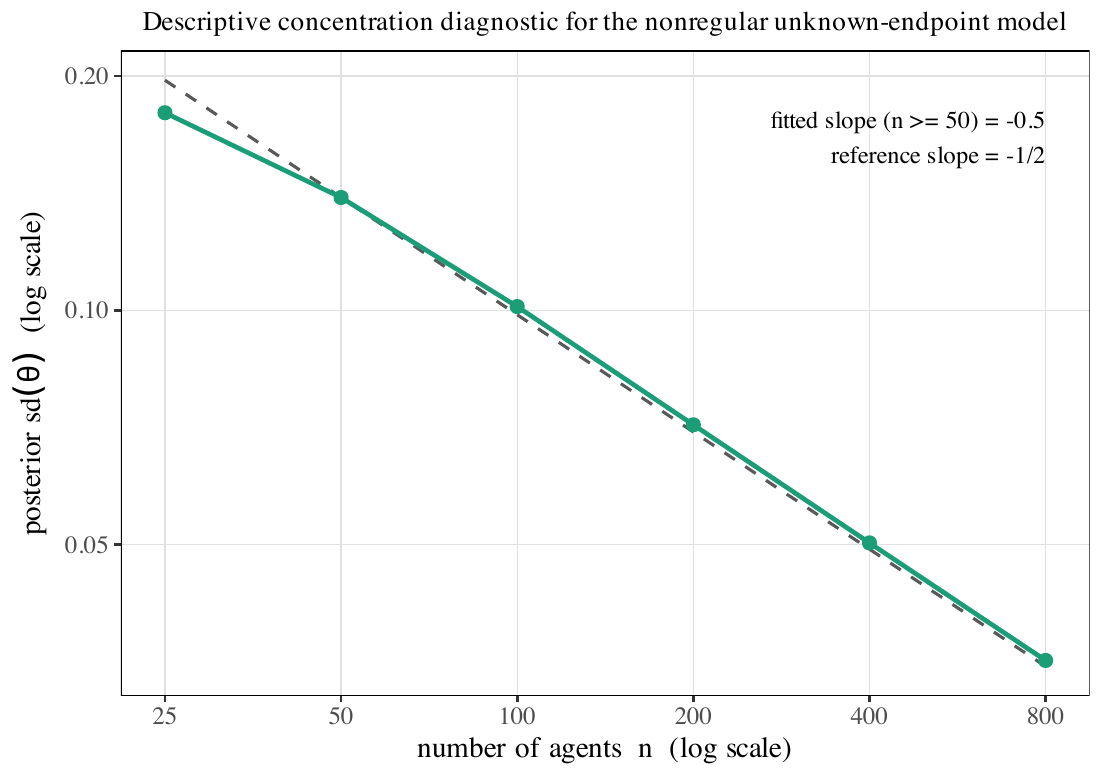}
  \end{minipage}
  \caption{
  Nonregular unknown-endpoint diagnostic. Left: relative reduced-form revenue
  shortfall of the hierarchical rule relative to the known-parameter oracle.
  Right: posterior standard deviation of the unknown endpoint. No
  Bernstein--von Mises or parametric-rate conclusion is claimed.
  }
  \label{fig:exp6}
\end{figure}

% ------------------------------------------------------------------
\subsection{Equal-Scale Welfare--Revenue Frontiers}
\label{app:numerics:equal_scale}
% ------------------------------------------------------------------

At equal scale, Laplace Blackwell-dominates logistic
(Proposition~\ref{prop:laplace_logistic_blackwell}). Both the maximal
reduced-form revenue and the full reduced-form welfare--revenue frontier are
therefore weakly higher under Laplace.

\begin{figure}[t]
  \centering
  \includegraphics[width=0.66\linewidth]
  {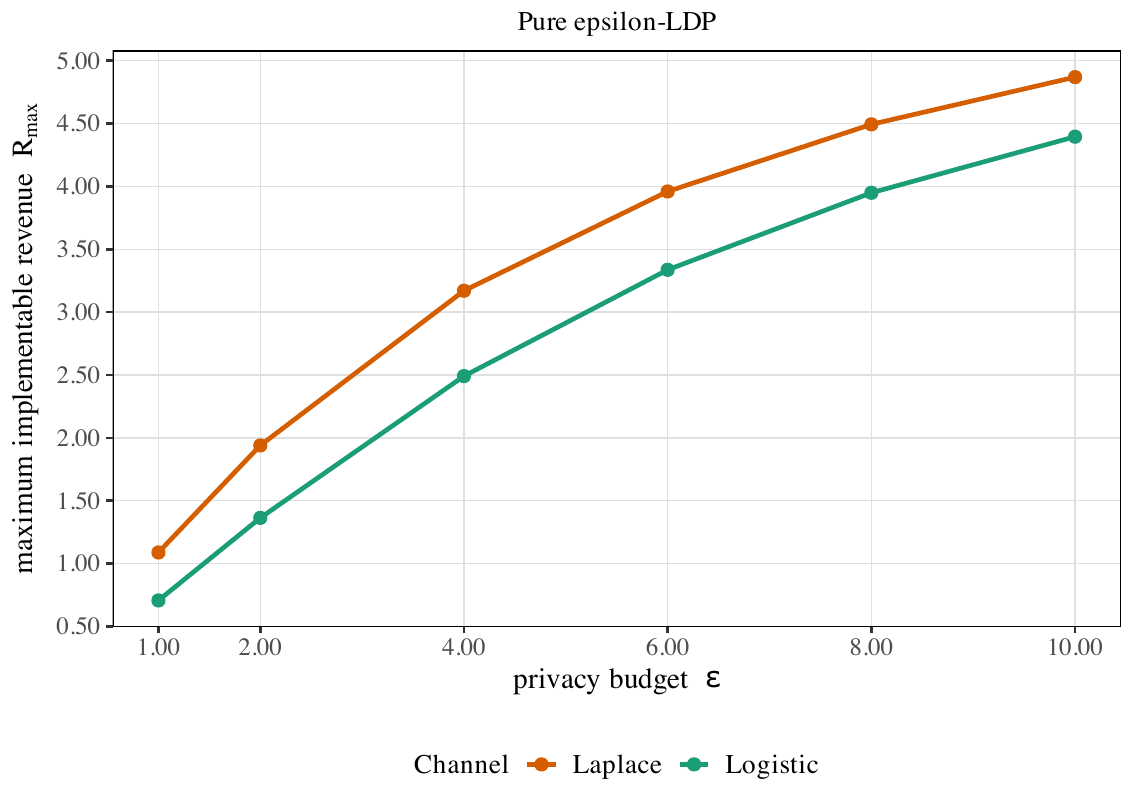}
  \caption{
  Equal-scale maximal reduced-form revenue against the common pure-\(\epsilon\)
  budget. The Laplace value is weakly greater than the logistic value,
  consistently with Proposition~\ref{prop:laplace_logistic_blackwell}.
  }
  \label{fig:exp8_rmax}
\end{figure}

\begin{figure}[t]
  \centering
  \includegraphics[width=\linewidth]
  {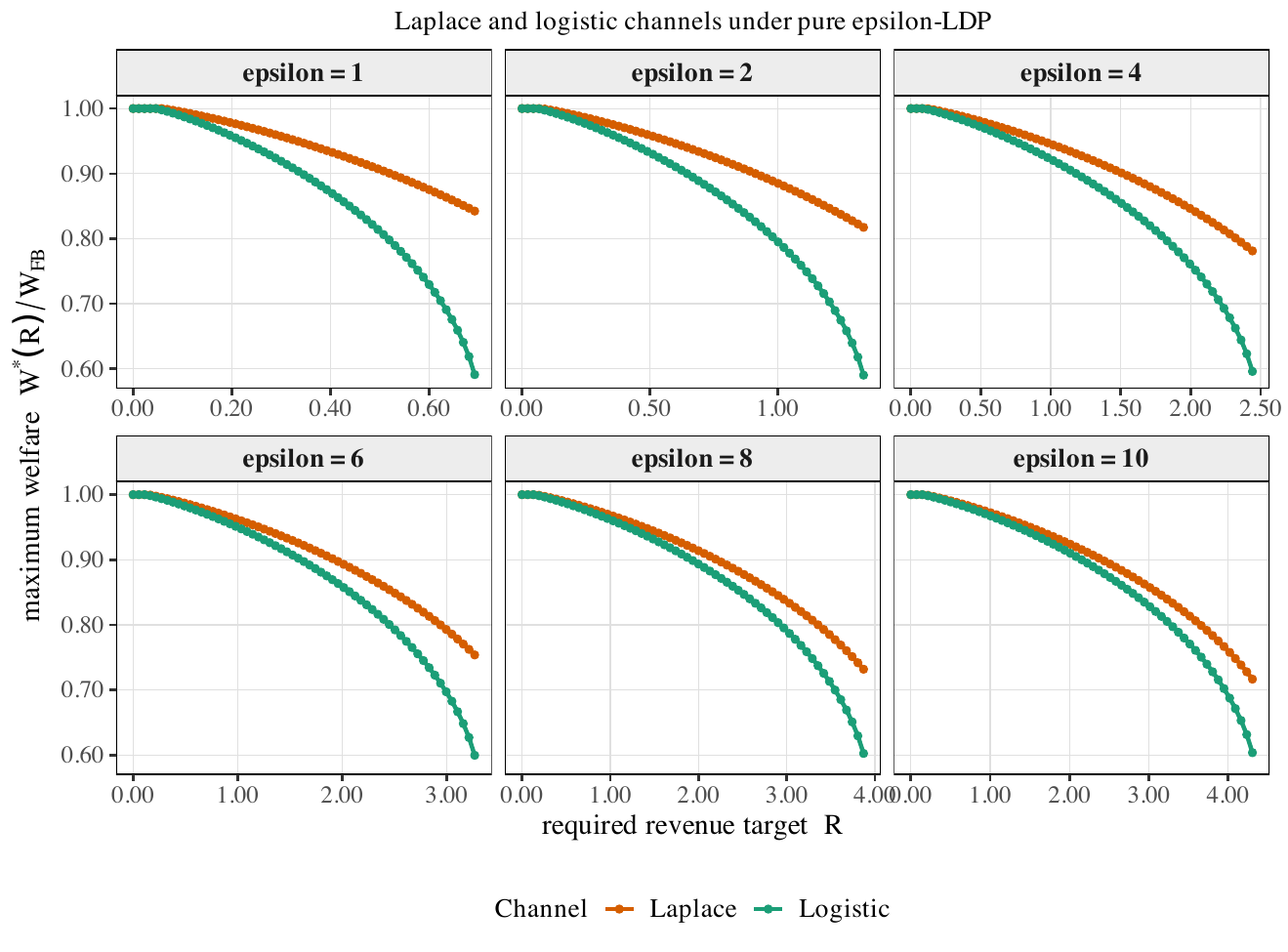}
  \caption{
  Equal-scale reduced-form welfare--revenue frontiers. In each panel the
  Laplace frontier is weakly above the logistic frontier.
  }
  \label{fig:exp8_frontier}
\end{figure}

% ------------------------------------------------------------------
\subsection{Common-\texorpdfstring{\(\mu\)}{mu} Endpoint and Continuous-Type Comparisons}
\label{app:numerics:common_mu}
% ------------------------------------------------------------------

The left panel of Figure~\ref{fig:exp9_appendix} displays the exact endpoint
trade-off difference
\[
\mathcal T_{\mathrm{Lap}}(\alpha)
-
\mathcal T_{\mathrm{Log}}(\alpha),
\]
which is nonnegative by Theorem~\ref{thm:endpoint_roc} and vanishes at
\[
\alpha\in\{0,\Phi(-\mu/2),1\},
\]
the three testing levels at which both endpoint experiments touch the Gaussian
benchmark (Remark~\ref{rmk:channel_scope}).

The right panel reports the corresponding continuous-type welfare difference.
That difference is a numerical quantity, not a consequence of the endpoint
Blackwell theorem. Its confidence band overlaps zero over a substantial part of
the displayed range, so the simulation does not support extending the endpoint
Blackwell order to the full continuum (Remark~\ref{rmk:endpoint_scope}).

\begin{figure}[t]
  \centering
  \begin{minipage}{0.52\linewidth}
    \centering
    \includegraphics[width=\linewidth]
    {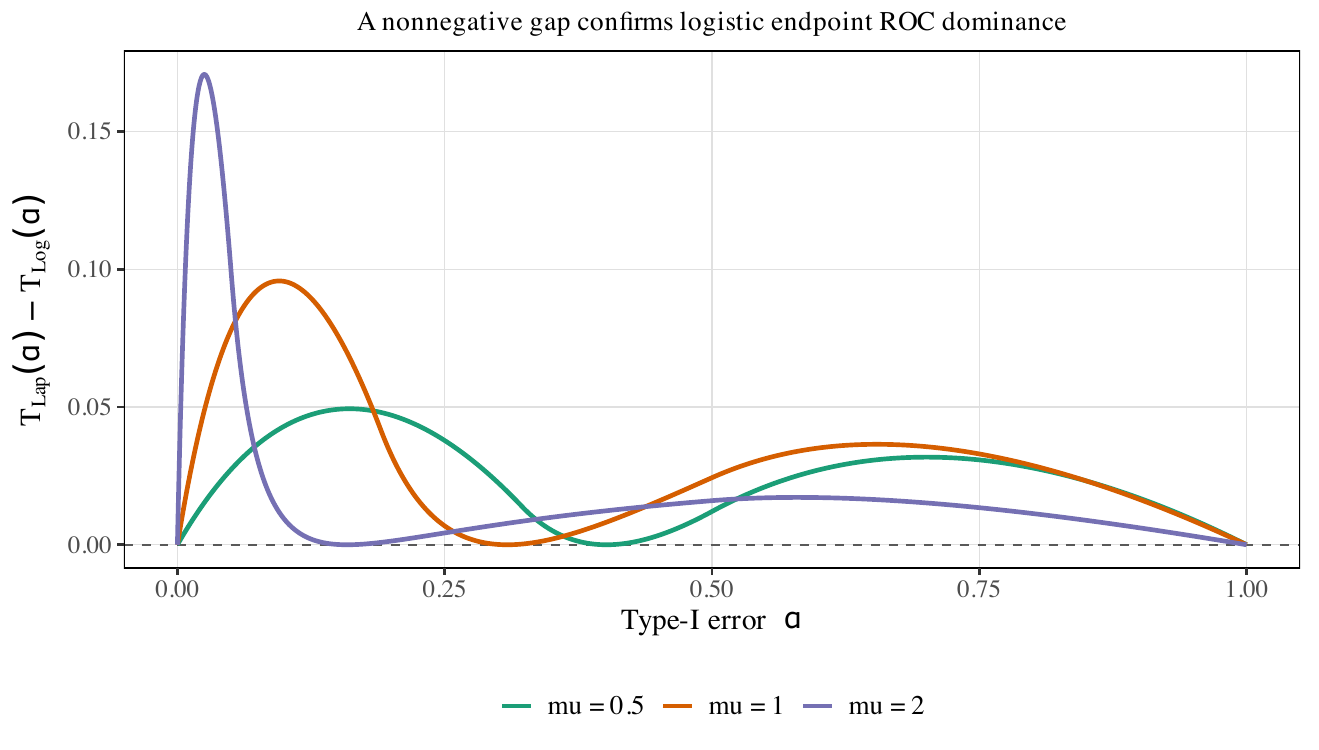}
  \end{minipage}
  \hfill
  \begin{minipage}{0.46\linewidth}
    \centering
    \includegraphics[width=\linewidth]
    {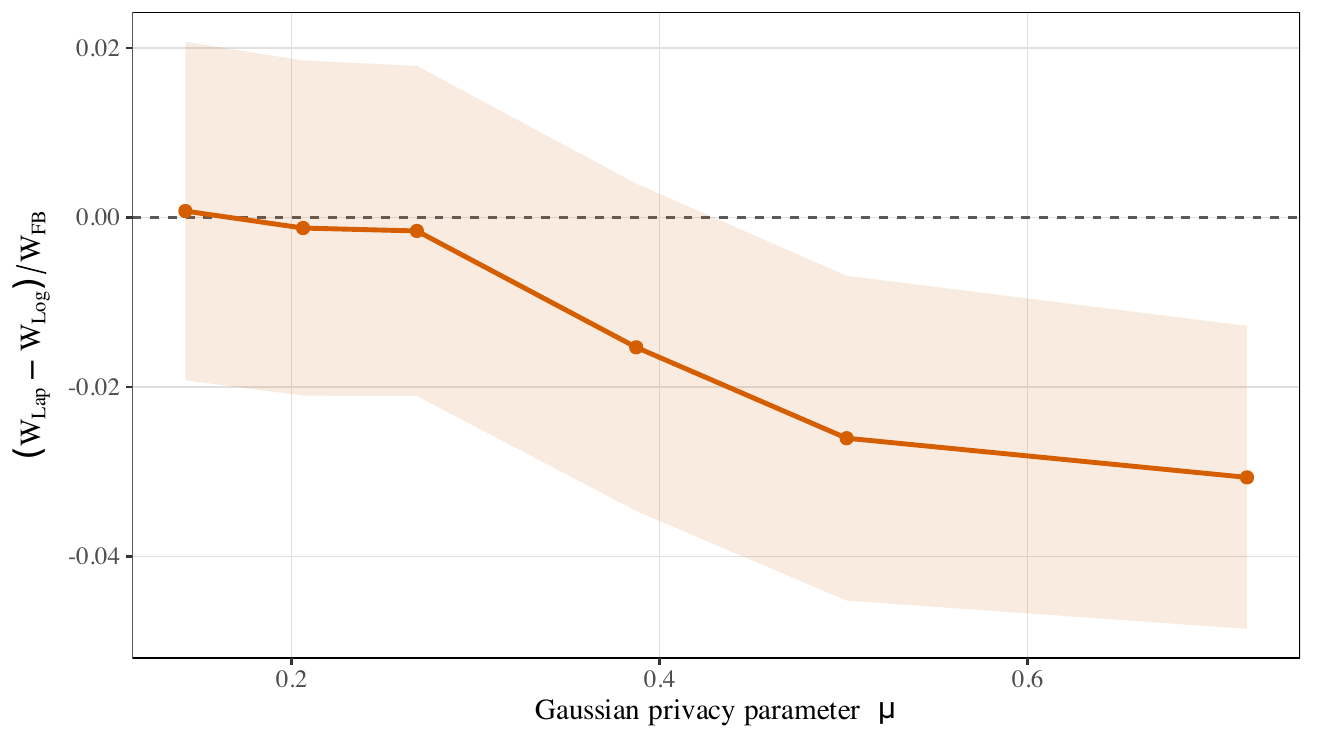}
  \end{minipage}
  \caption{
  Common tight-\(\mu\) comparison. Left: exact endpoint trade-off gap,
  nonnegative by Theorem~\ref{thm:endpoint_roc}. Right: continuous-type welfare
  difference, which is numerically small and of statistically ambiguous sign
  over much of the range.
  }
  \label{fig:exp9_appendix}
\end{figure}

% ------------------------------------------------------------------
\subsection{Welfare--Revenue Trade-off under Approximate LDP}
\label{app:numerics:tradeoff}
% ------------------------------------------------------------------

This experiment uses a common \((\epsilon,\delta)\)-LDP frontier with
\[
\delta=10^{-5},
\]
under the prior \(X\sim\Unif[-0.5,1.5]\) of the design section. The first panel
reports welfare efficiency against the privacy budget. The second reports
welfare and signed reduced-form revenue as the score multiplier varies, with the
vertical reference at the mean-score boundary \(\lambda_{\mathrm{crit}}=1\).

\begin{figure}[t]
  \centering
  \begin{minipage}{0.49\linewidth}
    \centering
    \includegraphics[width=\linewidth]
    {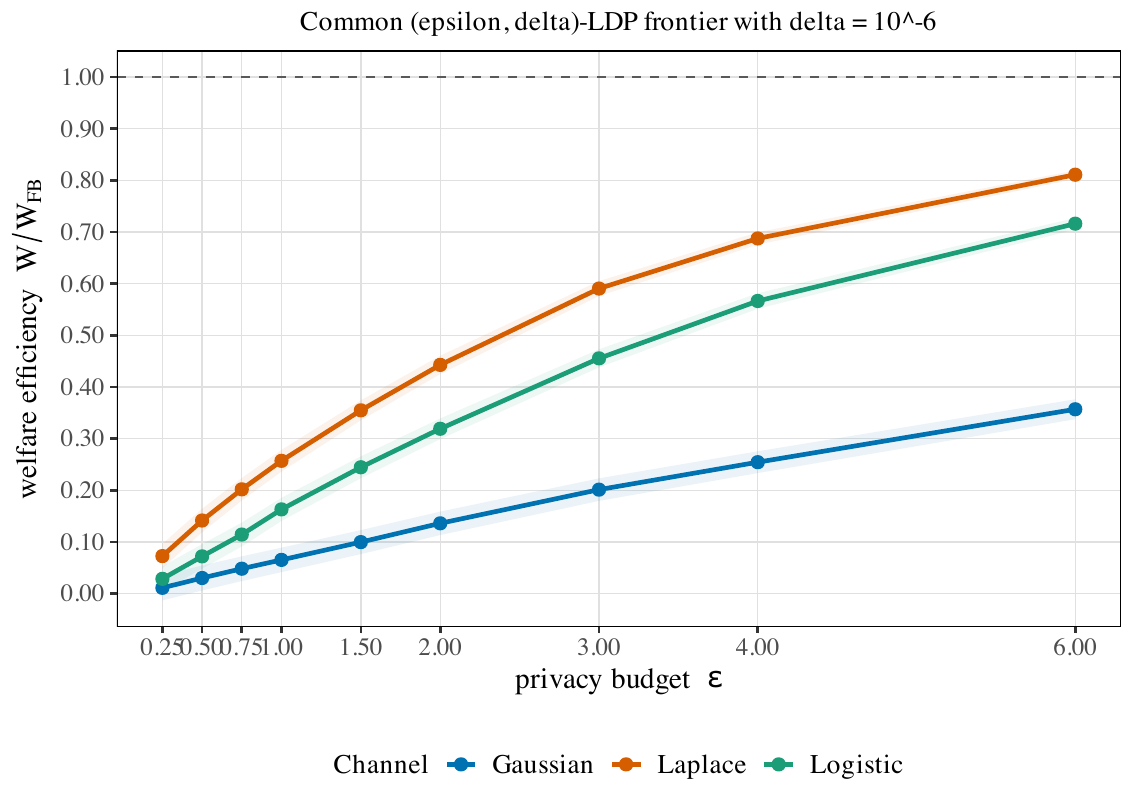}
  \end{minipage}
  \hfill
  \begin{minipage}{0.49\linewidth}
    \centering
    \includegraphics[width=\linewidth]
    {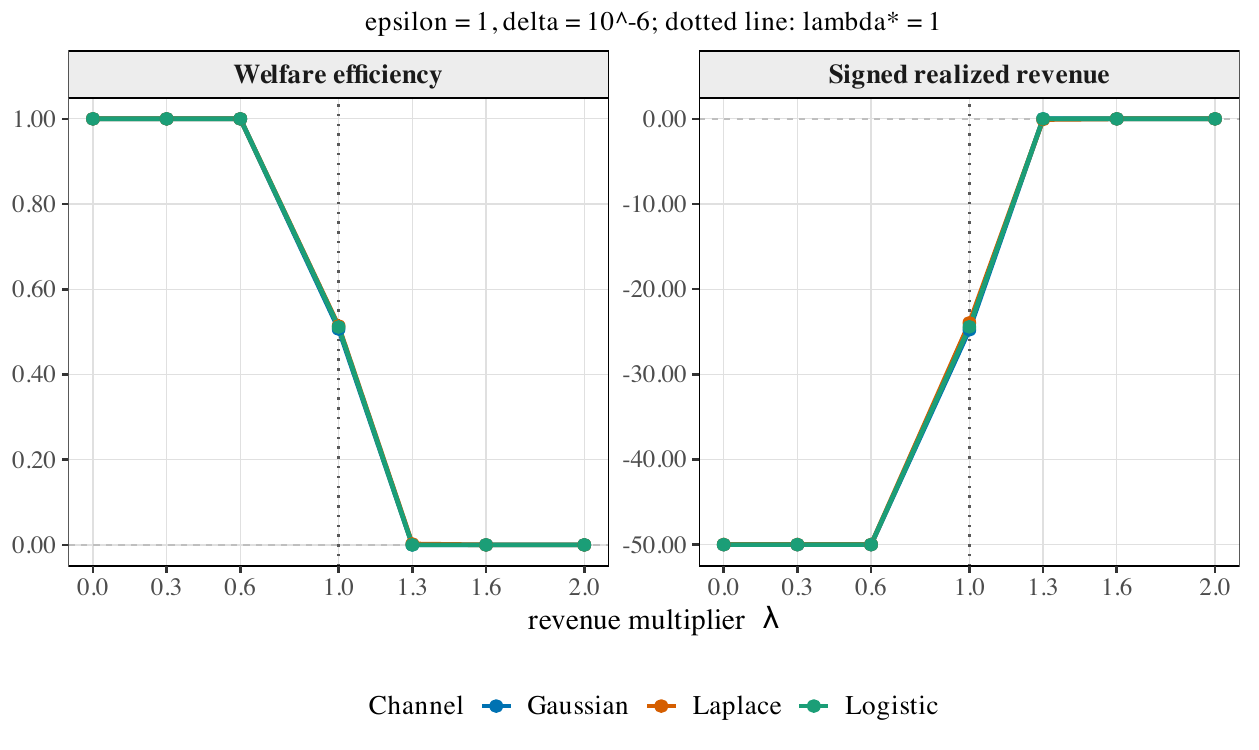}
  \end{minipage}
  \caption{
  Welfare--revenue trade-off under a common \((\epsilon,\delta)\)-LDP frontier
  with \(\delta=10^{-5}\). Left: welfare efficiency against \(\epsilon\). Right:
  welfare efficiency and signed revenue against the score multiplier. The
  vertical reference at \(\lambda_{\mathrm{crit}}=1\) is the mean-score
  boundary.
  }
  \label{fig:exp7}
\end{figure}

\clearpage
\begingroup
% --- S-numbering for the supplement ---
\setcounter{section}{0}\setcounter{equation}{0}%
\setcounter{table}{0}\setcounter{figure}{0}\setcounter{theorem}{0}%
\renewcommand{\thesection}{S\arabic{section}}%
\renewcommand{\theequation}{S\arabic{equation}}%
\renewcommand{\thetable}{S\arabic{table}}%
\renewcommand{\thefigure}{S\arabic{figure}}%
% --- new-page large heading ---
\begin{center}
  {\LARGE\bfseries Online Supplement}\\[6pt]
  {\large\itshape Public Good Provision under Locally Private Signals}
\end{center}
\vspace{1.5em}
% --- supplement body + its own reference list (separate bibliographic unit) ---
\begin{bibunit}[apalike]
  \section*{Scope of the Supplement}

This supplement collects the technical material, clarifications, and numerical
details accompanying the main paper. Its principal purpose is to make precise
the distinction among three objects that are easy to conflate.

\begin{enumerate}[label=(\roman*),leftmargin=2em,itemsep=3pt]
\item The \emph{reduced-form allocation problem}, in which Bayesian incentive
compatibility is expressed through interim monotonicity and the envelope
formula.
\item \emph{Exact signal-measurable implementation}, which additionally
requires the envelope-prescribed interim transfer to lie in the range of the
channel's conditional-expectation operator.
\item \emph{Approximate signal-measurable implementation}, obtained when the
implementing Fredholm equation is solved only up to a controlled residual.
\end{enumerate}

The main paper gives a sharp characterization of the optimal reduced-form
allocation and the corresponding reduced-form revenue envelope. Exact Bayesian
implementation by transfers measurable with respect to the privatized signals is
conditional on a Fredholm range condition, which need not hold automatically for
smoothing channels. This supplement therefore provides both an exact range
characterization and a quantitative approximate-implementation theorem that
converts a Fredholm residual into explicit incentive, participation, and revenue
errors.

The supplement also supplies the revelation argument appropriate to the
privacy-channel environment, explains the lowest-type participation
normalization, delimits the scope of completeness, records the precise status of
the hierarchical extension, and documents the complete numerical methodology. In
particular, no ironing argument is required for the Laplace channel: weak MLRP
preserves weak monotonicity, and the flat posterior-score tails create pooling
without violating Bayesian incentive compatibility.

Throughout, $\X=[\xlo,\xhi]$ is the bounded type space, $K$ the privacy channel
with conditional density $k(y\mid x)$, and all expectations are taken under the
truthful signal law. Cross-references prefixed ``M-'' point to the main paper.

\section{Exact Reduced-Form BIC and Signal-Measurable Implementation}
\label{supp:implementation}

\subsection{The conditional-expectation operator}

Fix an agent $i$. Once the allocation rule $q$ is selected, the envelope formula
determines the target interim transfer
\begin{equation}
\label{supp:eq_target_transfer}
T_i^{q}(x)
=
xQ_i^{q}(x)
-
\int_{\xlo}^{x}Q_i^{q}(z)\,dz,
\end{equation}
where the interim allocation is
\begin{equation}
\label{supp:eq_interim_allocation}
Q_i^{q}(x)
=
\int_{\Y^n}
q(y)\,
k(y_i\mid x)\,
f_{Y_{-i}}(y_{-i})\,dy.
\end{equation}

For any signal-measurable ex-post transfer $t_i:\Y^n\to\R$, define its opponent
average under the truthful law
\begin{equation}
\label{supp:eq_opponent_average}
g_i(y_i)
=
\int_{\Y^{n-1}}
t_i(y_i,y_{-i})\,
f_{Y_{-i}}(y_{-i})\,dy_{-i}.
\end{equation}
In the notation of the main paper, $g_i$ is the opponent-averaged transfer
$\bar t_i$. The induced interim transfer depends on $t_i$ only through $g_i$,
through the conditional-expectation operator
\begin{equation}
\label{supp:eq_operator}
(\mathcal K g_i)(x)
:=
\int_{\Y}
g_i(y)\,k(y\mid x)\,dy
=
\E[g_i(Y_i)\mid X_i=x].
\end{equation}
Exact signal-measurable implementation therefore requires
\begin{equation}
\label{supp:eq_exact_range}
\mathcal K g_i
=
T_i^{q}.
\end{equation}
Equation~\eqref{supp:eq_exact_range} is a Fredholm integral equation of the first
kind; it is the equation stated in the main paper, with $T_i^{q}$ the required
interim payment.

\begin{proposition}[Exact implementation is a range condition]
\label{supp:prop_exact_range}
Fix an allocation rule $q$ whose interim allocation $Q_i^{q}$ is weakly
increasing, and let $T_i^{q}$ be given by \eqref{supp:eq_target_transfer}. The
following are equivalent.
\begin{enumerate}[label=(\roman*),leftmargin=2em,itemsep=3pt]
\item There exists an integrable signal-measurable transfer $t_i:\Y^n\to\R$ that
implements the envelope interim transfer $T_i^{q}$.
\item There exists an integrable $g_i:\Y\to\R$ with $\mathcal K g_i=T_i^{q}$.
\item $T_i^{q}\in\Range(\mathcal K)$.
\end{enumerate}
Whenever $g_i$ solves the one-dimensional equation, the transfer
$t_i(y_i,y_{-i})=g_i(y_i)$ is an exact signal-measurable implementation.
\end{proposition}

\begin{proof}
If $t_i$ implements $T_i^{q}$, its opponent average $g_i$ from
\eqref{supp:eq_opponent_average} satisfies
$T_i^{q}(x)=\int_\Y g_i(y)k(y\mid x)\,dy=(\mathcal K g_i)(x)$, giving
(i)$\Rightarrow$(ii). The equivalence (ii)$\Leftrightarrow$(iii) is the
definition of the range. Finally, if $g_i$ solves the one-dimensional equation,
then with $t_i(y)=g_i(y_i)$,
\[
\E[t_i(Y)\mid X_i=x]
=
\E[g_i(Y_i)\mid X_i=x]
=
(\mathcal K g_i)(x)
=
T_i^{q}(x),
\]
which gives (ii)$\Rightarrow$(i).
\end{proof}

\begin{remark}[What exact BIC means in the paper]
\label{supp:rmk_exact_bic_scope}
The reduced-form allocation characterized in the main paper is exactly BIC in the
standard one-dimensional reduced-form sense: its interim allocation is monotone,
and the envelope formula gives the unique normalized interim transfer. The
stronger statement that the reduced form is implemented by transfers measurable
with respect to the privatized signals requires
Proposition~\ref{supp:prop_exact_range}. References to an ``exactly implementable
mechanism'' are therefore always conditional on $T_i^{q}\in\Range(\mathcal K)$.
\end{remark}

\subsection{Why smoothing creates a genuine implementation problem}

For Gaussian and other smoothing kernels, $\mathcal K$ maps functions of the
signal into smoother functions of the type. When the operator is considered on
Hilbert spaces for which its kernel is square integrable---for example, after
restriction to compact type and signal domains, or under appropriate weighted
$L^2$ conditions---it is a Hilbert--Schmidt operator and hence compact. More
generally, smoothing channel operators often have nonclosed range and unstable
inversion, but compactness must be verified for the particular function spaces
used. Whenever $\mathcal K$ is compact and injective on an infinite-dimensional
domain, its range cannot be both infinite dimensional and closed; its inverse on
the range is therefore unbounded, and solving $\mathcal K g=T$ is ill posed.

This has two distinct implications.

\begin{enumerate}[label=(\roman*),leftmargin=2em,itemsep=3pt]
\item \emph{Uniqueness does not imply existence.} Completeness can make
$\mathcal K$ injective, but injectivity only says two exact solutions cannot
differ; it does not produce a solution for a prescribed target.
\item \emph{Numerical regularization does not prove exact implementation.} A
small finite-grid residual shows that a discretized target is well approximated
on that grid. It does not establish that the continuous target lies in the
continuous operator range.
\end{enumerate}

The main paper therefore interprets its Fredholm numerical experiment as an
approximation-and-conditioning diagnostic rather than as a proof of exact range
membership.

\subsection{A nontrivial exact-implementation class}

Although a generic target need not lie in the range of a smoothing operator, the
range is neither empty nor economically trivial. For additive channels,
polynomial interim transfers admit exact signal-measurable implementations.

\begin{proposition}[Polynomial targets under additive noise]
\label{supp:prop_polynomial_range}
Suppose $Y=X+\varepsilon$, where $\varepsilon$ is independent of $X$ with finite
moments through order $m$, and let $(\mathcal K g)(x)=\E[g(x+\varepsilon)]$. Then,
restricted to polynomials of degree at most $m$, $\mathcal K$ is a bijection. In
particular, for every $T(x)=\sum_{r=0}^{m}a_rx^r$ there is a unique polynomial
$g(y)=\sum_{r=0}^{m}b_ry^r$ of degree at most $m$ with
$\E[g(x+\varepsilon)]=T(x)$ for all $x\in\R$.
\end{proposition}

\begin{proof}
For each $r\le m$,
\[
\E[(x+\varepsilon)^r]
=
\sum_{\ell=0}^{r}
\binom{r}{\ell}
x^{r-\ell}\,
\E[\varepsilon^\ell],
\]
so $\mathcal K$ sends the monomial $y^r$ to a degree-$r$ polynomial in $x$ with
leading coefficient one. In the monomial basis $(1,x,\ldots,x^m)$ the matrix of
this map is upper triangular with unit diagonal, hence invertible. Every
degree-$\le m$ target thus has a unique degree-$\le m$ preimage.
\end{proof}

The uniqueness asserted here is uniqueness within the finite-dimensional space of
polynomials of degree at most $m$. Uniqueness among all admissible transfers
additionally requires injectivity of the channel operator on the maintained
transfer space.

\begin{corollary}[Affine interim allocations]
\label{supp:cor_affine_allocation}
If the additive channel has finite second moment and the interim allocation is
affine, $Q_i(x)=a_ix+b_i$, then the normalized envelope transfer is the quadratic
\[
T_i(x)
=
xQ_i(x)-\int_{\xlo}^{x}Q_i(z)\,dz
=
\frac{a_i}{2}x^2
+
\frac{a_i}{2}\xlo^2
+
b_i\xlo,
\]
which has an exact signal-measurable implementation.
\end{corollary}

\begin{proof}
Direct integration gives
\[
T_i(x)
=
a_ix^2+b_ix
-
\frac{a_i}{2}(x^2-\xlo^2)
-
b_i(x-\xlo)
=
\frac{a_i}{2}x^2+\frac{a_i}{2}\xlo^2+b_i\xlo,
\]
a polynomial of degree at most two; apply
Proposition~\ref{supp:prop_polynomial_range}.
\end{proof}

\begin{example}[Explicit quadratic implementation]
\label{supp:ex_quadratic_implementation}
Let $Y=X+\varepsilon$ with $\E[\varepsilon]=0$ and
$\Var(\varepsilon)=\sigma_\varepsilon^2$. If $T(x)=c_2x^2+c_1x+c_0$, then
\[
g(y)=c_2\bigl(y^2-\sigma_\varepsilon^2\bigr)+c_1y+c_0
\]
satisfies $\E[g(Y)\mid X=x]=T(x)$. This covers centered Gaussian, Laplace, and
logistic additive noise.
\end{example}

\subsection{Approximate signal-measurable implementation}

When the target does not belong to $\Range(\mathcal K)$, or when range membership
is unknown, the economically relevant question is how a Fredholm approximation
error translates into incentive and participation errors.

\begin{definition}[Uniform Fredholm residual]
\label{supp:def_residual}
Let $T_i^{q}$ be the normalized envelope transfer associated with a monotone
interim allocation $Q_i^{q}$. An opponent-averaged transfer $g_i$ has uniform
Fredholm residual at most $\eta_i$ if
\begin{equation}
\label{supp:eq_uniform_residual}
\sup_{x\in\X}
\bigl|
(\mathcal K g_i)(x)-T_i^{q}(x)
\bigr|
\le
\eta_i.
\end{equation}
\end{definition}

Write $\widetilde T_i(x)=(\mathcal K g_i)(x)$ for the induced interim transfer and
$e_i(x)=\widetilde T_i(x)-T_i^{q}(x)$ for the pointwise error.

\begin{theorem}[Fredholm residual implies approximate BIC and IR]
\label{supp:thm_approx_implementation}
Let $Q_i:\X\to[0,1]$ be weakly increasing, let
$T_i(x)=xQ_i(x)-\int_{\xlo}^{x}Q_i(z)\,dz$, and suppose $g_i:\Y\to\R$ satisfies
$\sup_{x\in\X}|(\mathcal K g_i)(x)-T_i(x)|\le\eta_i$. Use the signal-measurable
transfer $t_i(y)=g_i(y_i)$. Then:
\begin{enumerate}[label=(\roman*),leftmargin=2em,itemsep=4pt]
\item the mechanism is $2\eta_i$-BIC for agent $i$,
\begin{equation}
\label{supp:eq_approx_bic}
xQ_i(x)-\widetilde T_i(x)
\ge
xQ_i(z)-\widetilde T_i(z)-2\eta_i
\qquad
\forall x,z\in\X;
\end{equation}
\item the mechanism is $\eta_i$-interim individually rational,
\begin{equation}
\label{supp:eq_approx_ir}
xQ_i(x)-\widetilde T_i(x)\ge-\eta_i
\qquad
\forall x\in\X;
\end{equation}
\item the expected-transfer error satisfies
$\bigl|\E[\widetilde T_i(X_i)]-\E[T_i(X_i)]\bigr|\le\eta_i$, so that for $n$ agents
\begin{equation}
\label{supp:eq_revenue_error}
\Bigl|
\E\!\Bigl[\textstyle\sum_{i=1}^{n}\widetilde T_i(X_i)\Bigr]
-
\E\!\Bigl[\textstyle\sum_{i=1}^{n}T_i(X_i)\Bigr]
\Bigr|
\le
\sum_{i=1}^{n}\eta_i,
\end{equation}
which is at most $n\eta$ when $\eta_i\le\eta$ for every $i$.
\end{enumerate}
\end{theorem}

\begin{proof}
The exact envelope pair $(Q_i,T_i)$ is BIC, so
$xQ_i(x)-T_i(x)\ge xQ_i(z)-T_i(z)$ for all $x,z$. Since
$\widetilde T_i=T_i+e_i$ with $|e_i|\le\eta_i$,
\[
\begin{aligned}
xQ_i(x)-\widetilde T_i(x)
&=
xQ_i(x)-T_i(x)-e_i(x)
\ge
xQ_i(z)-T_i(z)-e_i(x)
\\
&=
xQ_i(z)-\widetilde T_i(z)+e_i(z)-e_i(x)
\ge
xQ_i(z)-\widetilde T_i(z)-2\eta_i,
\end{aligned}
\]
proving (i). Under $U_i(\xlo)=0$, exact BIC gives
$U_i(x)=\int_{\xlo}^{x}Q_i(s)\,ds\ge0$, so
$xQ_i(x)-\widetilde T_i(x)=U_i(x)-e_i(x)\ge-\eta_i$, proving (ii). Finally
$|\E[\widetilde T_i(X_i)-T_i(X_i)]|=|\E[e_i(X_i)]|\le\E|e_i(X_i)|\le\eta_i$, and
summing gives \eqref{supp:eq_revenue_error}.
\end{proof}

\begin{corollary}[Exact implementation as the zero-residual case]
\label{supp:cor_zero_residual}
If $\eta_i=0$, Theorem~\ref{supp:thm_approx_implementation} yields exact BIC,
exact interim IR, and exact reduced-form revenue.
\end{corollary}

\begin{remark}[Norm choice]
\label{supp:rmk_norm_choice}
The uniform residual is used because it yields type-uniform incentive and
participation bounds. An $L^2(F)$ residual controls average transfer error but
does not, without further regularity, give a uniform BIC bound over all
type--report pairs. A numerical procedure should therefore report both its
weighted least-squares residual and a dense-grid approximation to
$\sup_{x\in\X}|(\mathcal K g_i)(x)-T_i(x)|$.
\end{remark}

\subsection{Regularization and the meaning of the numerical solution}

A common numerical estimator is the Tikhonov solution
\begin{equation}
\label{supp:eq_tikhonov}
g_{\alpha}
\in
\argmin_{g\in\mathcal G}
\Bigl\{
\|\mathcal K g-T\|_{\mathcal H_X}^{2}
+
\alpha\|g\|_{\mathcal H_Y}^{2}
\Bigr\},
\end{equation}
with regularization parameter $\alpha>0$ and chosen spaces
$\mathcal H_X,\mathcal H_Y$. The role of $\alpha$ is stability, not exact
implementation: for $\alpha>0$ the residual generally does not vanish even when
an exact solution exists, and a small discretized residual need not place the
continuous target in the exact range.

\begin{proposition}[Economic interpretation of a regularized solution]
\label{supp:prop_regularized_interpretation}
If a regularized solution $g_\alpha$ satisfies
$\sup_{x\in\X}|(\mathcal K g_\alpha)(x)-T(x)|\le\eta_\alpha$, then the mechanism
$t_i^\alpha(y)=g_\alpha(y_i)$ is $2\eta_\alpha$-BIC and $\eta_\alpha$-IR, with
per-agent expected-revenue error at most $\eta_\alpha$. The economically
meaningful output is therefore the uniform residual $\eta_\alpha$, not the
condition number or the positivity of the discretized singular values.
\end{proposition}

\begin{proof}
Apply Theorem~\ref{supp:thm_approx_implementation}.
\end{proof}

\section{Revelation through an Exogenous Privacy Channel}
\label{supp:revelation}

The strategic action of agent $i$ is not the realized signal $Y_i$, generated by
the exogenous channel, but the submitted report $r_i\in\X$. A signal-measurable
mechanism $(q,t)$ induces, for each submission $r_i$ and given truthful
opponents, the interim schedules
\begin{equation}
\label{supp:eq_submission_reduced_form}
Q_i(r_i)=\E[q(Y)\mid R_i=r_i],
\qquad
T_i(r_i)=\E[t_i(Y)\mid R_i=r_i],
\end{equation}
and a true type $x_i$ submitting $r_i$ obtains
\begin{equation}
\label{supp:eq_submission_utility}
u_i(x_i,r_i)=x_iQ_i(r_i)-T_i(r_i).
\end{equation}
All strategic comparisons are summarized by the one-dimensional menu
$r_i\mapsto(Q_i(r_i),T_i(r_i))$.

\begin{proposition}[Channel-induced taxation principle]
\label{supp:prop_taxation}
Fix the opponents' truthful strategy and the distribution of their signals. Every
signal-measurable mechanism induces a one-dimensional menu
$\{(Q_i(r),T_i(r)):r\in\X\}$, and truthful submission is Bayesian optimal if and
only if
\begin{equation}
\label{supp:eq_channel_bic}
xQ_i(x)-T_i(x)
\ge
xQ_i(r)-T_i(r)
\qquad
\forall x,r\in\X.
\end{equation}
Hence the usual one-dimensional envelope characterization applies to the
channel-induced reduced form. Conversely, a monotone reduced-form pair
$(Q_i,T_i)$ is implementable by the original privacy-channel mechanism only if
there is a signal-measurable allocation and transfer whose conditional
expectations equal $Q_i$ and $T_i$.
\end{proposition}

\begin{proof}
Conditional on $r_i$, the agent cannot affect the realized channel noise, so
expected utility is exactly \eqref{supp:eq_submission_utility} and truthfulness is
optimal precisely when \eqref{supp:eq_channel_bic} holds. The standard
one-dimensional incentive argument then forces $Q_i$ weakly increasing and the
envelope identity. The converse separates abstract reduced-form implementability
from implementation through the given kernel.
\end{proof}

\begin{remark}[No direct revelation of the privatized signal]
\label{supp:rmk_no_signal_report}
The planner does not ask the agent to report $Y_i$. The local randomizer produces
$Y_i$ after the agent selects $r_i$, and the realized signal is observed only by
the planner. The directness of the mechanism refers to the submitted channel
input $r_i$, not to control of the random output.
\end{remark}

\section{Why Interim IR Binds at the Lowest Type}
\label{supp:bottom_ir}

For weakly increasing nonnegative $Q_i$, BIC gives
$U_i(x)=U_i(\xlo)+\int_{\xlo}^{x}Q_i(s)\,ds\ge U_i(\xlo)$, so interim IR is
equivalent to $U_i(\xlo)\ge0$.

\begin{proposition}[Revenue-maximizing IR normalization]
\label{supp:prop_bottom_ir}
Fix a BIC interim allocation $Q_i$. Among all interim transfer schedules
implementing $Q_i$ and satisfying interim IR, expected revenue is maximized by
$U_i(\xlo)=0$. This holds regardless of the sign of $\xlo$.
\end{proposition}

\begin{proof}
Every BIC transfer implementing $Q_i$ has the form
$T_i(x)=xQ_i(x)-U_i(\xlo)-\int_{\xlo}^{x}Q_i(s)\,ds$. Interim IR requires
$U_i(\xlo)\ge0$, and raising $U_i(\xlo)$ subtracts the same positive constant from
every type's transfer, lowering expected revenue one-for-one. Hence $U_i(\xlo)=0$
is revenue-maximizing. The argument uses only $Q_i\ge0$, so utility is minimized
at the lowest type irrespective of the sign of $\xlo$.
\end{proof}

\section{Completeness: What Can Be Claimed}
\label{supp:completeness}

\subsection{Injectivity and uniqueness}

The completeness assumption in the main paper is an injectivity condition for
$\mathcal K g(x)=\int_\Y g(y)k(y\mid x)\,dy$: if $\mathcal K g\equiv0$ on $\X$
implies $g=0$ a.e., then two opponent-averaged transfers implementing the same
interim schedule coincide. This is a uniqueness statement only; it does not assert
$\Range(\mathcal K)$ equals the entire target space.

\subsection{The Fourier condition for full location families}

For an additive channel $Y=X+\varepsilon$ with noise density $h$,
$(\mathcal K g)(x)=\int_\R g(y)h(y-x)\,dy$, which up to reflection is convolution.

\begin{proposition}[Fourier injectivity on the full location family]
\label{supp:prop_fourier_injectivity}
Let $g\in L^1(\R)$, $h\in L^1(\R)$, and suppose the characteristic function
$\widehat h(t)=\int_\R e^{ity}h(y)\,dy$ is nowhere zero. If
$(g*\widetilde h)(x)=0$ for a.e.\ $x\in\R$, where $\widetilde h(u)=h(-u)$, then
$g=0$ a.e.
\end{proposition}

\begin{proof}
Taking Fourier transforms gives $\widehat g(t)\,\widehat{\widetilde h}(t)=0$ for
all $t$. Since $\widehat{\widetilde h}(t)=\overline{\widehat h(t)}\ne0$, we get
$\widehat g\equiv0$, and uniqueness of the Fourier transform on $L^1(\R)$ gives
$g=0$ a.e.
\end{proof}

The centered Gaussian, Laplace, and logistic characteristic functions are
\[
\widehat h_{\mathrm G}(t)=e^{-\sigma^2t^2/2},
\qquad
\widehat h_{\mathrm{Lap}}(t)=\frac{1}{1+b^2t^2},
\qquad
\widehat h_{\mathrm{Log}}(t)=\frac{\pi\beta t}{\sinh(\pi\beta t)},
\]
with the logistic value extended continuously to one at $t=0$. None vanishes on
$\R$, so the corresponding \emph{full} location families are injective on
$L^1(\R)$.

\subsection{Why a compact type interval requires care}

The main model indexes the family only by $x\in[\xlo,\xhi]$. Knowing
$(\mathcal K g)(x)=0$ for $x\in[\xlo,\xhi]$ only is weaker than knowing the
convolution vanishes for every $x\in\R$, and the nowhere-vanishing
characteristic-function argument alone does not bridge this gap.

For the Gaussian kernel, additional analyticity supplies unique continuation: the
Gaussian convolution is real analytic under standard integrability bounds, so
vanishing on a nonempty interval forces vanishing everywhere, after which Fourier
injectivity gives $g=0$. For Laplace and logistic kernels, completeness over a
compact parameter interval depends on the function class and on additional
continuation properties. The paper therefore does not infer compact-interval
completeness from the characteristic function alone; it states completeness as an
assumption wherever compact-interval uniqueness is needed.

\begin{remark}[Correct interpretation]
\label{supp:rmk_completeness_correct}
Keep separate: (i) a nowhere-vanishing characteristic function gives injectivity
for the full additive location family on standard convolution spaces; (ii)
Gaussian analyticity can extend interval-wise vanishing to global vanishing under
suitable integrability; (iii) compact-interval completeness for a particular
channel and function space needs its own verification; and (iv) none of these
injectivity statements implies surjectivity or exact range membership.
\end{remark}

\section{The Hierarchical Extension: Exact Scope}
\label{supp:hierarchical_scope}

With a latent common parameter $\theta$, the hierarchical posterior score
\[
\Psi_i(y,\lambda)
=
\E_{\theta\mid y}
\!\bigl[
\widehat x^\theta(y_i)+\lambda\widehat J^\theta(y_i)
\bigr]
\]
depends on the entire signal vector through the posterior $p(\theta\mid y)$.
Changing $y_i$ moves both the conditional posterior of $X_i$ and the posterior of
$\theta$. Consequently, MLRP of the individual channel does not by itself make
$y_i\mapsto\sum_{j}\Psi_j(y,\lambda)$ increasing.

\subsection{The high-level monotonicity condition}

The main hierarchical theorem assumes coordinatewise monotonicity of the
aggregate hierarchical score. This is sufficient for the acceptance region to be
an upper set in every signal coordinate; conditional MLRP then makes every
$\theta$-conditional interim allocation monotone. The condition is not claimed to
be primitive: an increase in one signal can move the common-state posterior in a
way that lowers other agents' posterior virtual values.

\begin{remark}[What the hierarchical theorem resolves]
\label{supp:rmk_hierarchical_resolves}
Conditional on coordinatewise score monotonicity, the theorem provides (i) the
posterior-score Lagrangian, (ii) the pointwise threshold structure, (iii)
conditional reduced-form BIC, (iv) the family of conditional Fredholm equations,
and (v) the hierarchical reduced-form revenue representation and complementary
slackness condition. It does not give a primitive
characterization of when coordinatewise monotonicity holds, nor of
Cr\'emer--McLean extraction under privatized common-state signals.
\end{remark}

\subsection{A primitive sufficient condition}

\begin{proposition}[A sufficient derivative condition]
\label{supp:prop_hierarchical_derivative}
For $j=1,\ldots,n$, define
\[
a_j^\theta(y_j,\lambda)
=
\widehat x^\theta(y_j)+\lambda\widehat J^\theta(y_j).
\]
Suppose the aggregate hierarchical score is differentiable in each signal
coordinate and the required differentiation-under-the-integral and
dominated-derivative conditions hold. If, for every $i$, $y$, and
$\lambda\ge0$,
\begin{equation}
\label{supp:eq_hierarchical_derivative_condition}
\E_{\theta\mid y}
\!\left[
\frac{\partial}{\partial y_i}
a_i^\theta(y_i,\lambda)
\right]
+
\Cov_{\theta\mid y}
\!\left(
\sum_{j=1}^{n}a_j^\theta(y_j,\lambda),\,
\frac{\partial}{\partial y_i}\log p(\theta\mid y)
\right)
\ge0,
\end{equation}
then the aggregate hierarchical score is coordinatewise weakly increasing, and the
hierarchical threshold allocation induces weakly increasing conditional interim
allocations. Equivalently, the covariance term in
\eqref{supp:eq_hierarchical_derivative_condition} may be written as
\[
\sum_{j=1}^{n}
\Cov_{\theta\mid y}
\!\left(
a_j^\theta(y_j,\lambda),\,
\frac{\partial}{\partial y_i}\log p(\theta\mid y)
\right).
\]
\end{proposition}

\begin{proof}
The aggregate hierarchical score is
\[
G_\lambda^{\mathrm H}(y)
=
\sum_{j=1}^{n}
\E_{\theta\mid y}
\!\left[a_j^\theta(y_j,\lambda)\right].
\]
For a posterior expectation of a $y$-dependent integrand $a_\theta(y)$, the
score-function identity gives
\[
\frac{\partial}{\partial y_i}
\E_{\theta\mid y}[a_\theta(y)]
=
\E_{\theta\mid y}
\!\left[\frac{\partial a_\theta(y)}{\partial y_i}\right]
+
\Cov_{\theta\mid y}
\!\left(
a_\theta(y),\,
\frac{\partial}{\partial y_i}\log p(\theta\mid y)
\right),
\]
because $\E_{\theta\mid y}\!\left[\frac{\partial}{\partial y_i}\log
p(\theta\mid y)\right]=0$. For $j\ne i$, the integrand $a_j^\theta(y_j,\lambda)$
has no direct dependence on $y_i$, so its $\E_{\theta\mid y}[\partial_{y_i}\,\cdot\,]$
contribution vanishes and only its covariance term survives. Summing over $j$,
\[
\frac{\partial}{\partial y_i}
G_\lambda^{\mathrm H}(y)
=
\E_{\theta\mid y}
\!\left[\frac{\partial}{\partial y_i}a_i^\theta(y_i,\lambda)\right]
+
\Cov_{\theta\mid y}
\!\left(
\sum_{j=1}^{n}a_j^\theta(y_j,\lambda),\,
\frac{\partial}{\partial y_i}\log p(\theta\mid y)
\right).
\]
Condition~\eqref{supp:eq_hierarchical_derivative_condition} therefore gives
$\partial G_\lambda^{\mathrm H}/\partial y_i\ge0$ for every $i$, proving
coordinatewise weak monotonicity. Conditional MLRP then yields weakly increasing
conditional interim allocations.
\end{proof}

The covariance term is the source of the difficulty. It vanishes under a known
prior and under a degenerate meta-prior, but is generally nonzero in a genuine
hierarchical model.

\subsection{The unknown-endpoint simulation}

The hierarchical numerical experiment uses supports of the form $[\,0.3,\theta\,]$.
This violates the common-support assumption of the formal theorem and is included
as a numerical extension, not a verification of it. Because $\theta$ is a support
endpoint, the model is nonregular before convolution, and privacy convolution
changes the identification and concentration problem. The supplement therefore
reports the observed posterior contraction descriptively and does not invoke an
unproved Bernstein--von Mises theorem.

\section{Revenue Envelopes and Economic Interpretation}
\label{supp:revenue_scope}

\subsection{Reduced-form revenue versus implementable revenue}

For the known-prior model, let $V_n(Y)=\sum_{i=1}^{n}\widehat J_i(Y_i)$. The
maximum reduced-form revenue is
\[
R_n^{*,\mathrm{red}}(K)=\E[(V_n(Y))_+],
\]
which is exact for the reduced-form problem because the pointwise
revenue-maximizing allocation is $q^{\mathrm{rev}}(Y)=\ind\{V_n(Y)>0\}$, with an
arbitrary value on the null tie set $\{V_n=0\}$. Let $R_n^{*,\mathrm{sig}}(K)$ be
the maximum revenue among mechanisms whose allocation and transfers are
signal-measurable and satisfy exact BIC and interim IR. Then
\begin{equation}
\label{supp:eq_revenue_upper_bound}
R_n^{*,\mathrm{sig}}(K)\le R_n^{*,\mathrm{red}}(K),
\end{equation}
Equality holds if a reduced-form revenue maximizer has envelope transfers
belonging to the relevant Fredholm ranges. Conversely, if the reduced-form
revenue maximizer is unique up to truthful-law almost-sure equality, then equality
in \eqref{supp:eq_revenue_upper_bound} requires that maximizer to admit
signal-measurable transfer implementation.

\begin{corollary}[Approximate revenue attainability]
\label{supp:cor_approx_revenue}
If the revenue-maximizing reduced form has envelope transfers
$T_i^{\mathrm{rev}}$ and, for every $i$, there is $g_i$ with
$\sup_{x\in\X}|\mathcal K g_i(x)-T_i^{\mathrm{rev}}(x)|\le\eta_i$, then there is a
signal-measurable mechanism that is $2\eta_i$-BIC and $\eta_i$-IR for each agent
and whose expected revenue is within $\sum_i\eta_i$ of $R_n^{*,\mathrm{red}}(K)$.
\end{corollary}

\begin{proof}
Apply Theorem~\ref{supp:thm_approx_implementation} to each agent's envelope
target.
\end{proof}

\subsection{The large-deviation rate function}

The negative-drift regime below is stated through the Cram\'er rate function of
the per-agent posterior virtual value. We define this object precisely before
using it.

Set
\begin{equation}
\label{supp:eq_Z_def}
Z_i=\widehat J_i(Y_i),
\qquad
S_n=\sum_{i=1}^{n}Z_i.
\end{equation}
Under the regularity and MLRP assumptions of the main paper, $\widehat J_i(Y_i)$
is bounded, say $|Z_i|\le M$, and the $Z_i$ are i.i.d.\ with common mean
\begin{equation}
\label{supp:eq_Z_mean}
\E[Z_i]=\E[\widehat J_i(Y_i)]=\xlo,
\end{equation}
the lower-endpoint identity of the main paper.

\begin{definition}[Cumulant generating function and Cram\'er rate]
\label{supp:def_cramer_rate}
The cumulant generating function (log moment generating function) of $Z_1$ is
\begin{equation}
\label{supp:eq_cgf}
\cgf(t)=\log\E\!\bigl[e^{tZ_1}\bigr],
\qquad t\in\R,
\end{equation}
which is finite for all $t\in\R$ because $Z_1$ is bounded. The associated
\emph{Cram\'er rate function} is its Legendre--Fenchel transform
\begin{equation}
\label{supp:eq_rate_function}
I_K(a)=\sup_{t\in\R}\bigl\{ta-\cgf(t)\bigr\},
\qquad a\in\R.
\end{equation}
\end{definition}

The function $I_K$ is convex, lower semicontinuous, and nonnegative, with
$I_K(\E[Z_1])=I_K(\xlo)=0$ (the supremum in \eqref{supp:eq_rate_function} at
$a=\xlo$ is attained at $t=0$, since $\cgf'(0)=\E[Z_1]=\xlo$). Because $\cgf$ is
convex with $\cgf(0)=0$ and $\cgf'(0)=\xlo$, the value at $a=0$ specializes to
\begin{equation}
\label{supp:eq_rate_at_zero}
I_K(0)=\sup_{t\in\R}\bigl\{-\cgf(t)\bigr\}=-\inf_{t\in\R}\cgf(t).
\end{equation}
If $\xlo<0$ and $0$ lies in the interior of the convex hull of $\supp(Z_1)$, then
the tilting equation $\cgf'(t^*)=0$ has a solution $t^*>0$, and
\begin{equation}
\label{supp:eq_rate_positive}
I_K(0)=-\cgf(t^*)>0,
\end{equation}
strictly positive because $\cgf$ is strictly decreasing at $t=0$
($\cgf'(0)=\xlo<0$), so $\inf_t\cgf(t)<\cgf(0)=0$.

\begin{lemma}[Cram\'er rate for the build event]
\label{supp:lem_cramer}
Suppose $\xlo<0$ and $0$ lies in the interior of the convex hull of
$\supp(Z_1)$. Then
\begin{equation}
\label{supp:eq_cramer_limit}
\lim_{n\to\infty}\frac1n\log\Prob\!\left(\frac{S_n}{n}\ge0\right)
=
-\,\inf_{a\ge0}I_K(a)
=
-\,I_K(0),
\end{equation}
with $I_K(0)$ given by \eqref{supp:eq_rate_at_zero}--\eqref{supp:eq_rate_positive}.
\end{lemma}

\begin{proof}
Cram\'er's theorem for i.i.d.\ bounded summands gives
$\lim_n\frac1n\log\Prob(S_n/n\ge0)=-\inf_{a\ge0}I_K(a)$. The function $I_K$ is
convex with unique minimizer $a=\xlo$ (value zero) and is nondecreasing on
$[\xlo,\infty)$. Since $0>\xlo$, the infimum over $a\ge0$ is attained at $a=0$,
giving $-I_K(0)$. Positivity of $I_K(0)$ is \eqref{supp:eq_rate_positive}.
\end{proof}

\subsection{The three regimes}

Using Definition~\ref{supp:def_cramer_rate}, the known-prior reduced-form theorem
gives the following behavior of $R_n^{*,\mathrm{red}}(K)=\E[(S_n)_+]$, with
$\sigma_J(K)^2=\Var(Z_1)\in(0,\infty)$.

\begin{enumerate}[label=(\roman*),leftmargin=2em,itemsep=3pt]
\item If $\xlo>0$, then $R_n^{*,\mathrm{red}}(K)=n\xlo+o(n)$; more precisely the
remainder is exponentially small.
\item If $\xlo=0$, then
$R_n^{*,\mathrm{red}}(K)=\dfrac{\sigma_J(K)}{\sqrt{2\pi}}\sqrt n+o(\sqrt n)$.
\item If $\xlo<0$ and $0$ lies in the interior of the convex hull of
$\supp(Z_1)$, then
$\dfrac1n\log R_n^{*,\mathrm{red}}(K)\to-I_K(0)<0$, with $I_K(0)$ as in
Definition~\ref{supp:def_cramer_rate}. By Lemma~\ref{supp:lem_cramer} the build
probability $\Prob(V_n\ge0)$ has the same exponential rate $-I_K(0)$.
\end{enumerate}

These statements describe the reduced-form envelope; by
\eqref{supp:eq_revenue_upper_bound} they are upper bounds for exact
signal-measurable revenue unless the range condition is verified. Under a sequence
of approximations with residuals $\eta_{i,n}$, the same rates transfer to
approximately BIC mechanisms whenever $\sum_{i=1}^{n}\eta_{i,n}$ is asymptotically
negligible relative to the relevant revenue scale: it suffices that
$\sum_i\eta_{i,n}=o(n)$ in the linear regime, $\sum_i\eta_{i,n}=o(\sqrt n)$ at the
knife edge, and that $\sum_i\eta_{i,n}$ be exponentially small to preserve the
exact exponential rate in the negative-drift regime.

\subsection{Relation to public-good impossibility}

The negative-endpoint regime reflects the same force as classical public-good
impossibility results: when average virtual surplus is negative, a
revenue-compatible build event sits in an increasingly rare upper tail. The
present paper measures the decay of maximum reduced-form revenue, whereas the
classical literature often emphasizes provision probabilities, efficiency, or
budget balance. A precise theorem equating $I_K(0)$ with a specific
provision-probability exponent in the classical model would require matching the
allocation, transfer normalization, and information structure, so the connection
is treated as an economic analogy rather than a formal equivalence.

\section{Additional Related Literature}
\label{supp:related}

This section expands the literature discussion in the main text.

\citet{McSherryTalwar2007} use differential privacy of an outcome rule to obtain
approximate incentive compatibility. The present paper takes a different route:
the local privacy channel is exogenous, and the planner solves an exact
reduced-form incentive problem conditional on the information that survives the
channel; exact implementation through signal-measurable transfers is then a
separate inverse problem (Section~\ref{supp:implementation}).

\citet{HuangKannan2012} obtain exact truthfulness with VCG-style payments in a
central-privacy setting in which a trusted curator observes raw reports before
perturbing the output. This differs from the local model, where the planner never
observes the raw valuation or submitted channel input without randomization.

\citet{GhoshRoth2011} study markets in which privacy loss is itself purchased;
\citet{PaiRoth2013} emphasize the multidimensional screening issue that arises
when willingness to accept privacy loss correlates with the private type; and
\citet{NissimOrlandiSmorodinsky2012} and related work treat agents who directly
value privacy. Here privacy enters through a fixed information channel rather than
through a privacy term in utility.

The local model originates in randomized response \citep{warner1965randomized},
with its modern formalization in \citet{kasiviswanathan2011what}; the statistical
contraction it induces is studied by \citet{duchi2013local,duchi2018minimax}. The
present paper studies the corresponding contraction of feasible surplus and
revenue in a mechanism-design problem.

The transfer problem is a first-kind inverse problem; the distinction among
injectivity, range membership, and stable inversion is standard for integral
equations \citep{kress2014linear}. Statistical completeness provides injectivity
but not surjectivity \citep{lehmann2005testing}, and bounded versus full
completeness can differ for location families, which is why compact-interval
completeness is stated explicitly rather than inferred from a characteristic
function. The Blackwell comparison of experiments is due to
\citet{blackwell1951comparison,blackwell1953equivalent}; the paper uses it only
when an explicit garbling is available, as in the equal-scale Laplace-over-logistic
convolution identity, and infers no general Gaussian-versus-pure-LDP ordering from
noise variance alone.

\section{Mathematical and Statistical Foundations}
\label{supp:math_foundations}

\subsection{Envelope theorem and virtual values}

For a scalar type, a reduced form is BIC if and only if its allocation
probability is weakly increasing and truthful utility satisfies the envelope
formula. Under $U_i(\xlo)=0$,
\[
U_i(x)=\int_{\xlo}^{x}Q_i(s)\,ds,
\qquad
T_i(x)=xQ_i(x)-\int_{\xlo}^{x}Q_i(s)\,ds,
\]
and integrating over types gives the one-dimensional Myerson reduction
$\E[T_i(X_i)]=\E[Q_i(X_i)\virt(X_i)]$ \citep{myerson1981optimal}.

\subsection{MLRP and posterior monotonicity}

If $k(y\mid x)$ satisfies MLRP and $a(\cdot)$ is weakly increasing, then
$y\mapsto\E[a(X)\mid Y=y]$ is weakly increasing. With $a(x)=x+\lambda\virt(x)$,
regularity of the virtual value makes the posterior score weakly increasing. For
Laplace noise the likelihood ratio is flat in the global tails, so the posterior
score is constant there; this is pooling, not nonmonotonicity, the acceptance set
remains an upper set, and the interim allocation remains weakly increasing. No
ironing is required: ironing is the device that replaces a nonmonotone virtual
value by its convexified version to restore implementability, whereas here
regularity and MLRP already make the score weakly increasing and Laplace tail
pooling only flattens it.

\subsection{Central limit behavior at the boundary}

When the per-agent score has mean zero and positive finite variance $\sigma_Z^2$,
the positive part of its sum is asymptotically half-normal:
\[
\E\!\left[\Bigl(\textstyle\sum_{i=1}^{n}Z_i\Bigr)_+\right]
=
\frac{\sigma_Z}{\sqrt{2\pi}}\sqrt n+o(\sqrt n),
\]
since $S_n/(\sigma_Z\sqrt n)\Rightarrow\Normal(0,1)$, the positive parts are
uniformly integrable (their second moments are bounded uniformly in $n$), and,
where $N\sim\Normal(0,1)$,
\[
\E[N_+]=\frac{1}{\sqrt{2\pi}}.
\]

\subsection{Large deviations}

When the per-agent mean is negative, the event $\{S_n\ge0\}$ is a large deviation.
Its exponential rate is governed by the Cram\'er rate function of
Definition~\ref{supp:def_cramer_rate}; the precise statement and the
specialization to $I_K(0)$ are given in Lemma~\ref{supp:lem_cramer}.

\subsection{Blackwell order}

If $K_1\succeq_{\mathrm B}K_2$, every decision rule under $K_2$ can be replicated
under $K_1$ by garbling the more informative signal, so the planner achieves at
least as much constrained reduced-form welfare under $K_1$. The order is partial:
variance alone does not imply Blackwell dominance.

\section{Notation}
\label{supp:notation}

The symbols below match the main paper. Where the supplement uses a distinct local
symbol for a main-text object, the main-text symbol is given in parentheses.

\subsection{Spaces and primitive objects}

\begin{center}
\renewcommand{\arraystretch}{1.2}
\begin{tabular}{@{}p{0.20\linewidth}p{0.70\linewidth}@{}}
\toprule
\textbf{Symbol} & \textbf{Meaning} \\
\midrule
$\X=[\xlo,\xhi]$ & Type and submission space \\
$\Y$ & Privatized signal space \\
$\Th$ & Hyperparameter space \\
$X_i,x_i$ & Agent $i$'s valuation \\
$R_i,r_i$ & Report submitted to the privacy channel \\
$Y_i,y_i$ & Privatized signal observed by the planner \\
$F,f$ & Prior distribution and density \\
$F_\theta,f_\theta$ & Conditional prior indexed by $\theta$ \\
$\pi$ & Meta-prior on $\theta$ \\
$K(\cdot\mid z)$ & Privacy-channel kernel \\
$k(y\mid z)$ & Privacy-channel density \\
$m(y)$ & One-signal marginal density \\
$m_\theta(y)$ & Conditional marginal signal density \\
\bottomrule
\end{tabular}
\end{center}

\subsection{Mechanism and implementation objects}

\begin{center}
\renewcommand{\arraystretch}{1.2}
\begin{tabular}{@{}p{0.20\linewidth}p{0.70\linewidth}@{}}
\toprule
\textbf{Symbol} & \textbf{Meaning} \\
\midrule
$q(y)$ & Probability of implementing the project after signal $y$ \\
$t_i(y)$ & Signal-measurable ex-post transfer \\
$Q_i(x)$ & Interim allocation induced by submitted value $x$ \\
$T_i(x)$ & Interim expected transfer (envelope target) \\
$U_i(x)$ & Truthful interim utility \\
$g_i(y_i)$ & Opponent-averaged transfer (main text: $\bar t_i$) \\
$\mathcal K g_i$ & Conditional expectation $\E[g_i(Y_i)\mid X_i=x]$ \\
$\Range(\mathcal K)$ & Exact Fredholm range \\
$\eta_i,\eta$ & Uniform Fredholm residual (bound) \\
$\widetilde T_i(x)$ & Interim transfer induced by $g_i$ \\
$e_i(x)$ & Pointwise transfer error $\widetilde T_i(x)-T_i(x)$ \\
$g_\alpha$ & Tikhonov-regularized opponent average \\
$\alpha$ & Regularization parameter \\
$\mathcal H_X,\mathcal H_Y$ & Tikhonov regularization spaces \\
$\mathcal Q^{\mathrm{mon}}$ & Known-prior monotone reduced-form allocation class \\
$\mathcal Q_{\mathrm H}^{\mathrm{mon}}$ & Hierarchical conditional-monotonicity class \\
\bottomrule
\end{tabular}
\end{center}

\subsection{Posterior, revenue, and large-deviation objects}

\begin{center}
\renewcommand{\arraystretch}{1.2}
\begin{tabular}{@{}p{0.20\linewidth}p{0.70\linewidth}@{}}
\toprule
\textbf{Symbol} & \textbf{Meaning} \\
\midrule
$\virt(x)$ & Myerson virtual value $x-(1-F(x))/f(x)$ \\
$\virtt(x)$ & Conditional virtual value under $F_\theta$ \\
$\widehat x_i(y_i)$ & Posterior mean of $X_i$ \\
$\widehat J_i(y_i)$ & Posterior mean of $\virt(X_i)$ \\
$Z_i$ & $\widehat J_i(Y_i)$; i.i.d., bounded, mean $\xlo$ \\
$S_i(y_i,\lambda)$ & Known-prior posterior generalized virtual score \\
$\mu_S(\lambda)$ & Per-agent mean score $\E[X]+\lambda\xlo$ \\
$\lamcrit$ & Mean-score boundary $-\E[X]/\xlo$ \\
$\Psi_i(y,\lambda)$ & Hierarchical posterior score \\
$V_n(Y)$ & Aggregate posterior virtual surplus $\sum_i\widehat J_i(Y_i)$ \\
$G_\lambda$ & Aggregate posterior score $\sum_iS_i(y_i,\lambda)$ \\
$R_n^{*,\mathrm{red}}(K)$ & Maximum reduced-form revenue \\
$R_n^{*,\mathrm{sig}}(K)$ & Maximum exactly implementable signal-measurable revenue \\
$W^{\mathrm{red}}(R;K)$ & Constrained reduced-form welfare value \\
$\cgf(t)$ & Cumulant generating function $\log\E[e^{tZ_1}]$ \\
$I_K(a)$ & Cram\'er rate function $\sup_t\{ta-\cgf(t)\}$; $I_K(0)$ its value at $0$ \\
$\sigma_J(K)$ & Standard deviation of $\widehat J_i(Y_i)$ \\
$\sigma_S(\lambda)$ & Standard deviation of $S_i(Y_i,\lambda)$ \\
$\widehat x^\theta(y_i)$ & Conditional posterior mean of $X_i$ given $\theta$ \\
$\widehat J^\theta(y_i)$ & Conditional posterior virtual value given $\theta$ \\
\bottomrule
\end{tabular}
\end{center}

\subsection{Privacy, calibration, and channel-comparison objects}

\begin{center}
\renewcommand{\arraystretch}{1.2}
\begin{tabular}{@{}p{0.20\linewidth}p{0.70\linewidth}@{}}
\toprule
\textbf{Symbol} & \textbf{Meaning} \\
\midrule
$\mu$ & Gaussian differential privacy (GDP) parameter \\
$\epsilon$ & Pure-LDP privacy budget \\
$\delta$ & Approximate-LDP slack \\
$\sigma$ & Gaussian noise scale \\
$b_L$ & Laplace noise scale \\
$\beta$ & Logistic noise scale \\
$\Delta_{\X}$ & Type-space width $\xhi-\xlo$ \\
$\varepsilon$ & Additive channel noise variable, $Y=X+\varepsilon$ \\
$\widehat h(t)$ & Characteristic function of the noise density \\
$K_{\cdot,s}$ & Privacy channel at scale $s$ \\
$\succeq_{\mathrm B}$ & Blackwell informativeness order \\
$\mathcal T(\alpha)$ & Endpoint trade-off (ROC) function at testing level $\alpha$ \\
$W_{\mathrm{FB}}$ & Full-information welfare benchmark \\
\bottomrule
\end{tabular}
\end{center}

\section{Laplace and Logistic Noise at Equal Pure-LDP Scale}
\label{supp:variance}

On a type space of width $\Delta_{\X}=\xhi-\xlo$, both additive channels satisfy
pure $\epsilon$-LDP at common scale $b=\beta=\Delta_{\X}/\epsilon$. Their variances
are $\Var(\Lap(0,b))=2b^2$ and $\Var(\Logistic(0,\beta))=\pi^2\beta^2/3$, so at
equal scale
\[
\frac{\Var(\Logistic(0,\beta))}{\Var(\Lap(0,b))}
=
\frac{\pi^2}{6}
\approx
1.645.
\]
This variance comparison is only descriptive. The substantive result is the
Blackwell ordering at equal scale,
\[
K_{\mathrm{Lap},\beta}\succeq_{\mathrm B}K_{\mathrm{Log},\beta},
\]
established by the convolution identity of the next section: a logistic noise
variable with scale $\beta$ is an independent sum of a $\Lap(0,\beta)$ variable
and additional noise, so the logistic experiment is a garbling of the Laplace
experiment. The Laplace posterior score is weakly increasing and flat in the
global signal tails; these flat regions do not create nonmonotonicity and require
no ironing.

\section{Complete Numerical Methodology}
\label{supp:sim_methodology}

This section documents the computations behind the numerical illustrations of the
main paper and its numerical appendix. All experiments share the same primitives:
a prior $F$ on $[\xlo,\xhi]$, an additive privacy channel, and the posterior
objects $\widehat x$, $\widehat J$, and $S(\cdot,\lambda)$.

\subsection{Posterior moments and channel densities}

Each additive channel has the form $Y=X+\eta$ with density $k(y\mid x)=h(y-x)$,
where $h$ is the centered noise density: Gaussian $\Normal(0,\sigma^2)$, Laplace
$\Lap(0,b)$, or logistic $\Logistic(0,\beta)$. For a given prior,
\[
m(y)=\int_{\xlo}^{\xhi}f(x)k(y\mid x)\,dx,
\qquad
\widehat x(y)=\frac{1}{m(y)}\int_{\xlo}^{\xhi}x\,f(x)k(y\mid x)\,dx,
\]
\[
\widehat J(y)=\frac{1}{m(y)}\int_{\xlo}^{\xhi}\virt(x)\,f(x)k(y\mid x)\,dx,
\qquad
\virt(x)=x-\frac{1-F(x)}{f(x)},
\]
and $S(y,\lambda)=\widehat x(y)+\lambda\widehat J(y)$. These integrals are
evaluated by numerical quadrature on a fine grid over $[\xlo,\xhi]$; the signal
argument $y$ is taken on a dense grid wide enough to capture the relevant noise
tails.

\subsection{Priors and calibrations}

The priors match the regime each experiment probes:
\[
\text{posterior score: } X\sim\Unif[-1,1];
\qquad
\text{square-root revenue: } X\sim\Unif[0,1]\ (\xlo=0);
\]
\[
\text{phase transition, budget trade-off, approximate LDP: } X\sim\Unif[-0.5,1.5],
\]
for which $\E[X]=0.5$, $\xlo=-0.5$, and $\lamcrit=-\E[X]/\xlo=1$. The
exponential-regime experiment uses a prior with $\xlo<0$.

Channel scales are calibrated in two ways. The common tight-$\mu$ calibration uses
\[
\sigma=\frac{\Delta_{\X}}{\mu},
\qquad
b_L=\frac{\Delta_{\X}}{-2\log(2\Phi(-\mu/2))},
\qquad
\beta=\frac{\Delta_{\X}}{2\log\!\big(\Phi(\mu/2)/\Phi(-\mu/2)\big)},
\]
so each channel satisfies $\mu$-GDP with least-private report pair $(\xlo,\xhi)$.
The equal-scale calibration uses a common pure-$\epsilon$ frontier,
$b_L=\beta=\Delta_{\X}/\epsilon$. The approximate-LDP experiment uses a common
$(\epsilon,\delta)$-LDP frontier with $\delta=10^{-5}$.

\subsection{Monte Carlo and common random numbers}

Implementation frequencies, revenue, and welfare are estimated by Monte Carlo. For
each population size $n\in\{25,\ldots,500\}$, types are drawn $X_i\sim F$ and
signals are generated through each channel from the same underlying uniform draws,
so that channel comparisons share randomness (common random numbers). Pairing
removes first-order sampling noise from cross-channel differences. The number of
replications is chosen so that reported Monte Carlo standard errors fall below
plotting resolution; error bands, where shown, are nonparametric across
replications.

\subsection{Estimands}

\begin{itemize}[leftmargin=1.5em,itemsep=3pt]
\item \emph{Implementation probability.} $\Prob(G_\lambda>0)$ with
$G_\lambda=\sum_iS_i$, and its standardized form $\Phi(-u)$, where
\[
u=-\frac{\sqrt n\,\mu_S(\lambda)}{\sigma_S(\lamcrit)},
\qquad
\mu_S(\lambda)=\E[X]+\lambda\xlo,
\qquad
\sigma_S^2(\lamcrit)=\Var\!\bigl(S_i(Y_i,\lamcrit)\bigr).
\]
\item \emph{Maximum reduced-form revenue.} $R_n^{*,\mathrm{red}}=\E[(V_n)_+]$ with
$V_n=\sum_i\widehat J_i$, and its normalization $R_n^{*,\mathrm{red}}/\sqrt n$ at
$\xlo=0$.
\item \emph{Large-deviation rate.} $(1/n)\log\Prob(V_n\ge0)$ and
$(1/n)\log R_n^{*,\mathrm{red}}$, compared with $-I_K(0)$ when $\xlo<0$, with
$I_K$ as in Definition~\ref{supp:def_cramer_rate}.
\item \emph{Welfare gaps.} The equal-scale gap
$(W_{\mathrm{Lap}}-W_{\mathrm{Log}})/W_{\mathrm{FB}}$ and the common-$\mu$ endpoint
gap $(W_{\mathrm{Log}}-W_{\mathrm{Lap}})/W_{\mathrm{FB}}$, with $W_{\mathrm{FB}}$
the full-information welfare benchmark.
\item \emph{Endpoint trade-off gap.}
$\mathcal T_{\mathrm{Lap}}(\alpha)-\mathcal T_{\mathrm{Log}}(\alpha)$ on the binary
endpoint experiment $\{\xlo,\xhi\}$.
\end{itemize}

\subsection{Transfer inversion and the reported residual}

The implementing transfer solves the Fredholm equation $\mathcal K g_i=T_i^{q}$ of
Section~\ref{supp:implementation}. We discretize $\mathcal K$ on the signal and
type grids and solve the Tikhonov problem \eqref{supp:eq_tikhonov}, selecting the
regularization parameter $\alpha$ by a standard stability criterion. Consistent
with Remark~\ref{supp:rmk_norm_choice} and
Proposition~\ref{supp:prop_regularized_interpretation}, the reported diagnostic is
the uniform residual
\[
\eta=\sup_{x\in\X}\bigl|(\mathcal K g_\alpha)(x)-T_i^{q}(x)\bigr|
\]
on a dense grid, which bounds the incentive, participation, and revenue errors by
$2\eta$, $\eta$, and $\eta$ respectively. The regularized solution is therefore
presented as an approximate implementation with an explicit error guarantee, not
as a proof of exact range membership.

\subsection{Reproducibility}

The computations are carried out in \textsf{R}. Each experiment is generated by a
self-contained script; together they produce the figures of the main paper and its
numerical appendix. Random seeds are fixed within each script, so the reported
figures are exactly reproducible.

\section{Further Remarks on Channel Comparisons}
\label{supp:channel_ordering}

\subsection{Variance is not a general informativeness order}

A smaller additive-noise variance does not by itself imply Blackwell dominance,
which requires an explicit type-independent garbling. Two noise laws can have
ordered variances yet be Blackwell-incomparable. The numerical welfare ordering
among Gaussian, Laplace, and logistic channels should therefore not be presented
as a general variance theorem.

\subsection{Laplace versus logistic}

At equal scale, $K_{\mathrm{Lap},\beta}\succeq_{\mathrm B}K_{\mathrm{Log},\beta}$.
This is a genuine general result. Euler's product
$\sinh(\pi z)/(\pi z)=\prod_{k\ge1}(1+z^2/k^2)$ gives, for the centered logistic
and Laplace characteristic functions,
\[
\phi_{\mathrm{Log},\beta}(t)
=
\prod_{k=1}^{\infty}\frac{1}{1+\beta^2t^2/k^2},
\qquad
\phi_{\mathrm{Lap},\beta}(t)
=
\frac{1}{1+\beta^2t^2},
\]
so that
\[
\phi_{\mathrm{Log},\beta}(t)
=
\phi_{\mathrm{Lap},\beta}(t)\,
\prod_{k=2}^{\infty}\frac{1}{1+\beta^2t^2/k^2}.
\]
The remaining product is the characteristic function of an independent sum
$\sum_{k\ge2}L_k$ with $L_k\sim\Lap(0,\beta/k)$, which converges in $L^2$ since
$\sum_{k\ge2}\Var(L_k)=2\beta^2\sum_{k\ge2}k^{-2}<\infty$. Hence logistic noise
equals Laplace noise plus independent noise, the logistic experiment is a garbling
of the Laplace experiment, and the ordering holds. Because both channels satisfy
pure $\epsilon$-LDP at $\beta=\Delta_{\X}/\epsilon$, the ordering applies directly
at a common pure-LDP budget.

\section{Scope and Limitations}
\label{supp:limitations}

The main paper establishes a sharp reduced-form allocation characterization. Its
strongest unconditional implementation statements are:

\begin{enumerate}[label=(\roman*),leftmargin=2em,itemsep=3pt]
\item monotone posterior-score allocations are exactly BIC at the reduced-form
level;
\item the normalized interim transfer is given exactly by the envelope formula;
\item every exact signal-measurable implementation must solve the Fredholm
equation;
\item completeness identifies the opponent-averaged transfer when a solution
exists;
\item exact signal-measurable implementation is equivalent to a range condition;
\item polynomial targets under additive channels form a nontrivial class of exact
implementations; and
\item an approximation residual $\eta$ yields explicit $2\eta$-BIC and $\eta$-IR
guarantees.
\end{enumerate}

The paper does not claim that every optimal reduced-form transfer belongs to the
range of a Gaussian, Laplace, or logistic channel operator, nor does it infer
exact implementation from a finite-dimensional regularized solution. The
hierarchical extension is conditional on a high-level coordinatewise-monotonicity
requirement, and the unknown-endpoint simulation lies outside the common-support
theorem and is presented as a numerical extension. The revenue taxonomy concerns
maximum reduced-form revenue and is an upper bound on exactly signal-measurable
revenue unless range membership is verified; approximate implementation transfers
the reduced-form conclusion to approximately BIC mechanisms with an explicit error
bound. Finally, the only general cross-channel welfare ordering proved among the
three reported families is the equal-scale Laplace-over-logistic Blackwell
ordering; other numerical rankings are specification dependent.
  \renewcommand{\bibsection}{\section*{References (Supplement)}}
  \putbib[reference]
\end{bibunit}
\endgroup
\end{document}